\documentclass[10pt, a4paper, oneside, reqno]{amsart}

\input{style}
\makeatletter
\def\munderbar#1{\underline{\sbox\tw@{$#1$}\dp\tw@\z@\box\tw@}}
\makeatother

\newtheorem{theorem}{Theorem}[section]
\newtheorem{definition}[theorem]{Definition}
\newtheorem{lemma}[theorem]{Lemma}
\newtheorem{remark}[theorem]{Remark}
\newtheorem{corollary}[theorem]{Corollary}
\newtheorem{proposition}[theorem]{Proposition}

\newtheorem{assumption}[theorem]{Assumption}

\newcommand{\muxo}{\mu_{0}}

\newcommand{\be}{\begin{equation}}
\newcommand{\ee}{\end{equation}}
\newcommand{\bea}{\begin{equation*}\begin{aligned}}
\newcommand{\eea}{\end{aligned}\end{equation*}}
\newcommand{\ds}{\displaystyle}

\newcommand{\R}{\mathbb{R}}

\newcommand{\Min}{\min\limits_}
\newcommand{\Sup}{\sup\limits_}
\newcommand{\Inf}{\inf\limits_}

\newcommand{\wh}{\widehat}
\newcommand{\mc}{\mathcal}
\newcommand{\mbb}{\mathbb}

\newcommand{\PP}{\mbb P}
\newcommand{\Pnom}{\wh{\mbb P}}
\newcommand{\QQ}{\mbb Q}
\newcommand{\DD}{\mbb D}
\newcommand{\dd}{\mathrm{d}}

\DeclareMathOperator{\st}{s.t.}

\newcommand{\PSD}{\mathbb{S}_{+}} % the set of positive semi-definite matrices of dimension p
\newcommand{\Let}{\triangleq}
\newcommand{\opt}{^\star}
\newcommand{\eps}{\varepsilon}

\newcommand{\BB}{\mbb B}

\newcommand{\Wass}{\mathds{W}}

\newcommand{\EE}{\mathds{E}}

\newcommand{\half}{\frac{1}{2}}

\usepackage{setspace}
\usepackage{threeparttable}
\newcommand{\modified}[1]{{\color{black} #1}}
\title[Robustifying Conditional Portfolio Decisions via Optimal Transport]{Robustifying Conditional Portfolio Decisions \\
via Optimal Transport}
\begin{document}

\doublespacing
\date{\today}
\begin{abstract}
    We propose a data-driven portfolio selection model that integrates side information, conditional estimation, and robustness using the framework of distributionally robust optimization. Conditioning on the observed side information, the portfolio manager solves an allocation problem that minimizes the worst-case conditional risk-return trade-off, subject to all possible perturbations of the covariate-return probability distribution in an optimal transport ambiguity set. Despite the nonlinearity of the objective function in the probability measure, we show that the distributionally robust portfolio allocation with a side information problem can be reformulated as a finite-dimensional optimization problem. If portfolio decisions are made based on either the mean-variance or the mean-Conditional Value-at-Risk criterion, the reformulation can be further simplified to second-order or semi-definite cone programs. Empirical studies in the US equity market demonstrate the advantage of our integrative framework against other benchmarks.
\end{abstract}

\author{Viet Anh Nguyen, Fan Zhang, Shanshan Wang, Jos\'e Blanchet, Erick Delage, Yinyu Ye}
\thanks{The authors are with the Department of Systems Engineering and Engineering Management, the Chinese University of Hong Kong (\texttt{nguyen@se.cuhk.edu.hk}), the Department of Management Science and Engineering, Stanford University, United States (\texttt{fzh, jose.blanchet, yinyu-ye@stanford.edu}), and the GERAD and Department of Decision Sciences, HEC Montr\'{e}al, Canada ({\tt shshwangbit@gmail.com, erick.delage@hec.ca}). Correspondence should be sent to \tt shshwangbit@gmail.com}

\maketitle
%%%%%%%%%%%%%%%%%%%%%%%%%%%%%%%%%%%%%%%%%%%%%%%%%%%%%%%%%%%%%%%%%%%%%%%%%%%%%%%%%%%%%%%%%%%%%%%%%%%%%%%

\section{Introduction}
 
We study distributionally robust portfolio decision rules, which are informed by conditioning on observed side information.
In the presence of contextual information that might be relevant for predicting future returns,  our objective is to direct the statistical task (i.e., conditional estimation) towards the downstream task of selecting a portfolio that is robust to distributional shifts (including local conditioning). Our ultimate goal is to provide a tractable, non-parametric, data-driven approach that bypasses the need for computing global conditional expectations. At the same time, our framework is flexible to accommodate a wide range of portfolio optimization selection criteria, as well as extend to other diverse applications. In order to explain our methodology and contribution in detail, let us consider a conventional single-period portfolio optimization selection problem of the form 
\[
    \Min{\alpha\in\mc A}~\mc R_{\PP}[Y^{\top}\alpha] - \eta\cdot\EE_{\PP}[Y^{\top}\alpha],
\]
where the real-valued vector $\alpha$ denotes the portfolio allocation choice within a feasible region $\mc A$, and the random vector $Y$ denotes the assets' future return under a probability measure $\PP$. The return of the portfolio is denoted by $Y^\top \alpha$. At the same time, $\mc R_{\PP}[\cdot]$ captures the risk associated with $\alpha$, and $\eta \geq 0$ is a parameter that weights the preference between the portfolio return and the associated risk. 

Different choices of the measure of risk~$\mc R_{\PP}[\cdot]$ lead to different portfolio optimization models, including the mean-variance model~\cite{ref:markowitz1952portfolio}, the mean-standard deviation model~\cite{ref:konno1991mean}, the mean-Value-at-Risk model~\cite{ref:benati2007mixed}, and the mean-Conditional Value-at-Risk (CVaR) model~\cite{ref:rockafellar2000optimization}. The measure of risk~$\mc R_{\PP}[\cdot]$ often is chosen to belong to certain classes to induce desirable properties. For instance, the class of coherent risk measures \cite{ref:artzner1999coherent} or the class of convex risk measures \cite{ref:follmer2002convex} are examples of risk measure classes that are often used in portfolio optimization \cite{ref:luthi2005convex}. In our setting, we consider a family of risk measures that admits a stochastic optimization representation of the form
$\mc R_{\PP}[Y^{\top}\alpha] = \min_{\beta\in\mc B} \EE_{\PP}[r(Y^\top \alpha,\beta)]$ for some auxiliary function $r$ and finite-dimensional statistical variable $\beta$ residing within an appropriate feasible region $\mc B$. Hence, we focus on a generic portfolio allocation that minimizes an expected loss of the form
\begin{equation}
    \label{eq:stochastic-opt}
    \Min{\alpha \in \mc A,~\beta \in \mc B}~\EE_{\PP}[ \ell(Y, \alpha, \beta)],\qquad 
    \ell(Y, \alpha, \beta)\Let 
    r(Y^\top \alpha,\beta) - \eta\cdot Y^\top \alpha.
\end{equation}
In problem~\eqref{eq:stochastic-opt}, optimizing over $\beta$ characterizes the improvement of risk estimation, which is entangled with the procedure of finding the optimal portfolio allocation $\alpha$. For example, consider when $\mc R_{\PP}$ is the variance of the portfolio return, then by setting $\mc B = \R$ and choosing $\ell$ as a quadratic function, we can rewrite the variance in the form $\mc R_{\PP}[Y^\top \alpha] = \min_{\beta \in \R}~\EE_{\PP}[ (Y^\top \alpha - \beta)^2]$. This further implies that the optimal solution of $\beta$ coincides with the expected return $\EE_{\PP}[Y^\top \alpha]$, and one thus can view $\beta$ as an auxiliary statistical variable representing the estimation of the expected portfolio return.

Despite the simplicity and popularity of the portfolio optimization models~\eqref{eq:stochastic-opt}, they are challenging to solve because the distribution $\PP$ of the random stock returns is unknown to the portfolio managers. To overcome this issue, problem~\eqref{eq:stochastic-opt} is usually solved by resorting to a plug-in estimator $\Pnom$ of $\PP$, and this estimator is usually inferred from the available data. Unfortunately, solving~\eqref{eq:stochastic-opt} using an estimated distribution $\Pnom$ may amplify the statistical estimation error. The corresponding portfolio allocation is vulnerable to the estimation errors of the input: \emph{even small changes in the input parameters can result in large changes in the weights of the optimal portfolio} \cite{ref:chopra1993effect,ref:britten1999sampling}. Consequently, in terms of out-of-sample performance, the advantage of deploying the optimal portfolio weights is overwhelmed by the offset of estimation error \cite{ref:demiguel2009optimal}. To mitigate the estimation error, one may resort, at least, to the following two strategies: 
\begin{enumerate}[leftmargin=*]
    \item reduce the estimation bias by incorporating side information into the portfolio allocation problems;
    \item diminish the impact of estimation variance by employing robust portfolio optimization models.
\end{enumerate}
Both the above-mentioned strategies have led to successful stories. Starting from the capital asset-pricing model \cite{ref:sharpe1964capital, ref:black1972capital}, numerous studies exploit side information to explain and/or predict the cross-sectional variation of the stock returns. The predictive side information may include macro-economic factors \cite{ref:flannery2002macroeconomic}, firms' financial statements \cite{ref:fama1992cross}, and historical trading data \cite{ref:grinblatt1995momentum}. Alternatively, robust optimization formulations are applied to the portfolio allocation tasks to alleviate the impact of statistical error due to input uncertainty~\cite{ref:pflug2007ambiguity, ref:delage2010distributionally, ref:zymler2013distributionally}.

In this paper, we endeavor to unify both strategies of leveraging side information and enhancing robustness. In order to incorporate the side information into portfolio allocation, we use a random vector $X$ to denote the side information that is correlated with the stock return $Y$, and consider the following conditional stochastic optimization (CSO) problem:
\begin{equation}
    \label{eq:stochastic-opt-side-info}
    \Min{\alpha \in \mc A,~ \beta \in \mc B}~\EE_{\PP}[ \ell(Y, \alpha, \beta)|X\in \mc N],
\end{equation}
where the set $\mc N$ represents our prior (or most current) knowledge about the covariate $X$.
In problem~\eqref{eq:stochastic-opt-side-info}, the distribution $\PP$ is lifted to a \textit{joint} distribution between the covariate $X$ and the return $Y$. Problem~\eqref{eq:stochastic-opt-side-info} is thus a joint portfolio allocation (over $\alpha$) and statistical estimation (over $\beta$) problem, that minimizes the conditional expected loss of the portfolio, given that the outcome of $X$ belongs to a set $\mc N$. We will apply the distributionally robust framework to promote robust solutions and hedge against the estimation error of $\PP$. Instead of assessing the expected loss with respect to a single distribution, the distributionally robust formulation minimizes the expected conditional loss uniformly over a distributional ambiguity set that represents the portfolio managers' ambiguity regarding the underlying distribution $\PP$. 

We propose a data-driven portfolio optimization framework that \emph{simultaneously exploits side information and robustifies over sampling error} and \emph{model misspecification}. We formulate the distributionally robust portfolio optimization problem with side information as
\begin{equation}
    \label{eq:DR-side-info}
    \Min{\alpha \in \mc A,~\beta \in \mc B}~
    \Sup{\QQ\in\BB, \QQ(X \in \mc N) \ge \eps}~
    \EE_{\QQ}[\ell(Y, \alpha, \beta)|X\in \mc N].
\end{equation}
We can view the min-max problem~\eqref{eq:DR-side-info} as a two-player zero-sum game: the portfolio manager chooses the portfolio allocation~$\alpha$ and the estimation parameters~$\beta$ so as to minimize the conditional expected loss, while the fictitious adversary chooses the joint distribution $\QQ$ of $(X, Y)$ so as to maximize the incurring loss. A typical choice of the ambiguity set $\mbb B$ is a neighborhood around a nominal distribution $\Pnom$. In a data-driven setting, a popular choice of $\Pnom$ is the empirical distribution, defined as a uniform distribution supported on the available data. As such, the data influences the optimization problem~\eqref{eq:DR-side-info} via the channel of the ambiguity set prescription. Given \modified{$\eps > 0$} as an input parameter, the constraint $\QQ(X \in \mc N) \ge \eps$ is a shorthand for $\QQ((X, Y) \in \mc N \times \mc Y) \ge \eps$, and it controls the probability mass requirement on the set $\mc N \times \mc Y$ and avoids conditioning on sets of measure zero. In the limit, we also will study the case $\QQ(X \in \mc N) > 0$ with strict inequality.

Intuitively, one can describe problem~\eqref{eq:DR-side-info} as a zero-sum game between two players. The outer player is the statistical portfolio manager who optimizes over the allocation $\alpha$, and the estimator $\beta$ to minimize the conditional expected loss, and the inner player can be regarded as a fictitious adversary whose goal is to increase the resulting loss by choosing an adversarial distribution.

Our formulation~\eqref{eq:DR-side-info} has several benefits. First, our model is an end-to-end framework that directly generates investment decisions using data. This approach, in the same spirit as~\cite{ref:kozak2020shrinking} and~\cite{ref:demiguel2020transaction}, integrates the sequential pipeline of return prediction and portfolio optimization, which is conventionally employed by portfolio managers. Therefore, our model, which consolidates return prediction and portfolio optimization, directly minimizes the conditional expected loss instead of the prediction error of return. Second, the robustness of the random return is directed by the portfolio choice. Indeed, the adversary who solves the supremum problem in~\eqref{eq:DR-side-info} will aim to maximize the portfolio loss and not the statistical loss. Thirdly, our proposed model does not require imposing any parametric family for prediction, and thus it reduces the risk of model misspecification (see \cite{10.2307/3648200,Fabozzi40} for examples of how this risk is typically mitigated).

\subsection{Main Contributions}
We study a robust portfolio optimization strategy that integrates conditional estimation and decision-making.  This strategy leverages side information to improve the ex-ante return prediction. Additionally, it robustifies the over-estimation error to improve the out-of-sample performance of the portfolio allocation. Our contributions are summarized as follows.
\begin{enumerate}[leftmargin=*]
    \item We present a comprehensive modeling framework for distributionally robust optimization for conditional decision-making, which provides flexibility for selecting different conditional sets and adjusting probability mass requirements. We study specifically the case where the set $\mbb B$ will be prescribed using the notions of optimal transport. We show that the qualitative behavior of the distributionally robust portfolio optimization problem with side information depends on whether the set $\mc N$ is a singleton and whether the probability mass requirement $\eps$ is zero.
    
    \item We derive finite-dimensional reformulations\footnote{To be precise, finite-dimensional optimization problems refer to problems with a finite number of variables and constraints.} of the distributionally robust portfolio optimization problem with side information. Further, we also provide tractable reformulations for the distributionally robust mean-variance problem and mean-CVaR problem as convex conic (second-order cone or semi-definite cone) problems. These 
   tractable reformulations are particularly interesting given that the conditional expectation is a nonlinear function of the probability measure (thus not amenable to traditional reformulations based on semi-infinite linear programming duality arguments) and that minimizing the worst-case conditional expected loss is equivalent to a min-max optimization problem with a fractional objective function. 
    
    \item We conduct extensive numerical experiments in the US equity markets. The results demonstrate the advantage of the distributionally robust portfolio optimization formulation, which has \modified{a lower out-of-sample expected loss} compared to other benchmarks.
\end{enumerate}

The distributionally robust optimization problem with side information presented in~\eqref{eq:DR-side-info} is generic, and our contributions are related to the vast literature in estimation-informed decision-making tasks. Moreover, by an appropriate choice of the loss function $\ell$, the sets $(\mc A, \mc B, \mc N)$ and the ambiguity set $\mbb B$, the formulation~\eqref{eq:DR-side-info} can be applied to many other problems in the field of operations research and management science such as supply chain management, transportation, and energy planning. In this regard, we refer interested readers to Appendix \ref{sec:contextTSLP}, where we extend our mean-CVaR results into \modified{a conservative approximation for} the contextual two-stage stochastic linear programming.

\subsection{Related Literature}

We start this section by discussing existing work that exploits side information to enhance portfolio optimization. For applications of fundamental analysis, Fama and French leveraged market capitalization and price-earning ratio to construct the celebrated factor models
\cite{ref:fama1992cross,ref:fama1996multifactor}, which was later applied to portfolio optimization~\cite{ref:brandt2009parametric}; P\'astor and Stambaugh \cite{ref:pastor2003liquidity} reported the evidence of equity risk premium due to illiquidity. For applications of technical indicators, examples include momentum factor \cite{ref:grinblatt1995momentum, ref:neely2014forecasting}, implied volatility \cite{ref:demiguel2013improving}, short interest \cite{ref:rapach2016short}, and investor sentiment \cite{ref:baker2006investor}. However, compared with our proposed approach, these methods are vulnerable to estimation errors.

There is a rich literature on distributionally robust portfolio optimization, with diverse choices of distributional ambiguity sets or different objective functions. When the objective function involves Value-at-Risk or CVaR, previous studies employed the moment-based ambiguity set \cite{ref:ghaoui2003worst, ref:zhu2009worst, ref:goh2011robust} and Wasserstein ambiguity set \cite{ref:esfahani2018data} to model the worst-case risk. For distributionally robust portfolio optimizations that maximize piecewise affine loss functions, see \cite{ref:natarajan2010tractable} for moment-based ambiguity set, and \cite{ref:chen2020robust} for event-wise ambiguity set. For the distributionally robust mean-variance problem, Lobo and Boyd \cite{ref:lobo2000worst} and Delage and Ye \cite{ref:delage2010distributionally} are among the earliest to consider the worst-case mean-variance model under the moment ambiguity set. Other types of ambiguity sets, such as Wasserstein ambiguity set \cite{ref:pflug2007ambiguity, ref:blanchet2018distributionally}, and optimal transport ambiguity set \cite{ref:blanchet2018optimal}, are also studied afterward. Due to the resemblance between Wasserstein distributionally robustification and norm regularization \cite{ref:blanchet2016robust}, the weight-constrained portfolio optimization problem \cite{ref:demiguel2009generalized} is also closely related to this stream of research aiming to boost the robustness of the optimal portfolio. Nevertheless, none of the aforementioned works considers how to incorporate side information and conditional estimation.

 The conditional stochastic optimization (CSO) problem~\eqref{eq:stochastic-opt-side-info} is also related to decision-making with side information, which has received significant attraction recently. Ban and Rudin~\cite{ref:ban2019big} applied weighted sample average approximation (SAA) with weights learned from kernel regression to learn the optimal decision in inventory management with covariate information. Bertsimas and Kallus \cite{ref:bertsimas2020predictive} applied weighted SAA to solve the general CSO problem, in which the weight is learned from nonparametric estimators such as k-nearest neighbors, kernel regressions, and random forests. Bertsimas and McCord \cite{ref:bertsimas2019predictions} generalized the approach of \cite{ref:bertsimas2020predictive} to multistage decision-making problems. Kannan et al.~\cite{ref:kannan2020data} developed a residual-based SAA method, where average residuals during training of the learning model are added to a point prediction of response, which is used as a response sample within the SAA. An extension to heteroscedastic residuals is provided in~\cite{ref:kannan2021hetero}. Decision tree-based approaches are also developed to solve CSO. 
Athey et al.~\cite{ref:athey2019generalized} proposed the forest-based estimates for local parameter estimation, generalizing the original random forest algorithm with observations of side information; Kallus and Mao~\cite{ref:kallus2020stochastic} generalized the idea to solve stochastic optimization problems with side information, in which they seek to minimize the expected loss instead of the prediction accuracy. Conditional chance-constrained programming was also studied recently~\cite{ref:rahimian2019contextual}.

Recently, distributionally robust optimization (DRO) formulations have been integrated into the CSO problem to improve the out-of-sample performance of the solutions~\cite{ref:bertsimas2019dynamic,ref:bertsimas2021bootstrap, ref:srivastava2021data}. For example, Kannan et al.~\cite{ref:kannan2020residuals} constructed a Wasserstein-based uncertainty set over the empirical distribution of the residuals, robustifying their previous work \cite{ref:kannan2020data}; Esteban-P\'erez and Morales \cite{ref:esteban2020distributionally} leveraged some probability trimming methods and a partial mass transportation problem to model the distributionally robust CSO problem.

Despite some common characteristics in terms of methodology, such as the idea of using the theory of optimal transport to robustify the CSO problem, our paper is an independent contribution in comparison to \cite{ref:kannan2020residuals} and \cite{ref:esteban2020distributionally}.\footnote{This claim is supported by the fact that this paper extends our previous work~\cite{ref:nguyen2020distributionally}, which predates the publication dates of both \cite{ref:kannan2020residuals} and \cite{ref:esteban2020distributionally}.} There are two key differences between our model and that considered in~\cite{ref:kannan2020residuals}. Firstly, we integrate prediction and decision-making into a single optimization problem, while in \cite{ref:kannan2020residuals}, the prediction and the decision-making procedures are separated to simplify the problem. Secondly, the conditional side information is required to be a singleton set $X = x_0$ in \cite{ref:kannan2020residuals}, but our approach allows for more flexible modeling of side information by conditioning on a more general event $X \in \mc N$. Our work is also distinct from Esteban-P\'erez and Morales \cite{ref:esteban2020distributionally}. In contrast to \cite{ref:esteban2020distributionally}, which employs probability trimmings to relax problem~\eqref{eq:DR-side-info}, we seek to provide an exact and tractable reformulation of the problem. Moreover, our approach tackles the optimal transport ambiguity set directly, and the results can be explained by picturing an adversary who intuitively perturbs the sample points. 

Finally, this paper is a complete and comprehensive extension to our previous work \cite{ref:nguyen2020distributionally}, with a twofold improvement. On the one hand, this paper naturally unifies prediction and decision-making into a single DRO problem, while \cite{ref:nguyen2020distributionally} solely considers the conditional estimation problem. On the other hand, this paper applies optimal transport distance to construct a more versatile class of distributional ambiguity set, whereas the ambiguity set in \cite{ref:nguyen2020distributionally} is based on the type-$\infty$ Wasserstein distance. Even though the type-$\infty$ Wasserstein distance is also motivated by the theory of optimal transport, it behaves differently from a type-$p$ ($p<\infty$) Wasserstein distance. The pathway to resolve the distributionally robust optimization problems in this paper hence differs significantly from the techniques employed in~\cite{ref:nguyen2020distributionally}. Finally, this paper also enriches the emerging field of Wasserstein DRO, which has gained momentum recently thanks to its applicability in a wide spectrum of practical problems~\cite{ref:blanchet2016robust, ref:gao2016distributionally, ref:zhao2018data, ref:duchi2018learning, ref:esfahani2018data, ref:kuhn2019wasserstein, ref:dou2019distributionally}.

\textbf{Notations.}  For any integer $M\geq 1$, we denote by $[M]$ the set of integers $\{1, \ldots, M\}$. For any $x\in\R$, we denote by $(x)^{+} = \max\{x,0\}$. We write $\R_{+}\Let [0,\infty)$ and $\R_{++} \Let (0, \infty)$.
Let $\Vert\cdot\Vert_p$ denote the $l_p$-norm for $p\in [1,\infty]$. For any space $\mc S$ equipped with a Borel sigma-algebra $\mc F_{\mc S}$, $\mc M(\mc S)$ is the space of all probability measures defined on $(\mc S,\mc F_{\mc S})$. For any $s\in\mc S$, the Dirac's delta measure corresponding to $s$ is denoted by $\delta_s$, i.e., for all $E\in\mc F_{\mc S}$, $\delta_s(E) = 1$ if $s\in E$ and $\delta_s(E) = 0$ if $s\notin E$. For any subset $E\subseteq \mc S$, let $\partial E$ denote the boundary of $E$, which is the closure of $E$ minus the interior of $E$. The cone of $m\times m$ real positive semi-definite matrices is denoted by $\PSD^m$. For $m\times m$ real symmetric matrices $A$ and $B$, we write $A\succeq B$ if and only if $A-B \in \PSD^m$. For a probability measure $\PP$, let $\EE_{\PP}[\cdot]$ denote the expectation under measure $\PP$. We write $\mathbbm{1}_{\mathcal S}(\cdot)$ as the indicator function of the set $\mc S$, i.e., $\mathbbm{1}_{\mathcal S}(s) = 1$ if $s \in \mc S$ and $0$, otherwise. Throughout, $\mc X$ denotes the covariate space, and $\mc Y$ denotes the asset return space. The available data consist of $N$ samples, and each sample is a pair of covariate-return denoted by $(\wh x_i, \wh y_i) \in \mc X \times \mc Y$.

All proofs are relegated to the appendix.

%%%%%%%%%%%%%%%%%%%%%%%

\section{Problem Setup}
\label{sec:setup}

This section delineates the details of our robustification of the conditional portfolio allocation problem. To this end, we will construct the ambiguity set using the optimal transport cost.
\begin{definition}[Optimal transport cost] \label{def:OT}
Fix any integer $k$ and let $\Xi \subseteq \R^k$.
Let $\DD$ be a nonnegative and continuous function on $\Xi \times \Xi$. The optimal transport cost between two distributions $\QQ_1$ and $\QQ_2$ supported on $\Xi$ is defined as
\[
    \Wass(\QQ_1, \QQ_2) \Let \min \left\{ \EE_\pi [\DD(\xi_1, \xi_2)] :
\pi \in \Pi(\QQ_1, \QQ_2)
\right\},
\]
where $\Pi(\QQ_1, \QQ_2)$ is the set of all probability measures on $\Xi \times \Xi$ with marginals $\QQ_1$ and $\QQ_2$, respectively.
\end{definition}

Intuitively, the optimal transport cost $\Wass(\QQ_1, \QQ_2)$ computes the cheapest cost of transporting the mass from an initial distribution $\QQ_1$ to a target distribution $\QQ_2$, given a function $\DD(\xi_1,\xi_2)$ prescribing the cost of transporting a unit of mass from position $\xi_1$ to position $\xi_2$. The function $\mathbb{D}$ is called the ground transport cost function. The existence of an optimal joint distribution $\pi\opt$ that attains the minimal value in the definition is guaranteed by~\cite[Theorem 4.1]{ref:villani2008optimal}. The class of optimal transport cost is rich enough to encompass common probability metrics. In particular, if $\DD$ is a metric on $\Xi$, then $\Wass$ coincides with the type-1 Wasserstein distance. In this paper, we do not restrict the ground transportation cost $\DD$ to be a metric, and as a consequence, $\Wass$ is not necessarily a proper distance on the space of probability measure. Nevertheless, we will usually refer to $\Wass$ as an optimal transport distance with a slight abuse of terminology. A comprehensive introduction to the theory of optimal transport can be found in~\cite{ref:villani2008optimal}.

We now consider the covariate space $\mc X$, equipped with a ground cost $\DD_{\mc X}$, and the return space $\mc Y$, equipped with a ground cost $\DD_{\mc Y}$. The joint sample space $\mc X \times \mc Y$ is endowed with a ground cost $\DD$, which is additively separable using $\DD_{\mc X}$ and $\DD_{\mc Y}$. We suppose that the portfolio manager possesses $N$ historical data samples, where each sample is a pair of covariate-return $(\wh x_i, \wh y_i) \in \mc X \times \mc Y$. Using the optimal transport cost in Definition~\ref{def:OT}, we define the \textit{joint ambiguity set} of the distributions of $(X, Y)$ as
\[
	\mbb B_\rho = \left\{ \QQ \in \mc M(\mc X \times \mc Y): \Wass(\QQ, \Pnom) \leq \rho  \right\}.
\]
The set $\mbb B_\rho$ contains all distributions that are of optimal transport distance less than or equal to $\rho$ from the empirical distribution 
\[
	\Pnom \Let \frac{1}{N} \sum_{i \in [N]} \delta_{(\wh x_i, \wh y_i)},
\]
defined as the distribution that makes every member of the dataset $(\wh x_i, \wh y_i)_{i=1}^N$ equiprobable.\footnote{While other discrete distributions can be used as the nominal one, the empirical distribution is often considered in cases where observations are assumed to be independently and identically distributed.}
In this form, $\mbb B_\rho$ is a non-parametric ambiguity set.  
We propose to solve the distributionally robust conditional portfolio allocation problem
\be \label{eq:dro}
    \Min{\alpha \in \mc A,~\beta \in \mc B}~\Sup{\QQ \in \mbb B_\rho, \QQ(X \in \mc N_\gamma(x_0)) \ge \eps}~ \EE_{\QQ}[ \ell(Y, \alpha, \beta) | X \in \mc N_\gamma(x_0) ].
\ee
Problem~\eqref{eq:dro} is a special instance of problem~\eqref{eq:DR-side-info} with two particular design choices. First, the ambiguity set is chosen as the optimal transport ambiguity set $\mbb B_\rho$. Second,
the side information input is specifically modeled as a neighborhood around a covariate $x_0 \in \mc X$, and the size of this neighborhood is controlled by a parameter $\gamma \in \R_+$. This neighborhood is prescribed using the ground cost $\DD_{\mc X}$ on $\mc X$ as
\[
    \mc N_{\gamma}(x_0) \Let \left\{ x \in \mc X: \DD_{\mc X}(x, x_0) \le \gamma\right\}.
\]
The set $\mc N_{\gamma}(x_0) \times \mc Y$ is called the fiber set. When $\mc N_{\gamma}(x_0)$ collapses into a singleton, we obtain a \textit{singular} fiber set $\{x_0\} \times \mc Y$, and this singular case represents conditioning on the event $X = x_0$. In this perspective, our formulation is a generalization of the conventional conditional decision-making problem in the literature, which usually focuses on conditioning on $X = x_0$. Throughout this paper, we use $\QQ(X \in \mc N_\gamma(x_0))$ as a shorthand for $\QQ((X, Y) \in \mc N_\gamma(x_0) \times \mc Y)$. The constraint $\QQ(X \in \mc N_\gamma(x_0)) \ge \eps$ indicates a minimum amount $\eps \in [0, 1]$ of probability mass to be assigned to the fiber set $\mc N_\gamma(x_0) \times \mc Y$. The feasible set $\mc A$ and $\mc B$ are assumed to be simple in the sense that they are representable using second-order cone constraints. Overall, problem~\eqref{eq:dro} involves three parameters: the ambiguity size $\rho \in \R_+$, the fiber size $\gamma \in \R_+$, and the probability mass requirement $\eps \in [0, 1]$.

It is now instructive to discuss the difficulty level of problem~\eqref{eq:dro}. By the definition of the conditional expectation, problem~\eqref{eq:dro} is equivalent to
\[
    \Min{\alpha \in \mc A,~\beta \in \mc B}~\Sup{\QQ \in \mbb B_\rho, \QQ(X \in \mc N_\gamma(x_0)) \ge \eps}~ \frac{\EE_{\QQ}[ \ell(Y, \alpha, \beta) \mathbbm{1}_{\mc N_\gamma(x_0)} ( X) ]}{\EE_{\QQ}[\mathbbm{1}_{\mc N_\gamma(x_0)}( X) ]}.
\]
In this form, one can observe that the objective function is a fractional, and thus is a non-linear, function of the probability measure $\QQ$.  Problem~\eqref{eq:dro} hence belongs to the class of distributionally robust fractional optimization problems~\cite{ref:ji2020data, ref:zhao2017distributionally}. Thus, problem~\eqref{eq:dro} is fundamentally different and also fundamentally more difficult to solve compared to existing models in the field of distributionally robust optimization that simply minimizes the worst-case expected loss~\cite{ref:rahimian2019distributionally}.

It is important to emphasize that the auxiliary variable $\beta \in \mc B$ is intentionally regrouped to the feasible set of the outer minimization problem of formulation~\eqref{eq:dro}. This regrouping of portfolio allocation variables~$\alpha \in \mc A$ and the estimation variable~$\beta \in \mc B$ is a modeling choice. In the expected loss minimization perspective, this regrouping is without any loss of optimality under conditions such as the compactness of $\mc A$ and $\mc Y$, as we illustrate in the following technical lemmas.

\begin{lemma}[Mean-variance loss function] \label{lemma:MV}
    The robustified conditional mean-variance portfolio allocation problem is
    \[
\Min{\alpha \in \mc A}~\Sup{\QQ \in \mbb B_\rho, \QQ(X \in \mc N_\gamma(x_0)) \ge \eps}~\mathrm{Variance}_\QQ[Y^\top \alpha| X \in \mc N_\gamma(x_0)] - \eta\cdot \EE_\QQ[ Y^\top\alpha| X \in \mc N_\gamma(x_0)].
    \]
    If $\mc A$ and $\mc Y$ are compact, then the above optimization problem is equivalent to~\eqref{eq:dro} with $\ell$ representing the mean-variance loss function of the form
    \be \label{eq:mean-var-loss}
        \ell(y, \alpha, \beta) = (y^\top \alpha - \beta)^2 - \eta \cdot y^\top \alpha
    \ee
    and the set $\mc B$ defined as
    \[\mc B = \big[\inf_{\alpha\in\mc A,~y\in\mc Y}y^\top \alpha,\;\sup_{\alpha\in\mc A,~y\in\mc Y}y^\top \alpha \big].\]
\end{lemma}

\begin{lemma}[Mean-CVaR loss function] \label{lemma:MCVaR}
The robustified conditional mean-CVaR portfolio allocation problem with risk tolerance $\tau \in (0, 1)$ is
\[
    \Min{\alpha \in \mc A}~\Sup{\QQ \in \mbb B_\rho, \QQ(X \in \mc N_\gamma(x_0)) \ge \eps}~\mathrm{CVaR}^{1-\tau}_\QQ[Y^\top \alpha|X \in \mc N_\gamma(x_0)]-\eta\cdot\EE_\QQ[ Y^\top\alpha| X \in \mc N_\gamma(x_0)].
    \]
    If $\mc A$ and $\mc Y$ are compact, the above optimization problem is equivalent to problem~\eqref{eq:dro} with $\ell$ representing the mean-CVaR loss function of the form
\be \label{eq:mean-cvar-loss}
\ell(y, \alpha, \beta)
= -\eta y^\top \alpha + \beta+ \frac{1}{\tau} (-y^\top \alpha - \beta)^+ = \max\left\{ - \eta y^\top \alpha + \beta, - (\eta + \frac{1}{\tau}) y^\top \alpha + (1 - \frac{1}{\tau}) \beta\right\}
\ee and $\mc B$ as defined in Lemma \ref{lemma:MV}.
\end{lemma}

Theoretically, the assumptions on the compactness of the sets $\mc A$ and $\mc Y$ may be restrictive. However, in practice and especially in the portfolio decision setting, the feasible allocation set $\mc A$ is usually compact. Moreover, we empirically observe that restricting $\mc Y$ to be inside a ball of sufficiently large diameter does not alter the numerical solution. Alternatively, weak duality also implies that problem~\eqref{eq:dro} is a conservative formulation of the robustified risk measure minimization problem, see Remark~\ref{remark:weak-duality} for an example. Thus even when the compactness assumptions do not hold, the optimal solution of~\eqref{eq:dro} can still be considered a robust, or more risk-averse, portfolio allocation.

Throughout this paper, we make the following regularity assumption.

\begin{assumption}[Regularity conditions] \label{a}
    The following assumptions hold.
    \begin{enumerate}[label=(\roman*)]
        \item \label{a:cost} Separable ground cost: The joint space $\mc X \times \mc Y$ is endowed with a separable cost function $\DD$ of the form
        \[
            \DD\big( (x, y), (x', y') \big) = \DD_{\mc X}(x, x') + \DD_{\mc Y}(y, y') \quad \forall (x, y),~(x', y') \in \mc X \times \mc Y.
        \]
        The individual ground transport costs $\DD_{\mc X}$ and $\DD_{\mc Y}$ are symmetric, non-negative, and continuous on $\mc X \times \mc X$ and $\mc Y \times \mc Y$, respectively. Moreover, $\DD\big( (x, y), (x', y') \big) = 0$ if and only if $(x, y) = (x', y')$.
        \item \label{a:proj} Projection: For any $i \in [N]$, there exists a unique projection $\wh x_i^p \in \mc N_{\gamma}(x_0)$  of $\wh x_i$ onto $\mc N_{\gamma}(x_0)$ under the cost $\DD_{\mc X}$, that is,
        \be \label{eq:kappa-def}
            0 \le \kappa_i \Let \Min{x \in \mc N_\gamma(x_0)}\DD_{\mc X}(x, \wh x_i) = \DD_{\mc X}(\wh x_i^p, \wh x_i).
        \ee
        \item \label{a:vic} Vicinity: For any $x \in \partial \mc N_{\gamma}(x_0)$ and for any radius $r > 0$, the neighborhood set $\{x'\in \mc X \backslash \mc N_{\gamma}(x_0) : \DD_{\mc X}(x', x) \le r\}$ around the boundary point $x$ is non-empty.
    \end{enumerate}
\end{assumption}

The conditions on $\DD_{\mc X}$ and $\DD_{\mc Y}$ in Assumption~\ref{a}\ref{a:cost} are trivially satisfied if these individual ground costs are chosen as continuous functions of norms on $\mc X$ and $\mc Y$, respectively.  Assumption~\ref{a}\ref{a:proj} asserts the existence of the projection of any training sample point $(\wh x_i, \wh y_i)$ onto the set $\mc N_{\gamma}(x_0) \times \mc Y$. It is easy to see that $\kappa_i = 0$ and $\wh x_i^p = \wh x_i$ whenever $\wh x_i \in \mc N_{\gamma}(x_0)$ thanks to the choice of $\DD$ in Assumption~\ref{a}\ref{a:cost}. Assumption~\ref{a}\ref{a:vic} indicates that any points on the boundary of the fiber set can be shifted outside the fiber with an arbitrarily small cost. Assumption~\ref{a}\ref{a:vic} holds whenever the set $\mc N_\gamma(x_0)$ lies in the interior of the set $\mc X$.

\subsection{Feasibility condition}

Notice that for a fixed amount of fiber probability $\eps$, if the transportation budget $\rho$ is small, the feasible set of the inner supremum problem in~\eqref{eq:dro} may be empty. We define the minimum value of the radius $\rho$ so that this feasible set is non-empty as
\be \label{eq:rho-LB-1}
    \rho_{\min}(x_0, \gamma, \eps) \Let \inf
    \left\{
    \Wass(\QQ, \Pnom): 
    \QQ \in \mc M (\mc X \times \mc Y), ~\QQ(X \in \mc N_\gamma(x_0)) \ge \eps
    \right\}.
\ee

The value $\kappa_i$ defined in Assumption~\ref{a}\ref{a:proj} signifies the unit cost of moving a point mass from the observation $(\wh x_i, \wh y_i)$ to the fiber set $\mc N_\gamma(x_0) \times \mc Y$. The magnitude of $\kappa_i$ depends on $\gamma$, however, this dependence is implicit. Using this definition of $\kappa$, the next proposition asserts that the value of $\rho_{\min}(x_0, \gamma, \eps)$ can be computed by solving a finite-dimensional optimization problem.

\begin{proposition}[Minimum radius] \label{prop:rho-LB-1}
The value $\rho_{\min}(x_0, \gamma, \eps)$ equals the optimal value of a linear program
\be  \label{eq:greedy-1}
	\rho_{\min}(x_0, \gamma, \eps) = 	\min \left\{
	    \ds N^{-1} \sum_{i\in [N]} \kappa_i \upsilon_i 
		: \upsilon \in [0, 1]^N, \, \ds\sum_{i\in [N]} \upsilon_i \ge N\eps
	\right\},
\ee
where $\kappa$ are defined as in~\eqref{eq:kappa-def}.
Furthermore, there exists a measure $\QQ \in \mbb B_\rho$ such that $\QQ(X \in \mc N_\gamma(x_0)) \ge \eps$ if and only if $\rho \ge \rho_{\min}(x_0, \gamma, \eps)$. 
\end{proposition}

Notice that the minimization problem~\eqref{eq:greedy-1} can be formulated as a fractional knapsack problem, and it can be solved in time $O(N \log N)$ using a greedy heuristic~\cite{ref:dantzig1957discrete}, see also~\cite[Proposition~17.1]{ref:korte2007}.

%%%%%%%%%%%%%%%%%%%%%%%%%%%%%

\subsection{Discussion and roadmap} \label{sec:roadmap}

Problem~\eqref{eq:dro} is governed by three parameters: the ambiguity size $\rho$, the fiber size $\gamma$, and the fiber probability requirement $\eps$. Figure~\ref{fig:roadmap} gives an overview of the results corresponding to all combinations of parameter choices for $(\rho, \gamma, \eps)$.

\begin{figure}[ht!]
    \centering
    \includegraphics[width=1.0\textwidth]{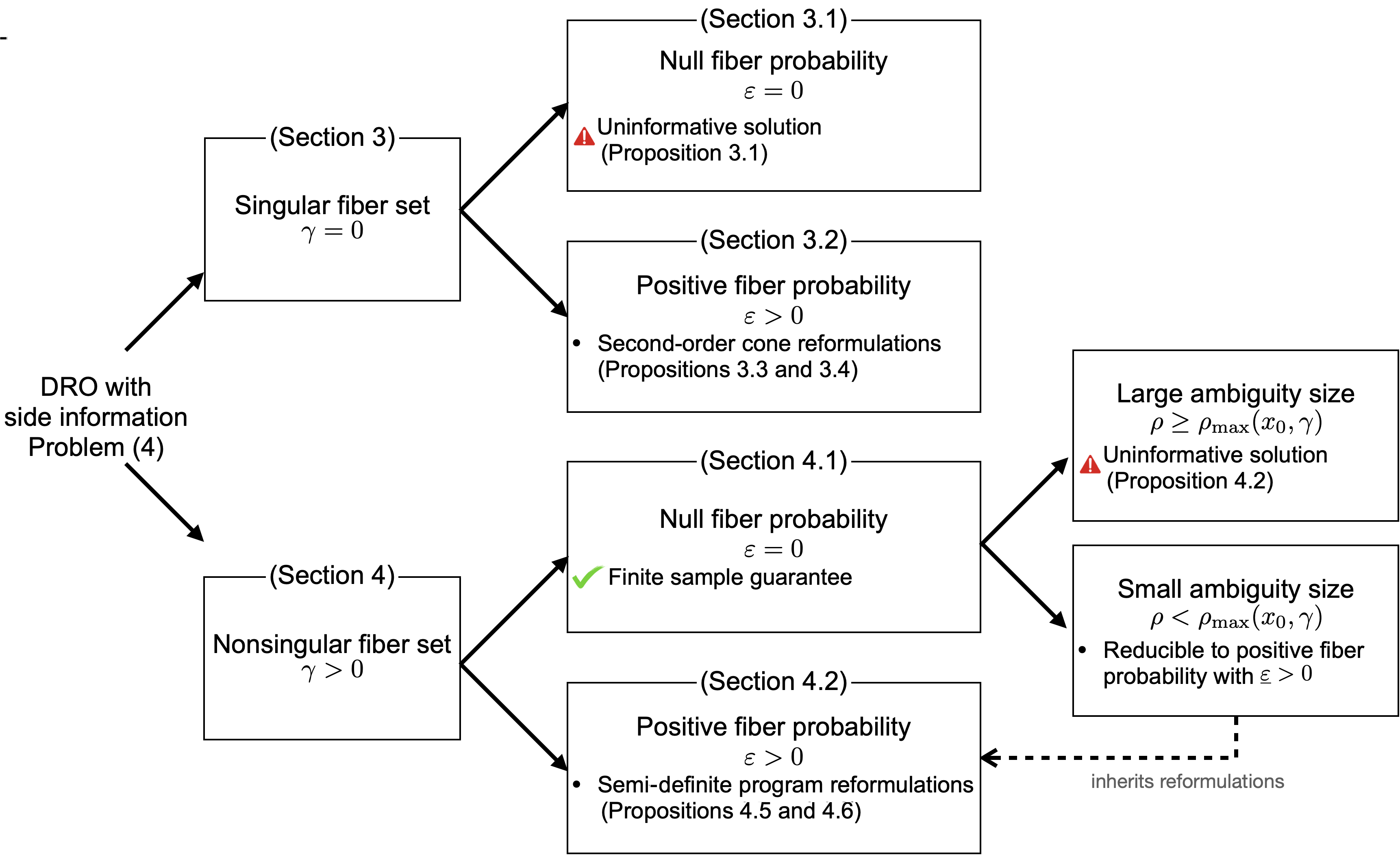}
    \caption{Schematic overview of the results in this paper.}
    \label{fig:roadmap}
\end{figure}

    The parameter $\gamma$ determines the size of the neighborhood around $x_0$. The case with $\gamma = 0$ results in a singular fiber set $\{x_0\} \times \mc Y$: in this case, either problem~\eqref{eq:dro} leads to uninformative conditioning result in Section~\ref{sec:epsgamma0} or we need to impose a strictly positive mass in order to get informative conditioning result ($\eps > 0$) in Section~\ref{sec:gamma0_epspositive}. Notice that imposing the constraint that $\QQ(X = x_0) \ge \eps > 0$ will automatically eliminate distributions with a density around $x_0$. Thus, if the modeler has a prior belief that the data-generating distribution admits a density around $x_0$, then it is more reasonable to use $\gamma > 0$. The parameter $\gamma$ prescribes a neighborhood in the covariate space $\mathcal{X}$, and it resembles the kernel bandwidth in the field of nonparametric statistics. One can rely on this resemblance and tune $\gamma$ based on the number of samples $N$ using the same theoretical guidance for choosing the bandwidth. For example, if $\mathcal{X}$ is one-dimensional, then~\cite[Section 3.2]{ref:fan1996local} suggests that $\gamma$ should be of the order of $1/N^{1/5}$. For higher dimensions, the rate can be obtained by imposing suitable assumptions on the data-generating distribution.
    
    Adding a neighborhood around $x_0$ also implies that the modeler wishes to hedge against potential misspecification or perturbation of the covariate $x_0$. In a portfolio optimization setting, the covariate $x_0$ may represent a combination of macroeconomics (inflation, GDP growth, etc.) and/or market indices (VIX, etc.), which are notoriously difficult to assign an exact value due to time lags, noisy or constantly-updated measurements. It is thus imperative to take the effects of an uncertain covariate $x_0$ into account.

    From a technical viewpoint, the parameter $\eps$ helps to avoid the ill-posedness of the optimization problem. As we will show in Sections~\ref{sec:epsgamma0} and~\ref{sec:gammapositive_eps0}, the conditioning in problem~\eqref{eq:dro} may become \textit{un}informative if $\eps = 0$, and this issue can be eliminated by imposing a strictly positive value of $\eps$. From a modeling viewpoint, the parameter $\eps$ can capture the prior information of the modeler on the magnitude of the density function around $x_0$, which translates integrally to a lower-confidence bound on the probability value assigned to the set $\mc N_{\gamma}(x_0) \times \mc Y$. Alternatively, the parameter $\eps$ can be tuned using a similar idea as tuning the kernel bandwidth in the nonparametric statistics literature: if the neighborhood around $x_0$ is of a radius $\gamma > 0$ and if the data-generating distribution admits a density, then we expect $\eps$ to scale in the order of $\gamma^n$, where $n$ is the dimension of the covariate space $\mc X$. 
    
Finally, the parameter $\rho \in \R_{++}$ hedges against possible error in the finite sample estimation, equivalently known as the epistemic uncertainty. It is possible to tune $\rho$ so that problem~\eqref{eq:dro} possesses some desirable theoretical guarantees. For example, under the condition $\gamma > 0$ and $\eps = 0$ of Section~\ref{sec:gammapositive_eps0}, problem~\eqref{eq:dro} satisfies the finite sample guarantee provided that $\rho$ is chosen so that $\mbb B_\rho$ contains the data-generating distribution with high probability (e.g., by choosing $\rho$ that satisfies the finite guarantee from~\cite{ref:fournier2015rate}). Alternatively, $\rho$ can be also tuned using the robust Wasserstein profile inference (RWPI) approach with the goal of recovering the correct decision~\cite{ref:blanchet2021statistical, ref:si2021testing}.

The most comprehensive and general case in this paper is presented in Section~\ref{sec:gammapositive_epspositive} with all parameters $(\rho, \gamma, \eps)$ being non-zero. In addition, the results in Section~\ref{sec:gammapositive_epspositive} also provide the reformulations for a specific case with $\eps = 0$; see the dashed arrow. Thus, the combination of non-zero parameters $(\rho, \gamma, \eps)$ presented in Section~\ref{sec:gammapositive_epspositive} is useful on two fronts: first, it equips the modeler with the most flexible setting to capture different prior information, and second, it serves as a reformulating auxiliary combination to resolve the problem under a null fiber probability assumption.

While the statistical performance guarantees (either in asymptotic or finite sample regime) of problem~\eqref{eq:dro} are also of high theoretical relevance, our paper focuses mainly on practical relevance. Towards this goal, our main focus is to provide the reformulations for problem~\eqref{eq:dro} via a complete and thorough analysis for each combination of the parameters $(\rho, \gamma, \eps)$. We leave the statistical performance guarantee of problem~\eqref{eq:dro} open for further research.

%%%%%%%%%%%%%%%%%%%%%%%%%%%%%
\section{Tractable Reformulations for Singular Fiber Set}
\label{sec:gamma0}

We study in this section the case where the radius $\gamma$ is zero, which implies that $\mc N_\gamma(x_0) = \{x_0\}$ and the fiber set becomes $\{x_0\} \times \mc Y$. In this case, we simply recover the conventional portfolio allocation problem conditional on $X = x_0$. Interestingly, the qualitative behavior of the robustified conditional portfolio allocation problem~\eqref{eq:dro} depends on whether $\eps= 0$ or $\eps > 0$. We will explore these two cases in the subsequent subsections. 

%%%%%%%%%%%%%%%%%%%%%%
\subsection{Null fiber probability $\eps = 0$}
\label{sec:epsgamma0}

We consider now the situation when $\gamma = \eps = 0$. 
Notice that the probability mass constraint in this case should be taken as a strict inequality of the form $\QQ(X = x_0) > 0$ to avoid conditioning on sets of measure zero. This is equivalent to viewing the mass constraint in the limit as $\eps$ tends to zero. If we use the type-$\infty$ Wasserstein distance to dictate the set $\mbb B_\rho$, then the results from~\cite{ref:nguyen2020distributionally} can be utilized in order the compute the worst-case conditional expected loss. However, in this paper, we use the optimal transport of Definition~\ref{def:OT} to prescribe $\mbb B_\rho$, and the worst-case expected loss becomes uninformative, as is shown in the following result.

\begin{proposition}[Uninformative solution when $\eps = \gamma = 0$] \label{prop:zero-eps}
    Suppose that $\eps = \gamma = 0$. For any $x_0 \in \mc X$ and $\rho \in \R_{++}$, the worst-case conditional expected loss becomes
		\[
		\Sup{\QQ \in \mbb B_\rho, \QQ(X = x_0) > 0}~ \EE_{\QQ}[ \ell(Y, \alpha, \beta) | X =x_0 ] = \Sup{y \in \mc Y}~ \ell(y, \alpha, \beta).
		\]
\end{proposition}

The result in Proposition~\ref{prop:zero-eps} can be justified heuristically as follows. Consider an adversary who can move sample points on $\mc X \times \mc Y$ to maximize the loss subject to the optimal transport distance budget constraint. If $\wh x_i = x_0$, the adversary can slightly perturb $\wh x_i$ within an infinitesimal distance so that the sample no longer belongs to the fiber set. Because of Assumption~\ref{a} on the continuity of the ground metric and the non-emptiness of the neighborhood, this perturbation costs an infinitesimally small amount of energy. The adversary can repeat until there is no sample point lying on the fiber set. Now, the adversary can pick any sample, slice out a tiny amount of mass, and move that slice to any location on $\{x_0\} \times \mc Y$. Because the ground metric is continuous by Assumption~\ref{a}\ref{a:cost} and the slice can be arbitrarily small, this would cost an infinitesimally small amount of energy. As the thin slice can be put on any point on $\{x_0\} \times \mc Y$, the resulting distribution will generate the robust conditional expected loss.

Proposition~\ref{prop:zero-eps} reveals that the conditioning problem with $\eps = \gamma = 0$ results in a robust optimization formulation, which can be overly conservative. This result is also negative in the sense that the worst-case conditional expected loss depends only on the support $\mc Y$, and it does \textit{not} depend on the data $(\wh x_i, \wh y_i)$ that were collected. Thus, in this case, the notion of data-driven decision-making becomes obsolete, and thus we do not pursue the reformulation any further. These observations also highlight the qualitative difference between using the optimal transport ambiguity set $\mathbb{B}_\rho$ as in this paper and using the $\infty$-Wasserstein ambiguity set as in~\cite[\S2]{ref:nguyen2020distributionally}.

%%%%%%%%%%%%%%%%%%%%%%%%%%%%
\subsection{Strictly Positive Fiber Probability $\eps > 0$} \label{sec:gamma0_epspositive}

We now study the situation with a singular fiber set prescribed by $\gamma = 0$, but the probability mass requirement is set to a strictly positive value $\eps \in (0, 1]$. The robust conditional portfolio allocation problem~\eqref{eq:dro} can now be rewritten explicitly as
\be \label{eq:dro-gamma0}
    \Min{\alpha \in \mc A,~\beta \in \mc B}~\Sup{\QQ \in \mbb B_\rho, \QQ(X = x_0) \ge \eps}~ \EE_{\QQ}[ \ell(Y, \alpha, \beta) | X = x_0 ].
\ee

The next theorem asserts that the worst-case conditional expected loss admits a finite-dimensional reformulation. 
\begin{theorem}[Worst-case conditional expected loss when $\gamma = 0$]
\label{thm:type-1-refor}
Suppose that $\eps \in (0, 1]$ and $\rho > \rho_{\min}(x_0, 0, \eps)$. For any feasible solution $(\alpha, \beta)$, we have 
\begin{align*}
&\Sup{\QQ \in \mbb B_\rho, \QQ(X = x_0) \ge \eps}~ \EE_{\QQ}[ \ell(Y, \alpha, \beta) | X = x_0 ] \\
=&\Inf{\substack{\lambda_1\in \R_+\\ \lambda_2\in\R}}\left\{
\rho \lambda_1 + \eps \lambda_2 + \frac{1}{N} \sum_{i \in [N]} \left( \sup_{y_i \in \mc Y} \big\{\eps^{-1} \ell(y_i, \alpha, \beta) - [\DD_{\mc X}(x_0, \wh x_i) + \DD_{\mc Y}(y_i, \wh y_i)] \lambda_1 - \lambda_2 \big\} \right)^+
\right\}.
\end{align*}
\end{theorem}

It is instructive to discuss the insights that lead to the result presented in Theorem~\ref{thm:type-1-refor}. Notice that the supremum problem in the statement of Theorem~\ref{thm:type-1-refor} has a fractional objective function. The first step of the proof establishes that without any loss of optimality, the \textit{inequality} constraint $\QQ(X = x_0) \ge \eps$ can be reduced to an \textit{equality} constraint of the form $\QQ(X = x_0) = \eps$ (see Proposition~\ref{prop:equiv} for a formal statement of this result). Leveraging this result, we derive the equivalence
\[
    \Sup{\QQ \in \mbb B_\rho, \QQ(X = x_0) \ge \eps}~ \EE_{\QQ}[ \ell(Y, \alpha, \beta) | X = x_0 ] = \Sup{\QQ \in \mbb B_\rho, \QQ(X = x_0) = \eps}~ \frac{1}{\eps}\EE_{\QQ}[ \ell(Y, \alpha, \beta) \mathbbm{1}_{\{x_0\}} (X) ],
\]
where the right-hand side problem is linear in the probability measure $\QQ$. At this point, duality techniques can be applied to reformulate the original problem into a finite-dimensional optimization problem. 

We acknowledge that a similar duality result has been proposed in~\cite[Theorem~1]{ref:esteban2020distributionally} using the trimming approach. \modified{Notice that the trimming procedure in~\cite{ref:esteban2020distributionally} imposed that the fictitious adversary can only choose a distribution $\QQ$ satisfying the strict fiber constraint $\QQ(X = x_0) = \eps$. One advantage of this equality constraint is that the conditional expectation is trivially reduced to a linear function of $\QQ$. In direct comparison, our approach endows the adversary with a bigger feasible set, wherein we only impose an \textit{in}equality $\QQ(X = x_0) \ge \eps$. As a result, our adversary is considered more powerful than the trimming adversary, and our problem is significantly harder than the trimming problem because there is no trivial way to reduce the fractional objective function into a linear function of $\QQ$. We contribute an insight that despite a bigger set with $\QQ(X = x_0) \ge \eps$, it is still optimal for our adversary to choose the distribution with an equality value $\QQ(X = x_0) = \eps$, and only after that, the reduction of the conditional expectation to a linear function of $\QQ$ follows.}

By combining the infimum reformulation of Theorem~\ref{thm:type-1-refor} with the outer infimum problem of problem~\eqref{eq:dro-gamma0}, the portfolio allocation problem~\eqref{eq:dro-gamma0} is thus reformulatable as a finite-dimensional optimization problem. More specifically, problem~\eqref{eq:dro-gamma0} becomes
\be\label{eq:reform-1}
\begin{array}{cll}
\inf & \ds \rho \lambda_1 + \eps \lambda_2 + \frac{1}{N}\sum_{i \in [N]} \theta_i \\
\st &\alpha\in \mc A,\;\beta\in\mc B,\; \lambda_1 \in \R_+,\; \lambda_2 \in \R,\; \theta \in \R_+^N \\
&\theta_i \ge \sup_{y_i \in \mc Y} \big\{\eps^{-1} \ell(y_i, \alpha, \beta) -  \DD_{\mc Y}(y_i, \wh y_i) \lambda_1 \big\} - \DD_{\mc X}(x_0, \wh x_i)\lambda_1 - \lambda_2\quad \forall i \in [N].
\end{array}
\ee
In Propositions~\ref{prop:markowitz-reform} and~\ref{prop:mean-cvar-reform}, we provide a second-order cone reformulation of problem~\eqref{eq:reform-1} tailored for the mean-variance and mean-CVaR objective functions for a special instance with $\mc X = \R^n$, $\mc Y  = \R^m$, $\DD_{\mc X}(x,\wh x) = \|x-\wh x\|^2$ and $\DD_{\mc Y}(y,\wh y) = \|y-\wh y\|_2^2$. Notice that $\DD_{\mc X}$ is constructed from an arbitrary norm on $\R^n$ while $\DD_{\mc Y}$ is constructed as the \textit{squared} Euclidean norm on $\R^m$.

\begin{proposition}[Mean-variance loss function]\label{prop:markowitz-reform} 
Suppose that $\ell$ is the mean-variance loss function of the form~\eqref{eq:mean-var-loss}, $\gamma = 0$, $\eps \in (0, 1]$ and $\rho > \rho_{\min}(x_0, 0, \eps)$. Suppose in addition that $\mc X = \R^n, \mc Y  = \R^m$, $\DD_{\mc X}(x,\wh x) = \|x-\wh x\|^2$ and $\DD_{\mc Y}(y,\wh y) = \|y-\wh y\|_2^2$. The distributionally robust portfolio allocation model with side information~\eqref{eq:dro-gamma0} is equivalent to the second-order cone program 
\be \notag
\begin{array}{cll}
\inf & \ds \rho \lambda_1 + \eps \lambda_2 + \frac{1}{N}\sum_{i \in [N]} \theta_i \\
\st & \alpha\in \mc A,\;\beta\in\mc B,\;\lambda_1 \in \R_+,\; \lambda_2 \in \R,\; \theta \in \R_+^N,\; z\in \R_+^N,\; w\in\R\\
& 0\leq w\leq 1, \quad \left\|
\begin{bmatrix}
2\alpha\\1 - w - \eps\lambda_1
\end{bmatrix}\right\|_2
\leq 1 - w + \eps\lambda_1\\
&  z_i = \eps\theta_i+\eps\|x_0 - \wh x_i\|^2 \lambda_1 + \eps\lambda_2 + \frac{1}{4}\eta^2+\eta\beta,\quad \left\|
\begin{bmatrix}
2\,\wh y_i^\top \alpha - 2\beta- \eta\\
z_i - w
\end{bmatrix}\right\|_2
\leq z_i + w & \forall i \in [N].
\end{array}
\ee
\end{proposition}

\begin{proposition}[Mean-CVaR loss function]\label{prop:mean-cvar-reform} 
Suppose that $\ell$ is the mean-CVaR loss function of the form~\eqref{eq:mean-cvar-loss}, $\gamma = 0$, $\eps \in (0, 1]$ and $\rho > \rho_{\min}(x_0, 0, \eps)$. Suppose in addition that $\mc X = \R^n$, $\mc Y  = \R^m$, $\DD_{\mc X}(x,\wh x) = \|x-\wh x\|^2$ and $\DD_{\mc Y}(y,\wh y) = \|y-\wh y\|_2^2$. The distributionally robust portfolio allocation model with side information~\eqref{eq:dro-gamma0} is equivalent to the second-order cone program 
\be \notag
\begin{array}{cll}
\inf & \ds \rho \lambda_1 + \eps \lambda_2 + \frac{1}{N}\sum_{i \in [N]} \theta_i \\
\st & \alpha\in \mc A,\;\beta\in\mc B,\; \lambda_1 \in \R_+,\; \lambda_2 \in \R,\; \theta \in \R_+^N,\; z\in \R_+^N,\; \tilde z\in \R_+^N\\
&\left. 
\!\!\!
\begin{array}{l}
z_i = \theta_i +\lambda_1 \|x_0 - \wh x_i\|^2+ \lambda_2 + \eps^{-1}\eta \wh y_i^\top \alpha - \eps^{-1}\beta \\
\tilde{z}_i = \theta_i +\lambda_1 \|x_0 - \wh x_i\|^2+ \lambda_2 + \eps^{-1}(\eta + \frac{1}{\tau}) \wh y_i^\top \alpha - \eps^{-1}(1 - \frac{1}{\tau})\beta\\
\left\|
\begin{bmatrix}
\eps^{-1}\eta \alpha\\
z_i-\lambda_1
\end{bmatrix}\right\|_2
\leq z_i+\lambda_1, \quad \left\|
\begin{bmatrix}
\eps^{-1}(\eta + \tau^{-1}) \alpha\\
\tilde z_i-\lambda_1
\end{bmatrix}\right\|_2
\leq \tilde z_i+\lambda_1
\end{array} \right\} \forall i \in [N].
\end{array}
\ee
\end{proposition}

Both Proposition~\ref{prop:markowitz-reform} and \ref{prop:mean-cvar-reform} leverage the fact that $\mc Y = \R^m$ in order to simplify the semi-infinite constraints into second-order cone constraints. We re-emphasize that under the assumption $\mc Y= \R^m$ of Proposition~\ref{prop:markowitz-reform} and~\ref{prop:mean-cvar-reform}, formulation~\eqref{eq:dro-gamma0} is a conservative approximation of the distributionally robust mean-variance and mean-CVaR portfolio allocation problem, respectively. In case the set $\mc Y$ is an ellipsoid of the form $\mc Y = \{ y \in \R^m: y^\top Q y + 2 q^\top y + q_0 \le 0\}$ with a non-empty interior, then semi-definite cone constraint counterparts are also available by employing the S-lemma~\cite{ref:polik2007Slemma}.

%%%%%%%%%%%%%%%%%%%%%%%%
\section{Tractable Reformulations for Nonsingular Conditioning Set}
\label{sec:gamma-positive}

 We now focus on the portfolio allocation conditional on $X \in \mc N_\gamma(x_0)$ for some radius $\gamma > 0$. A fiber set of the form $\mc N_\gamma(x_0) \times \mc Y$ was first used in~\cite{ref:nguyen2020distributionally} in the conditional estimation setting with the aim to hedge against noisy covariate information $x_0$ and also to improve the statistical performance of the solution approach. The results from~\cite{ref:nguyen2020distributionally} rely heavily on the specification of the ambiguity set using the type-$\infty$ Wasserstein distance, and these results are not transferrable to the ambiguity set under investigation in this paper. As a parallel counterpart to Section~\ref{sec:gamma0}, we will study two separate cases depending on whether the probability requirement $\eps$ is zero or strictly positive. We start by discussing the case when $\eps = 0$.

\subsection{Null fiber probability $\eps = 0$} \label{sec:gammapositive_eps0}

Similar to Section~\ref{sec:epsgamma0}, we consider the probability mass constraint of the form $\QQ(X \in \mc N_\gamma(x_0)) > 0$ with a strict inequality to avoid conditioning on sets of measure zero. Problem~\eqref{eq:dro} becomes
\be
\label{eq:eps0}
\Min{\alpha \in \mc A,~\beta \in \mc B}~\Sup{\QQ \in \mbb B_\rho, \QQ(X \in \mc N_\gamma(x_0)) > 0}~ \EE_{\QQ}[ \ell(Y, \alpha, \beta) | X \in \mc N_\gamma(x_0) ].
\ee
Problem~\eqref{eq:eps0} is of particular interest thanks to its finite sample guarantee: if the samples are independently and identically distributed, and the radius $\rho$ is chosen judiciously, then the (unknown) data-generating distribution belongs to $\mbb B_\rho$ with high probability. As such, the optimal value of problem~\eqref{eq:eps0} constitutes an upper bound on the worst-case conditional expected loss under the data-generating distribution. One possible way to choose $\rho$ to obtain the finite sample guarantee is by using the seminal result from~\cite{ref:fournier2015rate}.\footnote{The formal statement of the finite sample guarantee is omitted for brevity. Interested readers may refer to~\cite{ref:esfahani2018data} for similar finite sample guarantee results.}

The complication of conditioning on sets of measure zero, as highlighted in Section~\ref{sec:epsgamma0} for $\eps = 0$, still arises when the radius $\rho$ is big enough. To illustrate this problem, we define the following quantity
	\[
	    \rho_{\max}(x_0, \gamma) \Let \inf\left\{ \Wass(\QQ, \Pnom): \QQ \in \mc M(\mc X \times \mc Y),~\QQ(X \in \mc N_{\gamma}(x_0)) = 0 \right\}.
	\]
	Intuitively, $\rho_{\max}(x_0, \gamma)$ indicates the minimum budget required to transport all the training samples out of the fiber $\mc N_\gamma(x_0) \times \mc Y$. If the empirical distribution $\Pnom$ satisfies $\Pnom(X \in \mc N_\gamma(x_0)) = 0$, which means that there is no training samples $(\wh x_i, \wh y_i)$ falling inside the fiber set $\mc N_\gamma(x_0) \times \mc Y$, then it is trivial that $\rho_{\max}(x_0, \gamma) = 0$. If $\Pnom(X \in \mc N_\gamma(x_0)) > 0$, then the value of $\rho_{\max}(x_0, \gamma)$ is known in closed form. To this end, define the following index sets
	\be \label{eq:I-def}
	    \mc I_1 = \{i \in [N]: \wh x_i \in \mc N_{\gamma}(x_0)\}, \qquad \mc I_2 = \{ i \in [N]: \wh x_i \not \in \mc N_{\gamma}(x_0) \}.
	\ee
	The sets $\mc I_1$ and $\mc I_2$ divide the training samples into two mutually exclusive sets dependent on whether the training samples fall inside or outside the fiber. For any $i \in [N]$, let $d_i$ be the distance from $\wh x_i$ to the boundary $\partial \mc N_{\gamma}(x_0)$ of the set $\mc N_{\gamma}(x_0)$, that is,
	\be \label{eq:d-def}
	    \forall i \in [N]: \qquad d_i = \Min{x \in \partial \mc N_\gamma(x_0)} \DD_{\mc X} (x, \wh x_i).
	\ee
	Note that the distance $d_i$ defined above is closely related to the values of $\kappa_i$ defined in~\eqref{eq:kappa-def}. Indeed, if $\wh x_i \not \in \mc N_\gamma(x_0)$ then $d_i$ is equal to $\kappa_i$. However, if $\wh x_i$ is in the interior of the set $\mc N_{\gamma}(x_0)$ then $d_i > 0$ while $\kappa_i = 0$. Evaluating $d_i$ is, unfortunately, difficult in general because the set $\partial \mc N_\gamma(x_0)$ may be non-convex. However, $d_i$ can be computed efficiently under certain choices of $\DD_{\mc X}$. For example, when $\DD_{\mc X}$ is the Euclidean norm, then it is easy to see that
	\[
	d_i = \Min{x: \| x - x_0\|_2 = \gamma}~\| x - \wh x_i\|_2 = \gamma - \| x_0 - \wh x_i \|_2
	\]
	for any $i \in \mc I_1$ implying that $\| x_0 - \wh x_i \|_2 \le \gamma$.
	
	Using the definition of the set $\mc I_1$ and the boundary projection distance $d_i$, the maximum radius $\rho_{\max}(x_0, \gamma)$ can be computed in closed form, as asserted by the next proposition.
	\begin{proposition}[Expression for $\rho_{\max}(x_0, \gamma)$] \label{prop:rho_upper}
	    We have
	    $\rho_{\max}(x_0, \gamma) = N^{-1} \sum_{i \in \mc I_1} d_i$.
	\end{proposition}
	
	The result of Proposition~\ref{prop:rho_upper} is also intuitive: if $\wh x_i \in \mc N_{\gamma}(x_0)$ then $i \in \mc I_1$, and we will need a sample-wise budget of $d_i/N$ to transport $(\wh x_i, \wh y_i)$ out of the fiber set $\mc N_{\gamma}(x_0) \times \mc Y$. The value $\rho_{\max}(x_0, \gamma)$ is thus obtained by summing all $d_i/N$ over the set $\mc I_1$. If $\mc I_1 = \emptyset$ then there is no training sample inside the fiber set, and thus $\rho_{\max}(x_0, \gamma) = 0$. Proposition~\ref{prop:rho_upper} also highlights that computing $\rho_{\max}(x_0, \gamma)$ necessitates evaluating $|\mc I_1|$ values $d_i$ for each $i \in \mc I_1$.
	
	The computation of $\rho_{\max}(x_0, \gamma)$ provides a natural upper bound on the radius $\rho$. Indeed, if the radius $\rho$ prescribing the ambiguity set is bigger than $\rho_{\max}(x_0, \gamma)$, we recover the robust worst-case conditional loss. This robust loss is uninformative because it depends only on the support $\mc Y$, and it is \textit{in}dependent on the training samples. This negative result is reminiscent of Proposition~\ref{prop:zero-eps} and highlights the sophistication of the distributionally robust conditional decision-making problem.
	
\begin{proposition}[Uninformative solution when $\rho$ is sufficiently large] \label{prop:robust_eps0}
When $\gamma \in \R_{++}$, $\eps = 0$, and $\rho > \rho_{\max}(x_0, \gamma)$, we have
\[
    \Sup{\QQ \in \mbb B_\rho, \QQ(X \in \mc N_\gamma(x_0)) > 0}~ \EE_{\QQ}[ \ell(Y, \alpha, \beta) | X \in \mc N_\gamma(x_0) ] = \Sup{y \in \mc Y}~ \ell(y, \alpha, \beta).
\]
\end{proposition}

We now provide a heuristic justification for the result of Proposition~\ref{prop:robust_eps0}. To form the worst-case distribution, the adversary first moves all the samples with index in $\mc I_1$ out of the fiber, which would cost an amount of energy $\rho_{\max}(x_0, \gamma)$. After that, the adversary can pick any sample, slice out an infinitesimally small amount of mass, and then move that slice to any point in the fiber $\mc N_{\gamma}(x_0) \times \mc Y$. Because $\rho > \rho_{\max}(x_0, \gamma)$ and because the ground metric is continuous, this new arrangement of the samples is feasible and constitutes the worst-case distribution that maximizes the conditional expected loss.

Consider the case where the ambiguity size $\rho$ is strictly smaller than the maximum value $\rho_{\max}(x_0, \gamma)$. It can be shown that in this situation, any distribution $\QQ$ that is feasible in the supremum problem of~\eqref{eq:dro-gamma0} should satisfy $\QQ(X \in \mc N_{\gamma}(x_0)) \ge \underline\eps$ for some strictly positive lower bound $\underline \eps$. Further, this value $\underline \eps$ can be quantified by solving a linear optimization problem.

\begin{proposition}[Strictly positive probability requirement equivalence] \label{prop:eps_lower}
    Suppose that $\gamma \in \R_{++}$, $\eps = 0$ and $\rho < \rho_{\max}(x_0, \gamma)$. Then there exists an $\underline \eps>0$, such that the distributionally robust portfolio allocation model with side information~\eqref{eq:dro-gamma0} is equivalent to
    \[
        \min_{\alpha\in\mc A,~\beta \in\mc B}\;\Sup{\QQ \in \mbb B_\rho, \QQ(X \in \mc N_\gamma(x_0)) \geq \underline \eps}~ \EE_{\QQ}[ \ell(Y, \alpha, \beta) | X \in \mc N_\gamma(x_0) ].
    \]
    In particular, this equivalence holds for
    \[
       \underline \eps = \min\left\{ \frac{1}{N} \sum_{i \in \mc I_1} p_i : p \in [0, 1]^N,~ \frac{1}{N}\sum_{i \in \mc I_1} d_i (1-p_i) \le \rho \right\}.
    \]    
\end{proposition}

The linear program that defines~$\underline \eps$ in Proposition~\ref{prop:eps_lower} can be further shown to be a fractional knapsack problem, and it can be solved efficiently to optimality using a greedy heuristic~\cite[Proposition~17.1]{ref:korte2007}. The conditions of Proposition~\ref{prop:eps_lower} are satisfied only when $\rho_{\max}(x_0, \gamma) > 0$, which further implies that $\mc I_1 \neq \emptyset$ and there exists at least one training sample in the fiber set $\mc N_\gamma(x_0) \times \mc Y$. With an ambiguity size $\rho$ which is strictly smaller than $\rho_{\max}(x_0, \gamma)$, it is not possible for the adversary to remove all the samples out of the fiber set $\mc N_\gamma(x_0) \times \mc Y$. The lower bound value $\underline \eps$ represents the lowest possible amount of probability mass left on the fiber. 

The result of Proposition~\ref{prop:eps_lower} indicates that the portfolio allocation problem when $\rho < \rho_{\max}(x_0, \gamma)$ can be obtained by solving the general problem with a strictly positive probability mass requirement. This general case is our subject of study in the next subsection.

%%%%%%%%%%%%%%%%%%%%%%%%
\subsection{Strictly positive fiber probability $\eps > 0$} \label{sec:gammapositive_epspositive}

We now consider the last and most general case with a nonsingular fiber set $\gamma > 0$ and a strictly positive fiber probability $\eps > 0$. More specifically, we aim to solve the portfolio allocation of the form
\be
\label{eq:dro2}
    \Min{\alpha \in \mc A,~\beta \in \mc B}~\Sup{\QQ \in \mbb B_\rho, \QQ(X \in \mc N_{\gamma}(x_0)) \ge \eps}~ \EE_{\QQ}[ \ell(Y, \alpha, \beta) | X \in \mc N_{\gamma}(x_0) ].
\ee
Notice that problem~\eqref{eq:dro2} is also relevant to the case with null probability requirement of Section~\ref{sec:gammapositive_eps0} because problem~\eqref{eq:dro2} is equivalent to problem~\eqref{eq:eps0} when $\rho < \rho_{\max}(x_0, \gamma)$ by choosing a proper value of $\eps$ (cf.~Proposition~\ref{prop:eps_lower}). 
As the first step towards solving~\eqref{eq:dro2}, we define the following feasible set for the dual variables
\be \label{eq:dual-set}
\mc V \Let \left\{
\begin{array}{ll}
      (\lambda, s, \nu^+, \nu^-, \phi, \varphi, \psi) \in \R_+^N \times \R^N \times \R_+ \times \R_+ \times \R \times \R_+ \times \R_+^N \text{ such that: } \\
            \phi - d_i \varphi + \psi_i   - s_i \ge 0 &\forall i \in \mc I_1 \\
            \phi + d_i \varphi + \psi_i  - s_i \ge 0 &\forall i \in \mc I_2 \\
            \nu^+ - \nu^- + (\sum_{i \in \mc I_1} d_i - N\rho) \varphi  -\sum_{i \in [N]} \psi_i  \ge 0 \\
            \varphi - \lambda_i \ge 0 & \forall i \in [N]
\end{array}
\right\}.
\ee
The next theorem presents the finite-dimensional reformulation of the worst-case conditional expected loss.

\begin{theorem}[Worst-case conditional expected loss for $\eps > 0$] 
\label{thm:eps-zero}
Suppose that $\gamma \in \R_{++}$, $\eps \in \R_{++}$, and $\rho > \rho_{\min}(x_0,\gamma,\eps)$. Let the parameters $d_i$ be defined as in~\eqref{eq:d-def}. For any feasible solution $(\alpha, \beta)$, we have
\begin{align*}
&\Sup{\QQ \in \mbb B_\rho, \QQ(X \in \mc N_\gamma(x_0)) \ge \eps}~ \EE_{\QQ}[ \ell(Y, \alpha, \beta) | X \in \mc N_\gamma(x_0) ] 
= \left\{
        \begin{array}{cll}
            \inf & \phi + (N\eps)^{-1} \nu^+ - N^{-1}\nu^- \\
            \st & (\lambda, s, \nu^+, \nu^-, \phi, \varphi, \psi) \in \mc V \\
            &s_i\geq \Sup{y_i \in \mc Y}~ \{\ell(y_i, \alpha, \beta)  - \lambda_i \DD_{\mc Y}(y_i, \wh y_i)\}  &\forall i\in[N].
        \end{array}    
\right. 
\end{align*}
\end{theorem}

	We can join the minimization operator over $(\alpha, \beta)$ to the result of Theorem~\ref{thm:eps-zero}, and the portfolio allocation problem with side information~\eqref{eq:dro2} is equivalent to the following finite-dimensional optimization problem
	\begin{equation}
	    \label{eq:optimization-eps-0}
	    \begin{array}{cll}
            \min & \phi + (N\eps)^{-1} \nu^+ - N^{-1}\nu^- \\
            \st & \alpha \in \mc A,~\beta\in\mc B,~(\lambda, s, \nu^+, \nu^-, \phi, \varphi, \psi) \in \mc V \\
            &s_i\geq 
            \Sup{y_i\in\mc Y}~ \{
            \ell(y_i, \alpha, \beta)  - \lambda_i \DD_{\mc Y}(y_i, \wh y_i) \}
            &\forall i\in[N].
        \end{array}   
	\end{equation}
	The last constraint of problem~\eqref{eq:optimization-eps-0} is a semi-infinite constraint that can be reformulated into a set of semi-definite constraints under specific situations. These results are  highlighted in Proposition~\ref{prop:markowitz-reform-2} and~\ref{prop:mean-cvar-reform-2}.

	\begin{proposition}[Mean-variance loss function]\label{prop:markowitz-reform-2} 
    Suppose that $\ell$ is the mean-variance loss function of the form~\eqref{eq:mean-var-loss}, $\gamma \in \R_{++}$, $\eps \in (0, 1]$ and $\rho > \rho_{\min}(x_0,\gamma,\eps)$. Suppose in addition that $\mc Y  = \R^m$ and $\DD_{\mc Y}(y,\wh y) = \|y-\wh y\|_2^2$. Let the parameters $d_i$ be defined as in~\eqref{eq:d-def}. The distributionally robust portfolio allocation model with side information~\eqref{eq:optimization-eps-0} is equivalent to the semi-definite optimization problem
    \begin{equation*}
        \begin{array}{cll} 
            \min & \phi + (N\eps)^{-1} \nu^+ - N^{-1}\nu^- \\
            \st & \alpha \in \mc A,~\beta \in \mc B,~t\in\R_+,~A_i \in \PSD^m\quad \forall i \in [N] \\
            & (\lambda, s, \nu^+, \nu^-, \phi, \varphi, \psi) \in \mc V \\
            &\begin{bmatrix}
                \lambda_i I - A_i & \alpha \\ \alpha^\top & 1
            \end{bmatrix} \succeq 0,~\begin{bmatrix}
                A_i & (\beta + \eta/2) \alpha - \lambda_i \wh y_i \\ (\beta + \eta/2) \alpha^\top  - \lambda_i \wh y_i^\top & s_i + \lambda_i \| \wh y_i\|_2^2 - t 
            \end{bmatrix} \succeq 0 & \forall i \in [N] \\
            & \begin{bmatrix}
                t & \beta \\ \beta & 1
            \end{bmatrix} \succeq 0.
        \end{array}
    \end{equation*}
    \end{proposition}

    \begin{proposition}[Mean-CVaR loss function]\label{prop:mean-cvar-reform-2} 
    Suppose that $\ell$ is the mean-CVaR loss function of the form~\eqref{eq:mean-cvar-loss}, $\gamma \in \R_{++}$, $\eps \in (0, 1]$ and $\rho > \rho_{\min}(x_0, \gamma, \eps)$. Suppose in addition that $\mc Y  = \R^m$ and $\DD_{\mc Y}(y,\wh y) = \|y-\wh y\|_2^2$. Let the parameters $d_i$ be defined as in~\eqref{eq:d-def}. The distributionally robust portfolio allocation model with side information~\eqref{eq:optimization-eps-0} is equivalent to the semi-definite optimization problem
    \begin{equation*}
        \begin{array}{cll} 
            \min & \phi + (N\eps)^{-1} \nu^+ - N^{-1}\nu^- \\
            \st & \alpha \in \mc A,~\beta \in \mc B,~(\lambda, s, \nu^+, \nu^-, \phi, \varphi, \psi) \in \mc V \\
            &\begin{bmatrix}
                \lambda_i I & \frac{\eta}{2} \alpha - \lambda_i \wh y_i \\ \frac{\eta}{2} \alpha^\top  - \lambda_i \wh y_i^\top & s_i + \lambda_i \| \wh y_i\|_2^2 - \beta
            \end{bmatrix} \succeq 0,~
            \begin{bmatrix}
                \lambda_i I & \frac{\eta \tau + 1}{2\tau}\alpha - \lambda_i \wh y_i \\ \frac{\tau + 1}{2\tau} \alpha^\top  - \lambda_i \wh y_i^\top & s_i + \lambda_i \| \wh y_i\|_2^2 - (1-1/\tau)\beta
            \end{bmatrix} \succeq 0 & \forall i \in [N].
        \end{array}
    \end{equation*}
    \end{proposition}

    Both Proposition~\ref{prop:markowitz-reform-2} and \ref{prop:mean-cvar-reform-2} leverage the fact that $\mc Y = \R^m$ to reformulate the problem using semi-definite constraint. Again, under the assumption $\mc Y= \R^m$, the reformulation~\eqref{eq:optimization-eps-0} is a conservative approximation of the distributionally robust conditional portfolio allocation problem; thus the reformulations in  Proposition~\ref{prop:markowitz-reform-2} and~\ref{prop:mean-cvar-reform-2} are also conservative approximations of the distributionally robust conditional mean-variance and mean-CVaR portfolio allocation problems, respectively. In case $\mc Y$ is an ellipsoid with a non-empty interior of the form $\mc Y = \{ y \in \R^m: y^\top Q y + 2 q^\top y + q_0 \le 0\}$ for some symmetric matrix $Q$, then the S-lemma~\cite{ref:polik2007Slemma} can also be applied to devise a similar semi-definite optimization problem, the details can be found in Appendix~\ref{appendix:refor-compact}. 
    
    Despite having a fractional objective function, the results in this section assert that distributionally robust portfolio allocation with side information problems can be solved efficiently using convex conic programming solvers. This is in stark contrast with existing results in distributionally robust fractional programming where only \textit{non}convex reformulations are available, and the optimal solution is obtained by solving a sequence of convex optimization problems after bisection~\cite{ref:ji2020data, ref:zhao2017distributionally}.

\section{Numerical Experiment}
\label{sec:experiment}
\modified{
We conduct a comprehensive numerical experiment using real data with the following methods: (i) the equal-weighted model, (ii) the unconditional model, (iii) the distributionally robust unconditional model, (iv) the conditional model, (v) the distributionally robust conditional model with type-$\infty$ Wasserstein ambiguity set, and (vi) our proposed optimal transport conditional model. We experiment with two popular portfolio allocation strategies: the mean-variance criterion in Section~\ref{sec:experiment_mean-variance} and the mean-CVaR criterion in Section~\ref{sec:experiment_mean-cvar}. The details about the optimization models for the mean-variance and mean-CVaR problems will be provided in the corresponding section.

\textbf{Asset return and side information:} For the asset returns, we take the historical S\&P500 constituents data from January 01, 2017, to March 31, 2023.\footnote{The data is downloaded from the Wharton Research Data Services: \url{https://wrds-www.wharton.upenn.edu}.} To ensure a stable portfolio construction, we retain only assets that have been in the S\&P500 since 2010, which gives a universe of 399 assets in total. The sampling frequency is one sample per day: at day $t$, the response variable $Y$ is defined as the 1-day stock percentage return from day $t$ to day $t+1$. For the side information, we follow \cite{ref:Chenreddy2022data} and use five publicly available market indices: (i) Volatility Index (VIX), (ii) 10-year Treasury Yield Index (TNX), (iii) Crude Oil Index (CL=F), (iv) S\&P 500 (GSPC), (v) Dow Jones Index (DJI) to construct the side information covariate $X \in \R^5$.\footnote{Side information indices can be downloaded from yfinance: \url{https://pypi.org/project/yfinance/}.} 

We preprocess the side information using the Gaussian kernel bandwidth selection method. Namely, over the period from January 1, 2017 to December 31, 2020, we pick the side information indices as the predictors and pick the unweighted average stock return (averaged over the asset universe ) as the target values. We then use the least-squares cross-validation method from~\cite{li2004cross} to estimate the optimal Gaussian kernel bandwidth for each side information index in the Nadaraya-Watson kernel prediction. Then for each side information index, we divide the raw values by the corresponding optimal bandwidth.}

\modified{
\textbf{Parameter tuning:} We tune the hyperparameters of each portfolio allocation method as follows. At each trading day $t$ in the period between January 1, 2019 and December 31, 2020, we randomly sample $m = 50$ stocks from the asset universe to create a \emph{stock pool} instance. Then, we use the historical observations of the side information indices and the asset returns in the stock pool in the two-year window ($252\times 2$ observations) before $t$ to form the nominal distribution $\hat{\mathbb{P}}$. The covariate $x_0$ for time $t$ is chosen as the observed side information at time $t$. We then find the optimal portfolio allocation of each method by solving the corresponding optimization problem with varying hyperparameter values. The validation block consists of the subsequent $60$ trading days after $t$: we evaluate the expected loss of the portfolio by calculating the objective value $\EE_{\PP}[\ell(Y, \alpha, \beta)]$ using these 60 samples. We pick the hyperparameters that \textit{minimize} the average expected loss over the whole tuning period from January 1, 2019, to December 31, 2020.
}

\modified{\textbf{Out-of-sample performance computation:} To obtain the out-of-sample performance, we evaluate each method with the optimal hyperparameters in a rolling horizon scheme on the two-year test period from January 1, 2021, to December 31, 2022. On each trading day in the test period, we randomly sample $m=50$ assets from the universe to form the stock pool, and we pick the prior two-year window data on the side information and asset returns to form the nominal distribution. We solve the allocation problem associated with each method using the corresponding tuned hyperparameters to obtain the portfolio allocation, and then we calculate the realized daily percentage return and expected loss of this portfolio based on the empirical distribution of the subsequent 60 days.

If we denote the daily percentage return of a portfolio by $r_i$, $i = 1, \ldots, 60$, then the annualized Sharpe ratio is computed as:
\begin{equation*}
    \mbox{annualized Sharpe ratio}
    := \sqrt{252}\times
    \left(
    \widehat{\mbox{mean}}(\{r_i\}_{i=1, \ldots, 60}) \Big/
    \widehat{\mbox{std}}(\{r_i\}_{i=1, \ldots, 60})
    \right),
\end{equation*}
where $\widehat{\mbox{mean}}$ and $\widehat{\mbox{std}}$ denote the empirical mean and standard deviation estimator, respectively. 

The test procedure is repeated for trading days from the test period, resulting in $252\times2$ independent test experiments. For each optimally tuned method, we compute $252\times2$ Sharpe ratios, and will report on the average value. Other performance metrics of the portfolios, such as the expected loss $\ell(\cdot)$ values, the mean returns, the CVaR values, and the standard deviations, are computed using the empirical distribution of the 60 subsequent realized returns, and averaged over $252\times 2$ test experiments.
}

\textbf{Feasible set.} All methods share the same feasible set for the optimization problems. If an asset in the stock pool is not tradable at day $t$, we impose its weight as zero. 
In addition, we prohibit short selling by enforcing the investment weight vector $\alpha$ to be non-negative. 
Thus, the feasible region for the portfolio allocation problem is a subset of the non-negative simplex
\[
    \mc A = \left\{\alpha\in\R_{+}^{m} \mid \sum_{i\in[m]} \alpha_i = 1, ~\alpha_j = 0 \mbox{ if stock }j\mbox{ is not tradable} \right\}.
\]

Notice that the feasible set $\mc B$ is set to $\R$ throughout the experiment. The ground costs are chosen with $\DD_{\mc X}(x,\wh x) = \|x-\wh x\|_2^2$ and $\DD_{\mc Y}(y,\wh y) = \|y-\wh y\|_2^2$. All optimization problems are solved using the MOSEK quadratic program solver~\cite{mosek}. \modified{Due to the computational complexity of the hyperparameter fine-tuning process, we do not consider baselines that are reformulated as semi-definite programs in the experiment. We also provide computational solution time on the mean-variance and mean-CVaR models with $m=50$ and $399$ assets in Appendix \ref{sec:mean-cvar}.} All codes, data, and results are available at: 
\url{https://github.com/shanshanwang2019/Robustifying-Conditional-Portfolio-Decisions-via-Optimal-Transport}.

\subsection{Mean-Variance Portfolio Allocation Problems}\label{sec:experiment_mean-variance}

We first consider the portfolio allocation problem using the mean-variance criterion as in Lemma~\ref{lemma:MV}. Below, the loss function~$\ell$ is prescribed in~\eqref{eq:mean-var-loss}. We compare the following methods:
\begin{enumerate}[label = (\roman*), leftmargin = 5mm]
\item the Equal Weighted model (EW): in this case, the portfolio allocation solves
    \be \label{eq:MV} 
        \Min{\alpha \in \mc A}~\|\alpha\|_2.
    \tag{EW}
    \ee
    If all stocks are tradable, then the EW allocation coincides with the $1/m$-portfolio, which is renowned for its robustness~\cite{ref:demiguel2009optimal}. This portfolio is parameter-free, and was previously shown to be the limit of the distributionally robust portfolio allocations as the ambiguity size increases~\cite{ref:pflug20121/N}.
    \item the unconditional Mean-Variance model (MV) \cite{ref:markowitz1952portfolio}: the portfolio allocation solves
    \be 
    \Min{\alpha \in \mc A,~\beta\in \R} \EE_{\Pnom} [ \ell(Y, \alpha, \beta)]
    \tag{MV}
    \ee
    with the distribution $\Pnom$ is the empirical distribution supported on the available return data. 
    \item the Distributionally Robust unconditional Mean-Variance model (DRMV): the portfolio allocation solves
    \be
        \Min{\alpha \in \mc A}~\sqrt{\alpha^\top \textrm{Variance}_{\Pnom}(Y) \alpha } 
        -\eta\cdot \alpha^\top\EE_{\Pnom}[Y]
        +\sqrt{(1+\eta^2)\rho}\|\alpha\|_2.
        \tag{DRMV}
    \ee
    The above optimization problem is the reformulation of the distributionally robust mean-variance portfolio allocation with a Wasserstein ambiguity set on the distributions of the asset returns $Y$~\cite{ref:blanchet2018distributionally}. The tuning parameter for this method is $\rho \in \{0.05, 0.1, 0.25\}$.
    \item the Conditional Mean-Variance model (CMV): the portfolio allocation solves
    \be 
    \Min{\alpha \in \mc A,~\beta\in \R} \EE_{\Pnom} [ \ell(Y, \alpha, \beta)|X \in \mc N_{\gamma}(x_0)].
    \tag{CMV}
    \ee
    The parameter $\gamma$ is set to the $a$-quantile of the empirical distribution of the distance between $x_0$ and the training covariate vectors\footnote{More precisely, this is the empirical distribution supported on $\Delta_i = \DD_{\mc X}(\wh x_i, x_0)$ for $i = 1, \ldots, N$.}, where the quantile value is in the range $a \in \{5\%, 10\%, 25\%\}$. Notice that by the choice of $\gamma$, the empirical distribution satisfies $\Pnom(X \in \mc N_{\gamma}(x_0)) > 0$ and the conditional expectation is well-defined.
    \item the Distributionally Robust Conditional Mean-Variance model (DRCMV) with type-$\infty$ Wasserstein ambiguity set \cite{ref:nguyen2020distributionally}, in which the portfolio allocation solves 
    \be 
    \Min{\alpha \in \mc A,~\beta\in \R} \Sup{\QQ \in \mbb B^\infty_\rho, \QQ(X \in \mc N_\gamma(x_0)) > 0} ~\EE_{\QQ} [ \ell(Y, \alpha, \beta) | X \in \mc N_{\gamma}(x_0)].
    \tag{DRCMV}
    \ee
     The parameter $\gamma$ is set to the $a$-quantile of the empirical distribution of the distance between $x_0$ and the training covariate vectors, where the quantile value is in the range $\{5\%, 10\%, 25\%\}$. The radius $\rho$ is set to the $b$-quantile of the empirical distribution of the distance between $x_0$ and the training covariate vectors, where the quantile value is in the range $b \in \{5\%, 10\%, 25\%\}$. Using the result from~\cite[Proposition~2.5]{ref:nguyen2020distributionally}, this model can be reformulated as a second-order cone program.
    \item the Optimal Transport based (distributionally robust) Conditional Mean-Variance model (OTCMV) where the portfolio allocation is the solution to problem~\eqref{eq:dro-gamma0} 
    \be 
    \Min{\alpha \in \mc A,~\beta\in \R} \Sup{\QQ \in \mbb B_\rho, \QQ(X = x_0) \ge \eps} ~\EE_{\QQ} [ \ell(Y, \alpha, \beta) | X = x_0].
    \tag{OTCMV}
    \ee
The tuning parameters for this model includes the probability bound $\eps \in \{0.05,0.1,0.15\}$ and the radius $\rho = a\times\rho_{\min}$, where $a\in\{1.05,1.1,1.15\}$ and $\rho_{\min}$ denotes the minimum distance between the training covariate $(\wh x_i)_{i=1}^N$ and $x_0$. This optimization problem is equivalent to a second-order cone program thanks to Proposition~\ref{prop:markowitz-reform}.
\end{enumerate}

\modified{The tuning process and the out-of-sample performance computations are conducted as described previously. We present the average out-of-sample results regarding the average expected loss (exp. loss), portfolio return (mean), CVaR (CVaR) and standard deviation (stdDev) of return, and annualized Sharpe ratio (Sharpe) 
for the optimally-tuned methods in Table \ref{tab:mean-variance-50}. We present five different values of $\eta \in \{1, 3, 5, 7, 9\}$, corresponding to different possible trade-offs between the portfolio mean return and portfolio variance for the portfolio manager.

\begin{table}[!htp]
\caption{Average out-of-sample performance for optimally tuned mean-variance models. Bold values represent the best performance for each value of $\eta$.}\label{tab:mean-variance-50}
\begin{tabular}{llllllll}
\hline
$\eta$             & model & exp. loss $\downarrow$  & mean $\uparrow$ & CVaR  $\downarrow$ & stdDev  $\downarrow$ & Sharpe $\uparrow$ %& maxDraw  $\downarrow$
\\\hline
\multirow{6}{*}{1} & EW    & 1.421 & \textbf{0.037} & 0.605 & 1.162  & 0.735\\%  & 5.246   \\
                   & MV    & 0.885 & 0.021 & 0.571 & 0.921  & 0.534\\%  & 4.383   \\
                   & DRMV  & 0.895 & 0.022 & 0.526 & 0.923  & 0.609\\%  & 4.294   \\
                   & CMV   & \textbf{0.854} & 0.034 & 0.533 & \textbf{0.916}  & \textbf{0.758}\\%  & 4.253   \\
                   & DRCMV & 1.024 & 0.029 & 0.531 & 0.982  & 0.753\\%  & 4.421   \\
                   & OTCMV & 0.870 & 0.028 & \textbf{0.520} & 0.917  & 0.645\\\hline%  & \textbf{4.235}   \\\hline
\multirow{6}{*}{3} & EW    & 1.346 & \textbf{0.037} & 0.605 & 1.162  & 0.735\\%  & 5.246   \\
                   & MV    & 0.891 & 0.019 & 0.566 & 0.940  & 0.516\\%  & 4.433   \\
                   & DRMV  & 0.970 & 0.021 & \textbf{0.508} & 0.975  & 0.617\\%  & 4.394   \\
                   & CMV   & 0.956 & 0.032 & 0.553 & 0.998  & 0.676\\%  & 4.604   \\
                   & DRCMV & 1.005 & 0.029 & 0.540 & 1.003  & \textbf{0.738}\\%  & 4.501   \\
                   & OTCMV & \textbf{0.834} & 0.027 & 0.516 & \textbf{0.926}  & 0.624\\\hline%  & \textbf{4.261}   \\\hline
\multirow{6}{*}{5} & EW    & 1.272 & \textbf{0.037} & 0.605 & 1.162  & 0.735\\%  & 5.246   \\
                   & MV    & 0.933 & 0.017 & 0.567 & 0.971  & 0.503\\%  & 4.550   \\
                   & DRMV  & 1.076 & 0.020 & \textbf{0.512} & 1.041  & 0.616\\%  & 4.613   \\
                   & CMV   & 1.057 & 0.031 & 0.589 & 1.071  & 0.631\\%  & 4.933   \\
                   & DRCMV & 1.051 & 0.030 & 0.546 & 1.043  & \textbf{0.752}\\%  & 4.681   \\
                   & OTCMV & \textbf{0.808} & 0.026 & \textbf{0.512} & \textbf{0.938}  & 0.610 \\\hline% & \textbf{4.289}   \\\hline
\multirow{6}{*}{7} & EW    & 1.197 &\textbf{0.037} & 0.605 & 1.162  & 0.735\\%  & 5.246   \\
                   & MV    & 0.999 & 0.015 & 0.569 & 1.009  & 0.488\\%  & 4.699   \\
                   & DRMV  & 1.188 & 0.020 & 0.540 & 1.104  & 0.599\\%  & 4.847   \\
                   & CMV   & 1.140 & 0.032 & 0.615 & 1.133  & 0.607 \\% & 5.236   \\
                   & DRCMV & 0.989 & 0.030 & 0.543 & 1.042  & \textbf{0.752}\\%  & 4.667   \\
                   & OTCMV & \textbf{0.780} & 0.026 & \textbf{0.513} & \textbf{0.949}  & 0.609\\\hline%  & \textbf{4.325}   \\\hline
\multirow{6}{*}{9} & EW    & 1.122 & \textbf{0.037} & 0.605 & 1.162  & 0.735\\%  & 5.246   \\
                   & MV    & 1.082 & 0.013 & 0.579 & 1.051  & 0.476\\%  & 4.870   \\
                   & DRMV  & 1.163 & 0.023 & 0.551 & 1.119  & 0.621\\%  & 4.896   \\
                   & CMV   & 1.209 & 0.032 & 0.647 & 1.186  & 0.591\\%  & 5.499   \\
                   & DRCMV & 0.928 & 0.030 & 0.541 & 1.041  & \textbf{0.750}\\%  & 4.661   \\
                   & OTCMV & \textbf{0.745} & 0.026 & \textbf{0.515} & \textbf{0.957}  & 0.611\\\hline%  & \textbf{4.352}  \\\hline
\end{tabular}
\end{table}

\modified{Table \ref{tab:mean-variance-50} shows that OTCMV delivers the lowest expected loss compared to the baselines for four values of $\eta \in \{3, 5, 7, 9\}$. Only when $\eta=1$, the CMV baseline has a slightly smaller average expected loss than OTCMV. Compared with the EW baseline, OTCMV reduces by 36.3\%, on average over the set of $\eta$ values, for the average expected loss. 

From Table \ref{tab:mean-variance-50}, we observe that the MV baseline outperforms the DRMV and CMV baselines regarding the expected loss when $\eta \in \{3,5,7,9\}$, and the DRCMV baseline when $\eta \in \{1,3,5\}$. This indicates that incorporating type-$\infty$ Wasserstein ambiguity set and side information into the mean-variance models has a marginal effect on reducing the expected loss. Meanwhile, OTCMV provides better performance than MV in terms of the expected loss. The observations from Table~\ref{tab:mean-variance-50} further confirm the advantage of our proposed OT ambiguity set for the mean-variance portfolio problem. 

To gain further confirmation, we perform a collection of statistical hypothesis tests to compare the expected loss incurred by the OTCMV model with that incurred by the remaining five baselines. Let $\tilde L_{\mathrm{OTCMV}}$ be the random expected loss incurred by the OTCMV model, where the randomness is generated by a combination of the random stock pool, and pseudo-random trading day. Let $\tilde L_{i}$, $i \in  \{\mathrm{EW}, \mathrm{MV}, \mathrm{DRMV}, \mathrm{CMV}, \mathrm{DRCMV}\}$, be the (paired) random expected loss incurred by the competing baselines. Because a lower expected loss value is more desirable, for each  competing baseline $i$, we aim to test the following null hypothesis:
\[
\mathcal H_{0}^i: \text{Probability that $\tilde L_i > \tilde L_{\mathrm{OTCMV}}$ is 50\%},
\]
with the alternative hypothesis
\[
\mathcal H_{+}^i: \text{Probability that $\tilde L_i > \tilde L_{\mathrm{OTCMV}}$ is strictly more than 50\%}.
\]
We use the one-sided sign test to test the above null hypothesis. 
To this end, we form the pairs $(L_{i,t}, L_{\mathrm{OTCMV}, t})$, where $t = 1, \ldots, 252\times 2$ runs over all observed expected loss values collected in the two-year test period, and conduct the hypothesis test using the resulting paired-samples. The $p$-values are reported in the top part of Table~\ref{tab:p_mean-variance_50}. If we choose a significance level of 0.05 for the one-sided sign test with alternative hypothesis $\mathcal H_+$, the null hypothesis is rejected when $\eta\in\{3,5,7,9\}$. This indicates that the out-of-sample expected loss of the OTCMV model tends to be lower than the competing baselines. 

We also conduct the test with the reverse direction of the alternative hypothesis with
\[
\mathcal H_{-}^i: \text{Probability that $\tilde L_i > \tilde L_{\mathrm{OTCMV}}$ is strictly less than 50\%}.
\]
The $p$-values for this direction are reported in the bottom part of Table~\ref{tab:p_mean-variance_50}. If we choose a significance level
of 0.05, we fail to reject the null hypothesis for all competing baselines. Thus, no competing baseline can
outperform OTCMV. Overall, OTCMV appears to be the preferred choice among the mean-variance models.

\begin{table}[!htp]
\caption{$p$-values of the one-sided sign tests for mean-variance models. Bold values indicate that the $p$-value is smaller than 0.05.}\label{tab:p_mean-variance_50}
\begin{threeparttable}
\begin{tabular}{cllllll}
\hline
Alternative hypothesis                     & $\eta$ & EW        & MV       & DRMV     & CMV      & DRCMV    \\\hline
\multirow{5}{*}{$\mathcal H_+$} & 1      & $\bf{6.07\times10^{-114}}$ & 0.43 & 0.14 & 0.61 & $\bf{2.39\times10^{-21}}$ \\
                         & 3      & $\bf{1.17\times10^{-63}}$  & $\bf{2.14\times10^{-3}}$ & $\bf{1.10\times 10^{-15}}$ & $\bf{1.98\times 10^{-14}}$ & $\bf{2.24\times 10^{-12}}$ \\
                         & 5      & $\bf{5.42\times 10^{-37}}$  & $\bf{2.30\times 10^{-8}}$ & $\bf{1.43\times 10^{-24}}$ & $\bf{3.30\times 10^{-20}}$ & $\bf{2.80\times 10^{-9}}$ \\
                         & 7      & $\bf{1.01\times 10^{-26}}$  & $\bf{2.24\times 10^{-12}}$ & $\bf{1.30\times 10^{-27}}$ & $\bf{1.59\times 10^{-22}}$ & $\bf{8.82\times 10^{-6}}$ \\
                         & 9      & $\bf{1.14\times 10^{-16}}$  & $\bf{4.75\times 10^{-15}}$ & $\bf{2.39\times 10^{-21}}$ & $\bf{5.43\times 10^{-25}}$ & $\bf{3.44\times 10^{-4}}$ \\\hline
\multirow{5}{*}{$\mathcal H_-$}  & 1      & 1.00  & 0.61 & 0.88 & 0.43 & 1.00 \\
                         & 3      & 1.00  & 1.00 & 1.00 & 1.00 & 1.00 \\
                         & 5      & 1.00  & 1.00 & 1.00 & 1.00 & 1.00 \\
                         & 7      & 1.00  & 1.00 & 1.00 & 1.00 & 1.00 \\
                         & 9      & 1.00  & 1.00 & 1.00 & 1.00 & 1.00\\\hline
\end{tabular}
\end{threeparttable}
\end{table}

}
\subsection{Mean-CVaR Portfolio Allocation Problems}\label{sec:experiment_mean-cvar}
We next turn to considering the portfolio allocation problem using the mean-CVaR criterion as in Lemma~\ref{lemma:MCVaR}. Below, the loss function~$\ell$ is prescribed in~\eqref{eq:mean-cvar-loss}. We compare the following methods:
\begin{enumerate}[label = (\roman*), leftmargin = 5mm]
    \item the Equal Weighted model (EW);
    \item the unconditional Mean-CVaR model (MC): the portfolio allocation solves
    \be 
    \Min{\alpha \in \mc A,~\beta\in \R} \EE_{\Pnom} [ \ell(Y, \alpha, \beta)]
    \tag{MC}
    \ee
    with the distribution $\Pnom$ is the empirical distribution supported on the available return data. In the experiment, the risk tolerance $\tau = 0.05$ is fixed and also be used to form the objective function in subsequent methods. 
    \item the Distributionally Robust unconditional Mean-CVaR model (DRMC): the portfolio allocation solves
    \be
        \Min{\alpha \in \mc A}~\Sup{\QQ \in \mbb B_\rho} ~\EE_{\QQ} [ \ell(Y, \alpha, \beta)].
        \tag{DRMC}
    \ee
    The inner maximization problem can be reformulated using the results of~\cite{ref:blanchet2019quantifying,ref:esfahani2018data}. The parameter grid for $\rho$ is the same as for the DRMV model.
    \item the Conditional Mean-CVaR model (CMC): the portfolio allocation solves
    \be 
    \Min{\alpha \in \mc A,~\beta\in \R} \EE_{\Pnom} [ \ell(Y, \alpha, \beta)|X \in \mc N_{\gamma}(x_0)].
    \tag{CMC}
    \ee
    The parameter grid for $\gamma$ is the same as for the CMV model.
    
    \item the Distributionally Robust Conditional Mean-CVaR model (DRCMC) with type-$\infty$ Wasserstein ambiguity set \cite{ref:nguyen2020distributionally}, in which the portfolio allocation solves 
    \be 
    \Min{\alpha \in \mc A,~\beta\in \R} \Sup{\QQ \in \mbb B^\infty_\rho, \QQ(X \in \mc N_\gamma(x_0)) > 0} ~\EE_{\QQ} [ \ell(Y, \alpha, \beta) | X \in \mc N_{\gamma}(x_0)].
    \tag{DRCMC}
    \ee
    The parameter grids for $\gamma$ and $\rho$ are the same as for the DRCMV model.
    Using the result from~\cite[Proposition~2.5]{ref:nguyen2020distributionally}, this model can be reformulated as a second-order cone program (see Appendix~\ref{appendix:type-infty} for definition of $\BB_\rho^{\infty}$ and detailed derivation of reformulation).
    \item the Optimal Transport (distributionally robust) Conditional Mean-CVaR model (OTCMC), where the portfolio allocation is the solution to problem~\eqref{eq:dro-gamma0} with the mean-CVaR loss function~\eqref{eq:mean-cvar-loss}
    \be 
    \Min{\alpha \in \mc A,~\beta\in \R} \Sup{\QQ \in \mbb B_\rho, \QQ(X = x_0) \ge \eps} ~\EE_{\QQ} [ \ell(Y, \alpha, \beta) | X = x_0].
    \tag{OTCMC}
    \ee
    The parameter grids for $\eps$ and $\rho$ are the same as for the OTCMV model. 
    This model is equivalent to a second-order cone program thanks to Proposition~\ref{prop:mean-cvar-reform}.
\end{enumerate}

\modified{Table \ref{tab:mean-cvar-50} presents the average out-of-sample results for the optimally tuned mean-CVaR models. }
\begin{table}[!htp]
\caption{Average out-of-sample performance for optimally tuned mean-CVaR models. Bold values represent the best performance for each value of $\eta$.}\label{tab:mean-cvar-50}
\begin{tabular}{llllllll}
\hline
$\eta$             & model &exp. loss $\downarrow$ & mean $\uparrow$ & CVaR $\downarrow$ & stdDev $\downarrow$ & Sharpe $\uparrow$ \\\hline%& maxDraw $\downarrow$ \\\hline
\multirow{6}{*}{1} & EW    & 0.568 & \textbf{0.037} & 0.605 & 1.162  & \textbf{0.735}\\%  & 5.246   \\
                   & MC    & 0.564 & 0.020 & 0.583 & 0.966  & 0.459\\%  & 4.588   \\
                   & DRMC  & 0.508 & 0.024 & 0.531 & \textbf{0.936}  & 0.599\\%  & \textbf{4.360}   \\
                   & CMC   & 0.521 & 0.026 & 0.547 & 0.994  & 0.559\\%  & 4.619   \\
                   & DRCMC & \textbf{0.492} & 0.026 & \textbf{0.518} & 0.993  & 0.678\\%  & 4.431   \\
                   & OTCMC & 0.495 & 0.028 & 0.523 & 1.021  & 0.679\\\hline%  & 4.579   \\\hline
\multirow{6}{*}{3} & EW    & 0.493 & \textbf{0.037} & 0.605 & 1.162  & \textbf{0.735} \\% & 5.246   \\
                   & MC    & 0.526 & 0.018 & 0.579 & 0.983  & 0.448\\%  & 4.655   \\
                   & DRMC  & 0.469 & 0.021 & 0.531 & \textbf{0.950}  & 0.558\\%  & \textbf{4.403}   \\
                   & CMC   & 0.510 & 0.027 & 0.591 & 1.059  & 0.558\\%  & 4.912   \\
                   & DRCMC & 0.451 & 0.028 & 0.535 & 1.019  & 0.715\\%  & 4.560   \\
                   & OTCMC & \textbf{0.434} & 0.028 & \textbf{0.518} & 0.992  & 0.679\\\hline%  & 4.441   \\\hline
\multirow{6}{*}{5} & EW    & 0.418 & \textbf{0.037} & 0.605 & 1.162  & \textbf{0.735}\\% & 5.246   \\
                   & MC    & 0.508 & 0.015 & 0.582 & 1.012  & 0.418\\%  & 4.769   \\
                   & DRMC  & 0.430 & 0.019 & 0.527 & \textbf{0.969}  & 0.549 \\% & 4.450   \\
                   & CMC   & 0.496 & 0.029 & 0.640 & 1.129  & 0.550 \\% & 5.259   \\
                   & DRCMC & 0.401 & 0.029 & 0.546 & 1.035  & 0.726 \\% & 4.630   \\
                   & OTCMC & \textbf{0.385} & 0.028 & \textbf{0.524} & 0.994  & 0.688\\\hline%  & \textbf{4.444}   \\\hline
\multirow{6}{*}{7} & EW    & 0.344 & \textbf{0.037} & 0.605 & 1.162  & \textbf{0.735}\\%  & 5.246   \\
                   & MC    & 0.518 & 0.012 & 0.603 & 1.055  & 0.395 \\% & 4.958   \\
                   & DRMC  & 0.405 & 0.018 & \textbf{0.528} & \textbf{0.996}  & 0.542\\%  & 4.525   \\
                   & CMC   & 0.468 & 0.032 & 0.690 & 1.208  & 0.554 \\% & 5.673   \\
                   & DRCMC & \textbf{0.343} & 0.030 & 0.551 & 1.045  & 0.727 \\% & 4.676   \\
                   & OTCMC & 0.347 & 0.027 & 0.537 & 1.008  & 0.677 \\\hline% & \textbf{4.506}   \\\hline
\multirow{6}{*}{9} & EW    & \textbf{0.269} & \textbf{0.037} & 0.605 & 1.162  & \textbf{0.735}\\%  & 5.246   \\
                   & MC    & 0.552 & 0.010 & 0.639 & 1.116  & 0.378\\%  & 5.249   \\
                   & DRMC  & 0.376 & 0.017 & \textbf{0.527} & 1.028  & 0.545\\%  & 4.630   \\
                   & CMC   & 0.420 & 0.035 & 0.731 & 1.279  & 0.572\\%  & 6.046   \\
                   & DRCMC & 0.283 & 0.030 & 0.553 & 1.052  & 0.726\\%  & 4.708   \\
                   & OTCMC & 0.287 & 0.027 & 0.533 & \textbf{1.012}  & 0.675 \\\hline% & \textbf{4.522}  \\\hline
\end{tabular}
\end{table}

Table~\ref{tab:mean-cvar-50} shows that OTCMC has the lowest average expected loss when $\eta\in \{3,5\}$, and competitive average expected loss for the other three $\eta$ values. 
Moreover, the DRMC and CMC baselines outperform the MC baseline regarding the expected loss; this implies that incorporating distributionally robustification or side information improves the vanilla MC baseline. However, DRCMC and OTCMC deliver better performance than DRMC and CMC along these metrics. These observations imply that incorporating distributionally robustification \textit{and} side information further improves performance. Moreover, compared with EW, OTCMC reduces by 5.0\% on average for the average expected loss.

We perform a collection of hypothesis tests described in Section \ref{sec:experiment_mean-variance} for the mean-CVaR results. The $p$-values for the tests are reported in Table~\ref{tab:p_mean-cvar_50}. If we choose a significance level of 0.05 for the one-sided sign test with alternative hypothesis $\mathcal H_+$,  the null hypothesis is rejected for 40\% of instances. This indicates a reasonable tendency for OTCMC to deliver lower expected loss than the competing baselines. If we choose a significance level of 0.05, we fail to reject the null hypothesis in favor of the alternative $\mathcal{H}_-$ for all competing baselines. Thus, no competing baseline is found to outperform OTCMC in a statistically significant way. Overall, OTCMC appears to be the preferred choice among the mean-CVaR models.
\begin{table}[!htp]
\caption{$p$-values (round to two decimal numbers) of the one-sided sign tests for mean-CVaR models. Bold values indicate that the $p$-value is smaller than 0.05.}\label{tab:p_mean-cvar_50}
\begin{threeparttable}
\begin{tabular}{cllllll}
\hline
Alternative hypothesis                    & $\eta$ & EW       & MC       & DRMC     & CMC      & DRCMC    \\\hline
\multirow{5}{*}{$\mathcal H_+$} & 1      & $\bf{5.07\times10^{-11}}$ & 0.12 & 0.57 & 0.30 & 0.61 \\
                         & 3      & $\bf{2.87\times10^{-5}}$ & $\bf{3.72\times10^{-2}}$ & 0.16 & 0.09 & 0.05 \\
                         & 5      & $\bf{8.87\times10^{-4}}$ & $\bf{1.20\times10^{-3}}$ & 0.21 & $\bf{3.44\times10^{-4}}$ & \textbf{0.01} \\
                         & 7      & 0.05 & $\bf{1.05\times10^{-6}}$ & 0.09 & 0.16 & 0.27 \\
                         & 9      & 0.21 & $\bf{3.86\times10^{-6}}$ & $\bf{7.98\times10^{-3}}$ & 0.19 & 0.12 \\\hline
\multirow{5}{*}{$\mathcal H_-$}  &1	&1.00	&0.89	&0.46	&0.73	&0.43\\
&3	&1.00	&0.97	&0.86	&0.92	&0.95\\
&5	&1.00	&1.00	&0.81	&1.00	&0.99\\
&7	&0.95	&1.00	&0.92	&0.86	&0.76\\
&9	&0.81	&1.00	&0.99	&0.84	&0.89\\\hline
\end{tabular}
\end{threeparttable}
\end{table}

}

\subsection{\modified{Summary of Empirical Findings}}

\modified{The empirical experiments presented above indicate that OTCMV and OTCMV can significantly impact the out-of-sample performance of portfolios produced in an environment where side information can be exploited. As with any machine/statistical learning method, the performance of OTCMV and OTCMC will depend on the nature of the application, the quality of the data, and the choice of side information that is used. We expect that additional empirical evaluation with other data and application environments would benefit our understanding of the value of our proposed OTCMV and OTCMC approaches.}

\paragraph{\bf Acknowledgments.}
The material in this paper is based upon work supported by the Air Force Office of Scientific Research under award number FA9550-20-1-0397. Additional support is gratefully acknowledged from NSF grants 1915967, 1820942, 1838676, NSERC grant  RGPIN-2016-05208, and the China Merchant Bank. Finally, this research was enabled in part by support provided by Compute Canada.

\appendix

\section{Proofs of Main Results}

\subsection{Proofs of Section~\ref{sec:setup}}

The following results are needed to justify Lemmas~\ref{lemma:MV} and~\ref{lemma:MCVaR}. 

\begin{lemma}[Interchange] \label{lemma:interchange}
Suppose that $\mc Y$ and $\mc B$ are compact and that $\ell$ is continuous in $Y$ and convex in $\beta$. Then we have
\[
    \Sup{\QQ \in \mbb B_\rho, \QQ(X \in \mc N_\gamma(x_0)) \ge \eps}\Inf{\beta \in \mc B}~\EE_{\QQ}[ \ell(Y, \alpha, \beta)|X\in \mc N_\gamma(x_0)] =\Inf{\beta \in \mc B} \Sup{\QQ \in \mbb B_\rho, \QQ(X \in \mc N_\gamma(x_0)) \ge \eps}~\EE_{\QQ}[ \ell(Y, \alpha, \beta)|X\in \mc N_\gamma(x_0)].
\]
\end{lemma}

A similar result on the interchange of the infimum and the supremum operators is obtained in~\cite[Lemma~B.3]{ref:nguyen2020distributionally}. Nevertheless, there is a subtle distinction in the conditions of the two results: Lemma~\ref{lemma:interchange} requires that $\mathcal{Y}$ and $\mathcal{B}$ are compact, while~\cite[Lemma~B.3]{ref:nguyen2020distributionally} does not require these compactness conditions. This distinction emphasizes the qualitative difference between an $\infty$-Wasserstein ambiguity set of~\cite{ref:nguyen2020distributionally} and the optimal transport ambiguity set $\mbb B_\rho$ in this paper. Indeed, the weak-compactness condition is automatically satisfied by the $\infty$-Wasserstein set in~\cite{ref:nguyen2020distributionally}. Still, this condition is not satisfied by the set $\mbb B_\rho$, which leads to the additional requirement on the compactness of $\mathcal B$ in Lemma~\ref{lemma:interchange}. Further, Lemma~\ref{lemma:interchange} requires the compactness of $\mathcal{Y}$ for the upper-continuity of the objective function in $\mathbb{Q}$. On the contrary, the continuity condition is automatically satisfied by the $\infty$-Wasserstein set in~\cite{ref:nguyen2020distributionally}.

The proof of Lemma~\ref{lemma:interchange} follows from the convexity of the conditional ambiguity set. To this end, define the following ambiguity set
\[
\mc B_{x_0, \gamma, \eps}(\mbb B_\rho) \Let \left\{ \mu_{x_0} \in \mc M(\mc Y): \begin{array}{l}
	\exists \QQ \in \mbb B_\rho \text{ such that } \QQ(\mc N_\gamma(x_0) \times\mc Y) \ge \eps \\
	~\QQ(\mc N_\gamma(x_0)  \times A ) = \mu_{x_0}(A) \QQ(\mc N_\gamma(x_0)  \times \mc Y)~~\forall A  \subseteq \mc Y~\text{ measurable}
	\end{array}
	\right\}.
\]

\begin{lemma}[Convexity of $\mc B_{x_0, \eps, \gamma}(\mbb B_\rho)$] \label{lemma:convexity}
     If $\mbb B_\rho$ is convex, then the conditional ambiguity set $\mc B_{x_0, \gamma, \eps}(\mbb B_\rho)$ is also convex. 
\end{lemma}

The proof of Lemma~\ref{lemma:convexity} can be obtained by a minor modification of the proof for~\cite[Lemma~B.5]{ref:nguyen2020distributionally}. The proof is included here for completeness.

\begin{proof}[Proof of Lemma~\ref{lemma:convexity}]
    Let $\mu_0^0,~\mu_0^1 \in \mc B_{x_0, \gamma, \eps}(\mbb B_\rho)$ be two arbitrary probability measures supported on $\mc Y$. Associated with each $\mu_0^j$, $j \in \{0, 1\}$, is a corresponding joint measure $\QQ^j \in \mc M(\mc X \times \mc Y)$ such that
	    \[
	        \QQ^j(\mc N_\gamma(x_0) \times \mc Y) \ge \eps \quad \text{and} \quad 
		        \ds\frac{\QQ^j(\mc N_\gamma(x_0) \times A)}{\QQ^j(\mc N_\gamma(x_0) \times \mc Y)} = \mu_0^j(A) \quad \forall A \subseteq \mc Y~\text{measurable}.
	    \]
	    Select any $\lambda \in (0, 1)$. We now show that $\mu_0^\lambda = \lambda \mu_0^1 + (1-\lambda) \mu_0^0 \in \mc B_{x_0, \gamma, \eps}(\mbb B_\rho)$. Indeed, consider the joint measure 
	    \[
	        \QQ^\lambda = \theta \QQ^1 + (1-\theta) \QQ^0,
	    \]
	    where $\theta$ is defined as
	    \[
	        \theta = \frac{\lambda \QQ^0(\mc N_\gamma(x_0) \times \mc Y)}{\lambda \QQ^0(\mc N_\gamma(x_0) \times \mc Y) + (1-\lambda) \QQ^1(\mc N_\gamma(x_0) \times \mc Y)} \in [0, 1].
	    \]
	    By definition, we have $\QQ^\lambda(\mc N_\gamma(x_0) \times \mc Y) \ge \eps$, and by the convexity of $\mbb B_\rho$, we have $\QQ^\lambda \in \mbb B_\rho$. Moreover, for any set $A \subseteq \mc Y$ measurable, we find
	    \begin{align*}
	        \frac{\QQ^\lambda(\mc N_\gamma(x_0) \times A)}{\QQ^\lambda(\mc N_\gamma(x_0) \times \mc Y)} &= \frac{\theta \QQ^1(\mc N_\gamma(x_0) \times A) + (1-\theta) \QQ^0(\mc N_\gamma(x_0) \times A)}{\theta \QQ^1(\mc N_\gamma(x_0) \times \mc Y) + (1-\theta) \QQ^0(\mc N_\gamma(x_0) \times \mc Y)} \\
	        &= \frac{\lambda \QQ^0(\mc N_\gamma(x_0) \times \mc Y) \QQ^1(\mc N_\gamma(x_0) \times A) + (1-\lambda) \QQ^1(\mc N_\gamma(x_0) \times \mc Y) \QQ^0(\mc N_\gamma(x_0) \times A)}{\QQ^0(\mc N_\gamma(x_0) \times \mc Y) \QQ^1(\mc N_\gamma(x_0) \times \mc Y)} \\
	        &= \frac{\lambda \QQ^1(\mc N_\gamma(x_0) \times A)}{\QQ^1(\mc N_\gamma(x_0) \times \mc Y)} + \frac{(1-\lambda) \QQ^0(\mc N_\gamma(x_0) \times A)}{\QQ^0(\mc N_\gamma(x_0) \times \mc Y)} \\
	        &= \lambda \mu_0^1(A) + (1-\lambda) \mu_0^0(A),
	    \end{align*}
	   where the second equality holds thanks to the definition of $\theta$. This line of argument implies that $\mu_0^\lambda \in \mc B_{x_0, \gamma, \eps}(\mbb B_\rho)$, and further asserts the convexity of $\mc B_{x_0, \gamma, \eps}(\mbb B_\rho)$. 
\end{proof}

We are now ready to prove Lemma~\ref{lemma:interchange}.
\begin{proof}[Proof of Lemma~\ref{lemma:interchange}]
    By rewriting the conditional expectation using the conditional measure $\muxo$, we have for any value of $\alpha$
    \[
        \Sup{\QQ \in \mbb B_\rho, \QQ(X \in \mc N_\gamma(x_0)) \ge \eps}\Inf{\beta \in \mc B}~\EE_{\QQ}[ \ell(Y, \alpha, \beta)|X\in \mc N_\gamma(x_0)] = \Sup{\muxo \in \mc B_{x_0, \gamma, \eps}(\mbb B_\rho)}~\Inf{\beta \in \mc B}~\EE_{\muxo}[\ell(Y, \alpha, \beta)].
    \]
     The set $\mc B_{x_0, \gamma, \eps}(\mbb B_\rho)$ is convex by Lemma~\ref{lemma:convexity}. 
    Moreover, because $\ell$ is continuous in $Y$, the mapping $\muxo \mapsto \EE_{\muxo}[\ell(Y, \alpha, \beta)]$ is upper semicontinuous in the weak topology thanks to the compactness of $\mc Y$. Finally, $\mc B$ is compact and $\ell$ is convex in $\beta$. The interchangeability of the supremum and the infimum operators is now a consequence of Sion's minimax theorem~\cite{ref:sion1958minimax}. 
\end{proof}

We are ready to prove the results in Section~\ref{sec:setup}.

\begin{proof}[Proof of Lemma~\ref{lemma:MV}]
It is well known that 
\[
\mathrm{Variance}_\QQ[Y^\top \alpha| X \in \mc N_\gamma(x_0)] = \Min{\beta \in \R}~\EE_{\QQ}[ (Y^\top \alpha - \beta)^2 | X \in \mc N_{\gamma}(x_0)]
\]
for any probability measure $\QQ \in \mc M(\mc X \times \mc Y)$, where the optimal $\beta$ is the conditional mean of $Y^\top \alpha$ given $X \in \mc N_{\gamma}(x_0)$ under $\QQ$. When $\mc A$ and $\mc Y$ are compact, the random variable $Y^\top \alpha$ has a bounded mean for any $\QQ \in \mbb B_\rho$. More precisely, we have
\[
    \EE_{\QQ}[Y^\top \alpha | X \in \mc N_\gamma(x_0)] \in \mc B,
\]
where $\mc B$ is defined as in the statement of the lemma. Therefore, it is without any loss of optimality to restrict $\beta \in \mc B$. We thus find
\begin{align*}
    &\Min{\alpha \in \mc A}~\Sup{\QQ \in \mbb B_\rho, \QQ(X \in \mc N_\gamma(x_0)) \ge \eps}~\mathrm{Variance}_\QQ[Y^\top \alpha| X \in \mc N_\gamma(x_0)] - \eta\cdot \EE_\QQ[ Y^\top\alpha| X \in \mc N_\gamma(x_0)] \\
    =&\Min{\alpha \in \mc A}~\Sup{\QQ \in \mbb B_\rho, \QQ(X \in \mc N_\gamma(x_0)) \ge \eps}~\min_{\beta \in \R}~ \EE_{\QQ} [\ell(Y, \alpha, \beta) | X \in \mc N_{\gamma}(x_0)] \\
    =&\Min{\alpha \in \mc A}~\Sup{\QQ \in \mbb B_\rho, \QQ(X \in \mc N_\gamma(x_0)) \ge \eps}~\min_{\beta \in \mc B}~ \EE_{\QQ} [\ell(Y, \alpha, \beta) | X \in \mc N_{\gamma}(x_0)] \\
    =&\Min{\alpha \in \mc A}~\min_{\beta \in \mc B}~\Sup{\QQ \in \mbb B_\rho, \QQ(X \in \mc N_\gamma(x_0)) \ge \eps}~\EE_{\QQ} [\ell(Y, \alpha, \beta) | X \in \mc N_{\gamma}(x_0)],
\end{align*}
where the last equality follows from Lemma~\ref{lemma:interchange}. This completes the proof.
\end{proof}

\begin{proof}[Proof of Lemma~\ref{lemma:MCVaR}]
It is well known that 
\[
\mathrm{CVaR}^{1-\tau}_\QQ[Y^\top \alpha|X \in \mc N_\gamma(x_0)] = \Min{\beta \in \R}~\EE_{\QQ}[\beta+ \frac{1}{\tau} (-Y^\top \alpha - \beta)^+ | X \in \mc N_{\gamma}(x_0)]
\]
for any probability measure $\QQ \in \mc M(\mc X \times \mc Y)$, where the optimal $\beta$ is the conditional $(1-\tau)$-quantile of $Y^\top \alpha$ given $X \in \mc N_{\gamma}(x_0)$ under $\QQ$. When $\mc A$ and $\mc Y$ are compact, the random variable $Y^\top \alpha$ has uniformly bounded support $\mc B$ for any $\QQ \in \mbb B_\rho$ and $\alpha \in \mc A$, where $\mc B$ is defined as in the statement of the lemma. Therefore, it is without any loss of optimality to restrict $\beta \in \mc B$. The rest of the proof follows the same argument in the proof of Lemma~\ref{lemma:MV}.
\end{proof}

The proofs of Lemmas~\ref{lemma:MV} and~\ref{lemma:MCVaR} rely on a strong duality result. If strong duality fails, one can still resort to weak duality to obtain a conservative approximation of the decision problem. This fact is highlighted in the next remark.
\begin{remark}[Conservative approximation under weak duality] \label{remark:weak-duality}
    If the conditions of Lemma~\ref{lemma:interchange} do not hold, then the robustified conditional mean-variance portfolio problem still admits
    \begin{align*}
    &\Min{\alpha \in \mc A}~\Sup{\QQ \in \mbb B_\rho, \QQ(X \in \mc N_\gamma(x_0)) \ge \eps}~\mathrm{Variance}_\QQ[Y^\top \alpha| X \in \mc N_\gamma(x_0)] - \eta\cdot \EE_\QQ[ Y^\top\alpha| X \in \mc N_\gamma(x_0)] \\
    =&\Min{\alpha \in \mc A}~\Sup{\QQ \in \mbb B_\rho, \QQ(X \in \mc N_\gamma(x_0)) \ge \eps}~\min_{\beta \in \R}~ \EE_{\QQ} [\ell(Y, \alpha, \beta) | X \in \mc N_{\gamma}(x_0)] \\
    =&\Min{\alpha \in \mc A}~\Sup{\QQ \in \mbb B_\rho, \QQ(X \in \mc N_\gamma(x_0)) \ge \eps}~\min_{\beta \in \mc B}~ \EE_{\QQ} [\ell(Y, \alpha, \beta) | X \in \mc N_{\gamma}(x_0)] \\
    \le &\Min{\alpha \in \mc A}~\min_{\beta \in \mc B}~\Sup{\QQ \in \mbb B_\rho, \QQ(X \in \mc N_\gamma(x_0)) \ge \eps}~\EE_{\QQ} [\ell(Y, \alpha, \beta) | X \in \mc N_{\gamma}(x_0)],
\end{align*}
where the inequality is from weak duality. Thus, the ultimate min-max problem is a conservative approximation of the robustified conditional mean-variance portfolio problem. A similar argument also holds for the mean-CVaR problem.
\end{remark}

\begin{proof}[Proof of Proposition~\ref{prop:rho-LB-1}] 
Using the definition of the optimal transport cost, we can compute $\rho_{\min}(x_0, \gamma, \eps)$ as
\[
    \rho_{\min}(x_0, \gamma, \eps) = \left\{
	\begin{array}{ll}	
		\inf & \ds \int_{(\mc X \times \mc Y) \times (\mc X \times \mc Y)} [\DD_{\mc X}(x, x') + \DD_{\mc Y}(y, y')] 
		~\pi\big((\mathrm{d}x \times \mathrm{d}y), (\mathrm{d} x' \times \mathrm{d}y')\big)\\
        \st & \QQ \in \mc M(\mc X \times\mc Y),~\pi \in \Pi(\QQ, \Pnom) \\
			& \ds \int_{(\mc X \times \mc Y) \times (\mc X \times \mc Y)} 
		    \mathbbm{1}_{\mc N_\gamma(x_0)} (x)
		    ~\pi\big((\mathrm{d}x \times \mathrm{d}y), (\mathrm{d} x' \times \mathrm{d}y')\big) \ge \eps.
	\end{array} 
	\right.
\]
Because $\Pnom$ is an empirical measure, any joint probability measure $\pi \in \Pi(\QQ, \Pnom)$ can be written as $\pi = N^{-1} \sum_{i \in [N]} \pi_i \otimes \delta_{(\wh x_i, \wh y_i)}$ using the collection of probability measures $\{\pi_i\}_{i \in [N]}$, and $\otimes$ denotes the Kronecker product of two probability measures. One thus can reformulate $\rho_{\min}(x_0, \gamma, \eps)$ as
\begin{align} \label{eq:rho-min}
    \rho_{\min}(x_0, \gamma, \eps) &= \left\{
        \begin{array}{cl}
            \inf & \ds N^{-1} \sum_{i \in [N]} \int_{\mc X \times \mc Y} [\DD_{\mc X}(x, \wh x_i) + \DD_{\mc Y}(y, \wh y_i)] 
		    ~\pi_i(\dd x \times \dd y) \\
		    \st & \pi_i \in \mc M(\mc X \times \mc Y) ~~\forall i \in [N],~\ds \sum_{i \in [N]} \pi_i(\mc N_\gamma(x_0) \times \mc Y) \ge N \eps.
        \end{array}
    \right. 
\end{align}
Let $\{\pi_i\opt\}_{i \in [N]}$ be an optimal solution of the above optimization problem. We now show that $\pi_i\opt$ should be of the form
\[
    \pi_i\opt = \upsilon_i \delta_{(\wh x_i^p, \wh y_i)} + (1 - \upsilon_i) \delta_{(\wh x_i, \wh y_i)}
\]
for some $\upsilon_i \in [0, 1]$ and where $\wh x_i^p$ is the projection of $\wh x_i$ onto $\mc N_{\gamma}(x_0)$ defined in~\eqref{eq:kappa-def}. Suppose otherwise, then denote $P_i = \pi_i\opt(\mc N_\gamma(x_0) \times \mc Y)$. Consider now 
\[
    \pi_i' = P_i \delta_{(\wh x_i^p, \wh y_i)} + (1 - P_i) \delta_{(\wh x_i, \wh y_i)}.
\]
It is trivial that $\sum_{i\in[N]} \pi_i'(\mc N_\gamma(x_0) \times \mc Y) \ge \sum_{i \in [N]} P_i = \sum_{i \in [N]} \pi_i\opt(\mc N_\gamma(x_0) \times \mc Y) \ge N \eps$. Moreover, we have
\begin{align*}
    \sum_{i \in [N]} \int_{\mc X \times \mc Y} [\DD_{\mc X}(x, \wh x_i) + \DD_{\mc Y}(y, \wh y_i)] ~\pi_i'(\dd x \times \dd y) &= \sum_{i \in [N]} [\DD_{\mc X}(\wh x_i^p, \wh x_i) + \DD_{\mc Y}(\wh y_i, \wh y_i)] P_i \\
    &\le \sum_{i \in [N]} \int_{\mc X \times \mc Y} [\DD_{\mc X}(x, \wh x_i) + \DD_{\mc Y}(y, \wh y_i)] ~\pi_i\opt(\dd x \times \dd y),
\end{align*}
which implies that $\{\pi_i'\}_{i \in [N]}$ is at least as good as $\{ \pi_i\opt\}_{i \in [N]}$ in the optimization problem~\eqref{eq:rho-min}. 

Next, by restricting the decision variables to $\pi_i = \upsilon_i \delta_{(\wh x_i^p, \wh y_i)} + (1-\upsilon_i) \delta_{(\wh x_i, \wh y_i)}$, one now can consider the equivalent reformulation
\begin{align*}
    \rho_{\min}(x_0, \gamma, \eps) &= \inf \left\{ \ds N^{-1} \sum_{i \in [N]} \kappa_i \upsilon_i
       : \upsilon \in [0, 1]^N,~\ds \sum_{i \in [N]} \upsilon_i \ge N \eps
    \right\}. 
\end{align*}
Because the feasible set is compact and the objective function is continuous, the minimization operator is justified  thanks to Weierstrass' maximum value theorem~\cite[Corollary~2.35]{ref:aliprantis06hitchhiker}.

It remains to show the existence of a measure $\QQ\opt \in \mbb B_\rho$ with $\QQ\opt(X \in \mc N_\gamma(x_0)) \ge \eps$ if and only if $\rho \ge \rho_{\min}(x_0, \gamma, \eps)$. The definition of $\rho_{\min}(x_0, \gamma, \eps)$ immediately proves the ``only if'' part, because \eqref{eq:rho-LB-1} implies that
\[
\BB_{\rho} \cap
\Big\{
    \QQ \in \mc M (\mc X \times \mc Y): 
    ~\QQ(X \in \mc N_\gamma(x_0)) \ge \eps
\Big\} = \emptyset
\quad \forall \rho < \rho_{\min}(x_0, \gamma, \eps).
\]
Now we prove the ``if'' part. Given a minimizer $\{\upsilon_i\opt\}_{i \in [N]}$ that solves~\eqref{eq:greedy-1}, there exists a collection of probability measures $\{\pi_i\opt\}_{i \in [N]}$ defined as $\pi_i\opt \Let \upsilon_i\opt \delta_{(\wh x_i^p, \wh y_i)} + (1 - \upsilon_i\opt) \delta_{(\wh x_i, \wh y_i)}$ such that 
\[
    \int_{\mc X \times \mc Y} [\DD_{\mc X}(x, \wh x_i) + \DD_{\mc Y}(y, \wh y_i)] ~ \pi_i\opt(\dd x \times \dd y) = \kappa_i \upsilon_i\opt.
\]
The probability measure $\QQ\opt \Let N^{-1} \sum_{i \in [N]} \pi_i\opt$ satisfies
\[
    \Wass(\QQ\opt, \Pnom) \leq \ds \frac{1}{N} \sum_{i \in [N]} \int_{\mc X \times \mc Y} [\DD_{\mc X}(x, \wh x_i) + \DD_{\mc Y}(y, \wh y_i)] 
		    ~\pi_i\opt(\dd x \times \dd y) = \frac{1}{N} \sum_{i \in [N]} \kappa_i \upsilon_i\opt =  \rho_{\min}(x_0, \gamma, \eps) \leq \rho,
\]
where the last equality is from the optimality of $\upsilon\opt$. This implies that $\QQ\opt \in \mbb B_\rho$ and completes the proof.
\end{proof}

\subsection{Proofs of Section~\ref{sec:gamma0}}
\label{appendix:sec3}

We first collect the fundamental results that facilitate the proofs of Section~\ref{sec:gamma0}. For any $N \in \mbb N$, given some $c\in\R^N$, $d\in\R_+^N$ and $\eps \in (0, 1)$, consider the following two optimization problems
\begin{subequations}
\be \label{eq:equivalent1}
\left\{
\begin{array}{cl}
\Sup{\upsilon \in [0, 1]^N}& \ds \frac{\sum_i \upsilon_i c_i}{N\sum_i \upsilon_i} \\
\st 
&\ds \sum_{i} \upsilon_i \geq \eps N,~\sum_{i} \upsilon_i d_i \leq \rho
\end{array} \right.
\ee
and
\be \label{eq:equivalent2}
\left\{
\begin{array}{cl}
\Sup{\upsilon\in [0, 1]^N } & \ds \frac{\sum_i \upsilon_i c_i}{N\sum_i \upsilon_i} \\
\st 
&\ds \sum_i \upsilon_i = \eps N,~\sum_i \upsilon_i d_i \leq \rho, 
\end{array}
\right.
\ee
\end{subequations}
where the summations are taken over $i \in [N]$. Notice that the summation constraint of $\upsilon_i$ in~\eqref{eq:equivalent1} is an \textit{in}equality constraint, while in~\eqref{eq:equivalent2} it is an equality constraint. The next result asserts that the inequality in~\eqref{eq:equivalent1} can be strengthened to an equality as in~~\eqref{eq:equivalent2} without any loss of optimality.

\begin{lemma}\label{lemma:alphaProbSimplifies}
The optimal values of problems~\eqref{eq:equivalent1} and~\eqref{eq:equivalent2} are equal.
\end{lemma}
\begin{proof}[Proof of Lemma~\ref{lemma:alphaProbSimplifies}]
Let $\bar{\upsilon} \in [0, 1]^N$ be such that $\sum_i \bar{\upsilon}_i > \eps N$ and $\sum_i \bar{\upsilon}_i d_i \le \rho$. One can construct $\upsilon' = (\eps N/\sum_i \bar{\upsilon}_i)\bar{\upsilon}$ which satisfies $\upsilon' \in [0, 1]^N$, $\sum_i \upsilon_i'=\eps N$, and $\sum_i \upsilon_i' d_i \leq \rho$ since $d\geq 0$. Furthermore, it reaches the same objective value
\[
\frac{1}{N\sum_i \upsilon_i'} \sum_i \upsilon_i' c_i = \frac{1}{\eps N^2} \sum_i \frac{\eps N}{\sum_i \bar{\upsilon}_i}\bar{\upsilon}_i c_i = \frac{1}{N\sum_i \bar{\upsilon}_i} \sum_i \bar{\upsilon}_i c_i,
\]
which finishes the proof.
\end{proof}

\begin{proposition}[Equivalent representation] \label{prop:equiv}
Given $\eps \in (0, 1]$ and $\rho > \rho_{\min}(x_0, 0, \eps)$, then for any feasible solution $(\alpha, \beta)$,
we have
\begin{align*}
\Sup{\QQ \in \mbb B_\rho, \QQ(X = x_0) \ge \eps}~ \EE_{\QQ}[ \ell(Y, \alpha, \beta) | X = x_0 ] =
\left\{
\begin{array}{cl}
\sup & \ds (N\eps)^{-1} \sum_{i \in [N]} \upsilon_i \EE_{\muxo^i}[\ell(Y, \alpha, \beta)]\\
\st&\upsilon\in [0,\,1]^N, \; \muxo^i \in \mc M(\mc Y) \qquad \forall i \in [N]\\
& \ds \sum_{i \in [N]} \upsilon_i =  N \eps \\
& \ds \sum_{i \in [N]}  \upsilon_i \big( \DD_{\mc X}(x_0, \wh x_i) + \EE_{\muxo^i}[ \DD_{\mc Y}(Y, \wh y_i)] \big) \leq N \rho.
\end{array}
\right.
\end{align*}
\end{proposition}

\begin{proof}[Proof of Proposition~\ref{prop:equiv}]
    By exploiting the definition of the optimal transport cost and the fact that any joint probability measure $\pi \in \Pi(\QQ, \Pnom)$ can be written as $\pi = N^{-1} \sum_{i \in [N]}\pi_i \otimes \delta_{(\wh x_i, \wh y_i)}$, where each $\pi_i$ is a probability measure on $\mc X \times \mc Y$, we have
    \begin{align*}
        \Sup{\QQ \in \mbb B_\rho, \QQ(X = x_0) \ge \eps}~ \EE_{\QQ}[ \ell(Y, \alpha, \beta) | X = x_0 ] = &\left\{
            \begin{array}{cl}
                \sup & \EE_{\QQ}[ \ell(Y, \alpha, \beta) | X = x_0 ] \\
                \st & \pi_i \in \mc M(\mc X \times \mc Y) \quad \forall i \in [N] \\
                & \sum_{i \in [N]}\pi_i(\{x_0\} \times \mc Y) \ge N\eps \\
                & \QQ = N^{-1} \sum_{i \in [N]} \pi_i, ~\sum_{i \in [N]} \Wass(\pi_i, \delta_{(\wh x_i, \wh y_i)}) \le N\rho
            \end{array}
        \right. \\
        = &\left\{
            \begin{array}{ccll}
                \Sup{\upsilon \in \mc U} &\sup & \EE_{\QQ}[ \ell(Y, \alpha, \beta) | X = x_0 ] \\
                &\st & \pi_i \in \mc M(\mc X \times \mc Y) & \forall i \in [N] \\
                && \pi_i(\{x_0\} \times \mc Y) =\upsilon_i & \forall i \in [N]\\
                && \QQ = N^{-1} \sum_{i \in [N]} \pi_i, ~\sum_{i \in [N]} \Wass(\pi_i, \delta_{(\wh x_i, \wh y_i)}) \le N\rho,
            \end{array}
        \right.
    \end{align*}
    where the set $\mc U$ is defined as
    \[
        \mc U \Let \left\{ \upsilon\in [0, 1]^N:  \sum_{i \in [N]}~\upsilon_i \ge N \eps \right\}.
    \]
    Define the following two functions $g, h: \mc A \to \R$ as
    \be \label{eq:g-def}
        g(\upsilon) \Let \left\{
            \begin{array}{cll}
                \sup & \ds (\sum_{i\in[N]} \upsilon_i)^{-1} \sum_{i \in [N]} \int_{\mc Y} \ell(y, \alpha, \beta) \pi_i(\{x_0\} \times \dd y) \\
                 \st & \pi_i \in \mc M(\mc X \times \mc Y) & \forall i \in [N] \\
                 &\pi_i(\{x_0\} \times \mc Y) = \upsilon_i & \forall i \in [N]\\
                & \ds \sum_{i \in [N]} \Wass(\pi_i, \delta_{(\wh x_i, \wh y_i)}) \le N\rho
            \end{array}
        \right.
    \ee
    and
    \be \label{eq:h-def}
        h(\upsilon) \Let \left\{
            \begin{array}{cll}
                \sup & \ds (\sum_{i\in[N]} \upsilon_i)^{-1} \sum_{i \in [N]} \upsilon_i \EE_{\muxo^i}[\ell(Y, \alpha, \beta)] \\
                 \st & \muxo^i \in \mc M(\mc Y) & \forall i \in [N] \\
                & \ds \sum_{i \in [N]} \upsilon_i \EE_{\muxo^i}[\DD_{\mc X}(x_0, \wh x_i) + \DD_{\mc Y}(Y, \wh y_i)]  \le N\rho.
            \end{array}
        \right.
    \ee
    We can show that $\sup_{\upsilon \in \mc U}g(\upsilon)=\sup_{\upsilon \in \mc U}h(\upsilon)$.
    First, to show $\sup_{\upsilon \in \mc U}g(\upsilon)\leq \sup_{\upsilon \in \mc U}h(\upsilon)$, we fix an arbitrary value $\upsilon \in \mc U$. For any $\{\pi_i\}_{i \in [N]}$ that is feasible for~\eqref{eq:g-def}, define $\muxo^i \in \mc M(\mc Y)$ such that
    \[
        \upsilon_i \muxo^i(S) = \pi_i(\{x_0\} \times S) \qquad \forall S \subseteq \mc Y \text{ measurable} \qquad \forall i \in [N].
    \]
    One can verify that $\{\muxo^i\}_{i \in [N]}$ is a feasible solution to~\eqref{eq:h-def}, notably because
    \begin{align*}
        \sum_{i \in [N]} \upsilon_i \EE_{\muxo^i}[\DD_{\mc X}(x_0, \wh x_i) + \DD_{\mc Y}(Y, \wh y_i)] \le \sum_{i \in [N]} \EE_{\pi_i}[\DD_{\mc X}(X, \wh x_i) + \DD_{\mc Y}(Y, \wh y_i)] = \sum_{i \in [N]} \Wass(\pi_i, \delta_{(\wh x_i, \wh y_i)}) \leq N \rho,
    \end{align*}
    where the first inequality follows from the non-negativity of $\DD(\cdot, \wh x)$ and $\DD(\cdot, \wh y)$ by Assumption~\ref{a}\ref{a:cost}, and the second inequality is from the feasibility of $\{\pi_i\}_{i \in [N]}$ in~\eqref{eq:g-def}. Moreover, the optimal value of $\{\pi_i\}_{i \in [N]}$ in~\eqref{eq:g-def} and the optimal value of $\{\muxo^i\}_{i \in [N]}$ in~\eqref{eq:h-def} coincide because
    \[
        \int_{\mc Y} \ell(y, \alpha, \beta) \pi_i(\{x_0\} \times \dd y) =  \upsilon_i \int_{\mc Y} \ell(y, \alpha, \beta)  \muxo^i(\dd y) = \upsilon_i \EE_{\muxo^i}[\ell(Y, \alpha, \beta)] \qquad \forall i \in [N].
    \]
    This implies that $g(\upsilon) \le h(\upsilon)$ for any $\upsilon \in \mc U$, and thus we have
    \be \label{eq:equiv-1}
        \Sup{\upsilon \in \mc U}~g(\upsilon) \le \Sup{\upsilon \in \mc U}~h(\upsilon).
    \ee

    Next, we will establish the reverse direction of the inequality in~\eqref{eq:equiv-1}. To this end, consider any $\{\muxo^i\}_{i \in [N]}$ that is feasible for~\eqref{eq:h-def}, we will construct explicitly a sequence of probability families $\{\pi_{i, k}\}_{i \in [N], k \in \mbb N}$ that is feasible for~\eqref{eq:g-def} and attains the same objective value in the limit as $k \to \infty$.

In doing so, we start by supposing that $\sum_{i \in [N]} \upsilon_i \EE_{\muxo^i}[\DD_{\mc X}(x_0, \wh x_i) + \DD_{\mc Y}(Y, \wh y_i)]  < N\rho$. Let us define the measures $\{\pi_{i}\}_{i \in [N]}$ as
    \[
    \pi_{i} = \upsilon_i \delta_{x_0} \otimes \muxo^i + (1- \upsilon_i) \delta_{( x_{i},  \wh{y}_{i})},
    \]
    for  some $x_{i} \in \mc X\setminus \{\wh x_i\}$ that satisfy 
    \begin{align*}
        \sum_{i \in [N]} (1 - \upsilon_i) \DD_{\mc X}(x_{i}, \wh x_i) &\le N\rho - \sum_{i \in [N]} \upsilon_i \big( \DD_{\mc X}(x_0, \wh x_i) +  \EE_{\muxo^i}[\DD_{\mc Y}(Y, \wh y_i)] \big).
    \end{align*}
    Notice that under the condition of this case, the right-hand side is strictly positive, and the existence of such $x_{i}$'s is guaranteed thanks to the continuity of $\DD_{\mc X}$ and the fact that $\DD_{\mc X}(\wh x, \wh x) = 0$ by Assumption~\ref{a}\ref{a:cost}. 
    It is now easy to verify that $\{\pi_{i}\}_{i \in [N]}$ is feasible for~\eqref{eq:g-def}, and moreover, the objective value of $\{\pi_{i}\}_{i \in [N]}$ in~\eqref{eq:g-def} amounts to
    \[
        (\sum_{i\in[N]} \upsilon_i)^{-1} \sum_{i \in [N]} \int_{\mc Y} \ell(y, \alpha, \beta) \pi_i(\{x_0\} \times \dd y) = (\sum_{i \in [N]} \upsilon_i)^{-1}\sum_{i \in [N]} \Big[ \upsilon_i \EE_{\muxo^i}[\ell(Y, \alpha, \beta)] \Big] .
    \]

\newcommand{\ihat}{\wh{i}}    
    The case where $\sum_{i \in [N]} \upsilon_i \EE_{\muxo^i}[\DD_{\mc X}(x_0, \wh x_i) + \DD_{\mc Y}(Y, \wh y_i)]  = N\rho$ is more complex. In particular, we start by focusing on the situation where there exists some $\ihat$ for which $v_{\ihat}\neq 0$ and $\mu_0^{\ihat}\neq\delta_{\wh{y}_{\ihat}}$. Here, we can construct a sequence of measure
$\{\pi_{i, k}\}_{i \in [N], k \in \mbb N}$ as
    \[
    \pi_{i, k} = 
    \begin{cases}
        (\upsilon_i-\gamma_{k}) \delta_{x_0} \otimes \muxo^i + \gamma_{k} \delta_{(x_0, \wh y_i)} + (1- \upsilon_i) \delta_{( x_{i,k},  \wh{y}_{i})}  & \text{if } i = \ihat, \\
        \upsilon_i\delta_{x_0} \otimes \muxo^i + (1- \upsilon_i) \delta_{( x_{i,k},  \wh{y}_{i})} & \text{otherwise,}
    \end{cases}
    \]
    for some $\gamma_{k} \in [0, \upsilon_{\ihat}]$, $\lim_{k\to \infty} \gamma_{k} = 0$,  and some  $x_{i,k} \in \mc X\setminus \{\wh x_i\}$
    \[
        x_{i,k} \xrightarrow{k \to \infty} \wh x_i \quad \forall k.
    \]
    Furthermore, let the sequences $\gamma_{k}$ and $x_{i,k}$ satisfy  for any $k \in \mbb N$
    \begin{align*}
        \sum_{i \in [N]} (1 - \upsilon_i) \DD_{\mc X}(x_{i,k}, \wh x_i)  &\le \gamma_{k} \EE_{\muxo^{\ihat}}[\DD_{\mc Y}(Y, \wh y_{\ihat})].
    \end{align*}
    Notice that under the condition of this case, the right-hand side is strictly positive, and the existence of the sequence $(x_{i,k}, \gamma_k)$ is again guaranteed thanks to the continuity of $\DD_{\mc X}$ and the fact that $\DD_{\mc X}(\wh x, \wh x) = 0$ by Assumption~\ref{a}\ref{a:cost}. 
    It is now easy to verify that for any $k$, $\{\pi_{i,k}\}_{i \in [N]}$ is feasible for~\eqref{eq:g-def}:
\begin{align*}
\sum_{i \in [N]} \Wass(\pi_{i,k}, \delta_{(\wh x_i, \wh y_i)})&=\sum_{i \in [N]} \upsilon_i (\DD_{\mc X}(x_0, \wh x_i)) + \EE_{\muxo^i}[\DD_{\mc Y}(Y, \wh y_i)]) +  (1 - \upsilon_i) \DD_{\mc X}(x_{i,k}, \wh x_i) - \gamma_{k} \EE_{\muxo^{\ihat}}[\DD_{\mc Y}(Y, \wh y_{\ihat})] \\    
&\leq \sum_{i \in [N]} \upsilon_i (\DD_{\mc X}(x_0, \wh x_i) + \EE_{\muxo^i}[\DD_{\mc Y}(Y, \wh y_i)])) \leq N\rho\,.
\end{align*}
    Moreover, the objective value of $\{\pi_{i,k}\}_{i \in [N]}$ in~\eqref{eq:g-def} amounts to
    \[
        (\sum_{i \in [N]} \upsilon_i)^{-1}\left(\sum_{i \in [N]} \upsilon_i \EE_{\muxo^i}[\ell(Y, \alpha, \beta)] + \gamma_{k} (\ell( \wh y_i, \alpha, \beta) - \EE_{\muxo^{\ihat}}[\ell(Y, \alpha, \beta)])\right) \xrightarrow{k \to \infty} (\sum_{i \in [N]} \upsilon_i)^{-1}\sum_{i \in [N]} \upsilon_i \EE_{\muxo^i}[\ell(Y, \alpha, \beta)],
    \]
    where the limit holds because $\lim_{k\to \infty} \gamma_{k} = 0$.    

In the final case, we have $\muxo^i = \delta_{\wh y_i}$ for all $i \in [N]$ and $\sum_{i \in [N]} \upsilon_i \kappa_i  = N\rho$. Since we have assumed that $\rho>\rho_{\min}(x_0,0,\varepsilon)$, it must be that there exists an $\ihat$ for which $v_{\ihat}>0$ and $\wh{x}_{\ihat}\neq x_0$. Indeed, otherwise we would have that $N\rho = \sum_{i \in [N]} \upsilon_i \kappa_i = 0 \leq N \rho_{\min}(x_0,0,\varepsilon)$ which is a contradiction. Now let us one final time construct a sequence of measures
\[
    \pi_{i, k} = 
    \begin{cases}
        (\upsilon_i-\gamma_{k}) \delta_{(x_0,\wh{y}_i)}  + (1- \upsilon_i+\gamma_{k}) \delta_{(\wh{x}_i, \wh y_i)} & \text{if } i = \ihat, \\
        \upsilon_i\delta_{(x_0,\wh{y}_i)} + (1- \upsilon_i) \delta_{(\wh{x}_i, \wh{y}_{i})} & \text{if } \wh{x}_i\neq x_0, \\
        \upsilon_i\delta_{(x_0,\wh{y}_i)} + (1- \upsilon_i) \delta_{( x_{i,k},  \wh{y}_{i})} & \text{otherwise,}
    \end{cases}
    \]
for some $\gamma_k\in(0,\;v_{\ihat})$ and some $x_{i,k}\in\mc X \setminus \{x_0\}$, such that $x_{i,k} \xrightarrow{k\rightarrow \infty} \wh{x}_i$, and that
\[\gamma_k \DD_{\mc X}(x_0, \wh x_{\ihat}) \geq \sum_{i:\wh{x}_i=x_0} (1-v_i) \DD_{\mc X}(x_0,  x_{i,k})\,.\]
Indeed, it is easy to verify that $\{\pi_{i,k}\}_{i \in [N]}$ is always feasible for~\eqref{eq:g-def}, and moreover, the objective value of $\{\pi_{i,k}\}_{i \in [N]}$ in~\eqref{eq:g-def} amounts to
\begin{align*}
        (\sum_{i \in [N]} \upsilon_i)^{-1} &\left(\sum_{i \in [N]} \upsilon_i \ell(\wh y_i, \alpha, \beta)  - \gamma_k \ell(\wh y_{\ihat}, \alpha, \beta)\right) \\
        &\xrightarrow{k\rightarrow \infty} (\sum_{i \in [N]} \upsilon_i)^{-1} \sum_{i \in [N]} \upsilon_i \ell(\wh y_i, \alpha, \beta)   = (\sum_{i\in[N]} \upsilon_i)^{-1} \sum_{i \in [N]} \upsilon_i \EE_{\muxo^i}[\ell(Y, \alpha, \beta)] \,,
\end{align*}
where the limit holds because $\lim_{k \to \infty} \gamma_{k} \to 0$.

    Combining the three cases, we can establish that
    \[
        \Sup{\upsilon \in \mc U}~g(\upsilon) \ge \Sup{\upsilon \in \mc U}~h(\upsilon),
    \]
    and by considering the above inequality along with~\eqref{eq:equiv-1}, we can claim that 
    \[
    \Sup{\upsilon \in \mc U}~g(\upsilon) = \Sup{\upsilon \in \mc U}~h(\upsilon).
    \]
    One now can rewrite
    \begin{subequations}
    \begin{align}
        \Sup{\upsilon \in \mc U}~h(\upsilon) &= \left\{
            \begin{array}{ccl}
                \Sup{\muxo^i \in \mc M(\mc Y)~\forall i} & \sup & \ds (\sum_{i\in[N]} \upsilon_i)^{-1} \sum_{i \in [N]} \upsilon_i \EE_{\muxo^i}[\ell(Y, \alpha, \beta)] \\
                & \st & \upsilon \in [0, 1]^N \\
                && \ds \sum_{i \in [N]} \upsilon_i \EE_{\muxo^i}[\DD_{\mc X}(x_0, \wh x_i) + \DD_{\mc Y}(Y, \wh y_i)]  \le N\rho \\
                && \sum_{i \in [N]} \upsilon_i \ge N \eps.
            \end{array}
        \right. \label{eq:h-func1} \\
        &= \left\{
            \begin{array}{ccl}
                \Sup{\muxo^i \in \mc M(\mc Y)~\forall i} & \sup & \ds (\sum_{i\in[N]} \upsilon_i)^{-1} \sum_{i \in [N]} \upsilon_i \EE_{\muxo^i}[\ell(Y, \alpha, \beta)] \\
                & \st & \upsilon \in [0, 1]^N \\
                && \ds \sum_{i \in [N]} \upsilon_i \EE_{\muxo^i}[\DD_{\mc X}(x_0, \wh x_i) + \DD_{\mc Y}(Y, \wh y_i)]  \le N\rho \\
                && \sum_{i \in [N]} \upsilon_i = N \eps.
            \end{array}
        \right. \label{eq:h-func2}
        \\
        &= \left\{
            \begin{array}{ccl}
                \Sup{\upsilon \in [0, 1]^N, \sum_{i\in [N]} \upsilon_i = N\eps} & \sup & \ds (\sum_{i\in[N]} \upsilon_i)^{-1} \sum_{i \in [N]} \upsilon_i \EE_{\muxo^i}[\ell(Y, \alpha, \beta)] \\
                & \st & \muxo^i \in \mc M(\mc Y) \quad \forall i \in [N] \\
                && \ds \sum_{i \in [N]} \upsilon_i \EE_{\muxo^i}[\DD_{\mc X}(x_0, \wh x_i) + \DD_{\mc Y}(Y, \wh y_i)]  \le N\rho,
            \end{array}
         \right. \label{eq:h-func3}
    \end{align}
    \end{subequations}
    where equality~\eqref{eq:h-func1} and~\eqref{eq:h-func3} are by interchanging the order of the two supremum operators, equality~\eqref{eq:h-func2} is from Lemma~\ref{lemma:alphaProbSimplifies}. This completes the proof.
\end{proof}

%%%%%%%%%%%%%%%%%%%%%%%%%%%%%

We are now ready to prove the results of Section~\ref{sec:gamma0}.

\begin{proof}[Proof of Proposition~\ref{prop:zero-eps}] 
	Because $\Pnom$ is an empirical measure, any joint probability measure $\pi \in \Pi(\QQ, \Pnom)$ can be written as $\pi = \frac{1}{N} \sum_{i \in [N]}\pi_i \otimes \delta_{(\wh x_i, \wh y_i)}$, where each $\pi_i$ is a probability measure on $\mc X \times \mc Y$. Thus, by the definition of the optimal transport cost, we find
	\begin{align*}
	\mbb B_\rho
	= \left\{ \QQ \in \mc M(\mc X \times \mc Y):
	\begin{array}{l}
	\exists ~\pi_i \in \mc M(\mc X \times \mc Y) ~\forall i \in [N] \text{ such that }  \QQ = N^{-1} \sum_{i \in [N]} \pi_i\\
	\ds \frac{1}{N} \sum_{i \in [N]} \int_{\Xi} [\DD_{\mc X}(x, \wh x_i) + \DD_{\mc Y}(y, \wh y_i)] \pi_i(\mathrm{d}x \times \mathrm{d} y) \leq \rho
	\end{array}
	\right\}.
	\end{align*}
	Fix any arbitrary $y_0 \in \mc Y$. For any $i \in [N]$, let $(x_i', y_i') \in \mc X \times \mc Y$ be such that $x_i' \neq x_0$ and
	\[
	    \DD_{\mc X}(x_i', \wh x_i) + \DD_{\mc Y}(y_i', \wh y_i) \le \frac{\rho}{2N}.
	\]
	Consider the following set of probability measures $\{\pi_i\}_{i \in [N]}$ defined through
	\[
	\forall i \in [N-1]: \quad \pi_i = \begin{cases}
	\delta_{(x_i', y_i')} & \text{if } i \in [N-1], \\
	\upsilon \delta_{(x_0, y_0)} + (1-\upsilon) \delta_{(x_N',  y_N')} & \text{if } i = N,
	\end{cases}
	\]
	for some $\upsilon \in (0, 1]$ satisfying
	$\upsilon [\DD_{\mc X}(x_0, \wh x_N) + \DD_{\mc Y}(y_0, \wh y_N)] \leq \rho/(2N)$.
	Given this specific construction of $\{\pi_i\}_{i \in [N]}$, we can verify that
	\[
	\ds \frac{1}{N} \sum_{i \in [N]} \int_{\mc X \times \mc Y} [\DD_{\mc X}(x, \wh x_i) + \DD_{\mc Y}(y, \wh y_i)] \pi_i(\mathrm{d}x \times \mathrm{d} y) \leq \rho, \quad \text{ and } \quad \sum_{i\in [N]} \pi_i(X = x_0) = \upsilon > 0.
	\]
    As such, the measure $\QQ' = N^{-1} \sum_{i \in [N]} \pi_i$ satisfies $\QQ' \in \mbb B_\rho$ and $\QQ'(X = x_0) > 0$. We thus have
    \[
    \Sup{\QQ \in \mbb B_\rho, \QQ(X = x_0) > 0}~ \EE_{\QQ}[ \ell(Y, \alpha, \beta) | X = x_0 ] \ge \EE_{\QQ'}[\ell(Y, \alpha, \beta) | X = x_0] = \ell(y_0, \alpha, \beta).
    \]
    Because the choice of $y_0$ is arbitrary, we find
    \[
        \Sup{\QQ \in \mbb B_\rho, \QQ(X = x_0) > 0}~ \EE_{\QQ}[ \ell(Y, \alpha, \beta) | X = x_0 ] \ge \Sup{y \in \mc Y}~\ell(y, \alpha, \beta).
    \]
	Moreover, because the distribution of $Y$ given $X = x_0$ is supported on $\mc Y$, we have
	\[
	    \Sup{\QQ \in \mbb B_\rho, \QQ(X = x_0) > 0}~ \EE_{\QQ}[ \ell(Y, \alpha, \beta) | X = x_0 ] \le \Sup{\QQ \in \mbb B_\rho} \EE_{\QQ} [ \ell(Y, \alpha, \beta)| X = x_0] \leq \sup_{y \in \mc Y}~\ell(y, \alpha, \beta).
	\]
	This observation establishes the postulated equality and completes the proof.
\end{proof}

\begin{proof}[Proof of Theorem~\ref{thm:type-1-refor}]
    By applying Proposition~\ref{prop:equiv}, we have
    \begin{align*}
        & \Sup{\substack{\QQ \in \mbb B_\rho \\ \QQ(X = x_0) \ge \eps}} \EE_{\QQ}[ \ell(Y, \alpha, \beta) | X = x_0 ] =\left\{
            \begin{array}{ccl}
                \Sup{\upsilon \in [0, 1]^N, \sum_{i\in [N]} \upsilon_i = N\eps} & \sup & \ds (N\eps)^{-1} \sum_{i \in [N]} \upsilon_i \EE_{\muxo^i}[\ell(Y, \alpha, \beta)] \\
                & \st & \muxo^i \in \mc M(\mc Y)~\forall i \in [N] \\
                && \ds \sum_{i \in [N]} \upsilon_i \big( \DD_{\mc X}(x_0, \wh x_i) + \EE_{\muxo^i}[ \DD_{\mc Y}(Y, \wh y_i)] \big)  \le N\rho.
            \end{array}
         \right. 
    \end{align*}
    For any $\upsilon$ satisfying $\sum_{i \in [N]} \upsilon_i \kappa_i > N \rho$, the inner supremum subproblem is infeasible, thus, without loss of optimality, we can add the constraint $\sum_{i \in [N]} \upsilon_i \kappa_i \le N \rho$ into the outer supremum. Denote temporarily by $\mc U$ the set
    \[
        \mc U = \left\{ \upsilon \in [0, 1]^N: \sum_{i \in [N]} \upsilon_i = N \eps,~\sum_{i \in [N]} \upsilon_i \kappa_i \leq N \rho \right\}.
    \]
    For any $\upsilon \in \mc U$, strong duality holds because $\muxo^i:=\delta_{ \wh{y}_i}$ constitutes a Slater point. The duality result for moment problem~\cite[Proposition 3.4]{ref:shapiro2001on} implies that the inner supremum problem is equivalent to
    \[
        \begin{array}{cll}
            \inf & \ds \lambda_1 \big(N \rho - \sum_{i \in [N]} \upsilon_i \DD_{\mc X}(x_0, \wh x_i) \big) + \sum_{i \in [N]} \theta_i \\
            \st & \lambda_1 \in \R_+,\; \theta \in \R^N \\
            & \upsilon_i \DD_{\mc Y}(y_i, \wh y_i) \lambda_1 + \theta_i \ge (N\eps)^{-1} \upsilon_i \ell(y_i, \alpha, \beta) & \forall y_i \in \mc Y,\; \forall i \in [N].
        \end{array}
    \]
    By rescaling $\theta_i \leftarrow \upsilon_i \theta_i$, we have the equivalent form
    \begingroup
    \allowdisplaybreaks
     \begin{align*}
        & \Sup{\QQ \in \mbb B_\rho, \QQ(X = x_0) \ge \eps}~ \EE_{\QQ}[ \ell(Y, \alpha, \beta) | X = x_0 ] \\
        =& \left\{
            \begin{array}{ccl}
                \Sup{\upsilon\in [0, 1]^N, \sum_{i\in [N]} \upsilon_i = N\eps} & \inf & \ds\lambda_1 \big(N \rho - \sum_{i \in [N]} \upsilon_i \DD_{\mc X}(x_0, \wh x_i) \big) + \sum_{i \in [N]} \upsilon_i \theta_i \\
            &\st & \lambda_1 \in \R_+,\; \theta \in \R^N \\
            && \DD_{\mc Y}(y_i, \wh y_i) \lambda_1 + \theta_i \ge (N\eps)^{-1}  \ell(y_i, \alpha, \beta) \quad \forall y_i \in \mc Y,\; \forall i \in [N]
            \end{array}
         \right. \\
         =& \left\{
            \begin{array}{cl}
                 \inf & \Sup{\upsilon\in [0, 1]^N, \sum_{i\in [N]} \upsilon_i = N\eps} \; \ds\lambda_1 \big(N \rho - \sum_{i \in [N]} \upsilon_i \DD_{\mc X}(x_0, \wh x_i) \big) + \sum_{i \in [N]} \upsilon_i \theta_i \\
                \st & \lambda_1 \in \R_+,\; \theta \in \R^N \\
                & \DD_{\mc Y}(y_i, \wh y_i) \lambda_1 + \theta_i \ge (N\eps)^{-1}  \ell(y_i, \alpha, \beta) \quad \forall y_i \in \mc Y,\; \forall i \in [N].
            \end{array}
         \right. \\
         =& \left\{
            \begin{array}{cll}
                \inf & \ds N\rho \lambda_1 + N\eps \lambda_2 + \sum_{i \in [N]} \vartheta_i \\
                \st & \lambda_1 \in \R_+,\; \lambda_2 \in \R, \; \theta \in \R^N,\; \vartheta \in \R_+^N \\
                & \DD_{\mc Y}(y_i, \wh y_i) \lambda_1 + \nu_i \ge (N\eps)^{-1}  \ell(y_i, \alpha, \beta) & \forall y_i \in \mc Y,\; \forall i \in [N]\\
                & \lambda_2 + \vartheta_i \ge \theta_i -  \DD_{\mc X}(x_0, \wh x_i)\lambda_1 & \forall i \in [N],
            \end{array}
         \right.
    \end{align*}
    \endgroup
    where the second equality follows from interchanging the min-max operators using Sion's minimax theorem~\cite{ref:sion1958minimax}. By eliminating $\theta$, we have 
    \begin{align*}
        & \Sup{\QQ \in \mbb B_\rho, \QQ(X = x_0) \ge \eps}~ \EE_{\QQ}[ \ell(Y, \alpha, \beta) | X = x_0 ] \\
        &= \left\{
            \begin{array}{cll}
                \inf & \ds N\rho \lambda_1 + N\eps \lambda_2 + \sum_{i \in [N]} \vartheta_i \\
                \st & \lambda_1 \in \R_+,\; \lambda_2 \in \R,\; \vartheta \in \R_+^N \\
                & \vartheta_i \ge \sup_{y_i \in \mc Y}\left\{ (N\eps)^{-1}  \ell(y_i, \alpha, \beta) -  [\DD_{\mc X}(x_0, \wh x_i) + \DD_{\mc Y}(y_i, \wh y_i)]\lambda_1 - \lambda_2 \right\} & \forall i \in [N],
            \end{array}
         \right. \\
         &= \left\{
            \begin{array}{cll}
                \inf & \ds \rho \lambda_1 + \eps \lambda_2 + \frac{1}{N}\sum_{i \in [N]} \vartheta_i \\
                \st & \lambda_1 \in \R_+,\; \lambda_2 \in \R,\; \vartheta \in \R_+^N \\
                &\vartheta_i \ge \sup_{y_i \in \mc Y} \big\{\eps^{-1} \ell(y_i, \alpha, \beta) - [\DD_{\mc X}(x_0, \wh x_i) + \DD_{\mc Y}(y_i, \wh y_i)] \lambda_1 - \lambda_2
                \big\} \quad \forall i \in [N]
            \end{array}
         \right. 
    \end{align*}
    where the second equality follows by rescaling the dual variables. Eliminating $\vartheta$ leads to the desired result.
\end{proof}

\begin{proof}[Proof of Proposition~\ref{prop:markowitz-reform}] 
By exploiting the quadratic form of $\ell$, the supremum problem in the last constraint of \eqref{eq:reform-1} is a quadratic optimization problem that admits a closed-form expression as 
\begin{align*}
    &\Sup{y_i \in \mc Y} \big\{\eps^{-1} \ell(y_i, \alpha, \beta) - \lambda_1 \DD_{\mc Y}(y_i, \wh y_i)  \big\}\\
    =& \Sup{y_i \in \mc Y} \big\{\eps^{-1} (y_i^{\top} \alpha - \beta)^2 - \eps^{-1} \eta\cdot y_i^{\top} \alpha -\lambda_1\cdot \|y_i - \wh y_i\|_2^2\big\}\\
    =&-\frac{1}{4}\eps^{-1}\eta^2-\eps^{-1}\eta\beta
    + \Sup{y_i \in \mc Y} \big\{\eps^{-1} (y_i^{\top} \alpha - \beta - \frac{1}{2}\eta)^2 - \lambda_1\cdot \|y_i - \wh y_i\|_2^2\big\}.
\end{align*}
Let $\Delta_i = y_i^{\top} \alpha - \wh y_i^{\top} \alpha$, then we have
\begin{align*}
 &\Sup{y_i \in \mc Y} \big\{\eps^{-1} \ell(y_i, \alpha, \beta) -  \DD_{\mc Y}(y_i, \wh y_i) \lambda_1 \big\}\\
 =& -\frac{1}{4}\eps^{-1}\eta^2-\eps^{-1}\eta\beta
 + \Sup{\Delta_i \in \R} \big\{\eps^{-1} (\Delta_i+\wh y_i^{\top} \alpha - \beta - \frac{1}{2}\eta)^2 - \lambda_1\cdot\Inf{y_i \in \R^m: y_i^{\top} \alpha - \wh y_i^{\top} \alpha = \Delta_i} \|y_i - \wh y_i\|_2^2\big\}\\
 =& -\frac{1}{4}\eps^{-1}\eta^2-\eps^{-1}\eta\beta
 + \Sup{\Delta_i \in \R} \big\{\eps^{-1} (\Delta_i+\wh y_i^{\top} \alpha - \beta - \frac{1}{2}\eta)^2 - \lambda_1\Delta^2_i\|\alpha\|_2^{-2}\big\},
\end{align*}
where the last equality is obtained by applying the Cauchy-Schwarz inequality, i.e., $\Delta_i = (y_i-\wh{y}_i)^\top\alpha\leq \|y_i-\wh{y}_i\|_2\|\alpha\|_2$, which implies that the minimum for $\|y_i-\wh{y}_i\|_2^2$ is $\Delta_i \|\alpha\|_2^{-2}$. By combining cases, we find
\begin{align*}
    &\Sup{y_i \in \mc Y} \big\{\eps^{-1} \ell(y_i, \alpha, \beta) - \lambda_1 \DD_{\mc Y}(y_i, \wh y_i)  \big\}\\
    =& \begin{cases}
    \eps^{-1}\cdot 
    \frac{(\wh y_i^\top \alpha - \beta-\half \eta)^2}{1-\eps^{-1}\|\alpha\|_2^2/\lambda_1}
    -\frac{1}{4}\eps^{-1}\eta^2-\eps^{-1}\eta\beta&   \mbox{if }
    \lambda_1 > \eps^{-1}\|\alpha\|_2^2,\\
    -\frac{1}{4}\eps^{-1}\eta^2-\eps^{-1}\eta\beta& 
    \mbox{if }\lambda_1 = \eps^{-1}\|\alpha\|_2^2\mbox{ and } \wh y_i^\top \alpha = \beta+ \eta/2,\\
    +\infty&   \mbox{otherwise}.
    \end{cases}
\end{align*}
Consider the case of $\lambda_1 > \eps^{-1}\|\alpha\|^2$ first, in which the last constraint of \eqref{eq:reform-1} is satisfied if and only if there exist ancillary variables $w\in (0,1)$ and $z \in \R_{+}^{N}$, such that
\begin{align*}
    1-\eps^{-1}\|\alpha\|_2^2/\lambda_1 \geq w\quad\Longleftrightarrow\quad\eps (1 - w)\lambda_1 \geq \|\alpha\|_2^2
\end{align*}
and
\begin{align*}
    \theta_i
    \ge
    \eps^{-1}\cdot 
    \frac{(\wh y_i^\top \alpha - \beta-\half \eta)^2}{1-\eps^{-1}\|\alpha\|_2^2/\lambda_1}
    -\frac{1}{4}\eps^{-1}\eta^2-\eps^{-1}\eta\beta
    \Longleftrightarrow
    \left\{
    \begin{array}{l}
    z_i = \eps\theta_i+
    \eps\|x_0 - \wh x_i\|^2 \lambda_1 + \eps\lambda_2 + \frac{1}{4}\eta^2+\eta\beta\\
    z_i w \geq \big(\wh y_i^\top \alpha - \beta-\half \eta\big)^2
    \end{array}
    \right.
    \; \forall i \in [N].
\end{align*}
Using the equivalent reformulation between hyperbolic constraint and second-order cone constraint \cite{ref:lobo1998applications}, we have
\begin{align*}
    \eps (1 - w)\lambda_1 \geq \|\alpha\|_2^2 \Longleftrightarrow
    \left\|
    \begin{bmatrix}
    2\alpha\\
    1 - w - \eps\lambda_1
    \end{bmatrix}\right\|_2
    \leq 1 - w + \eps\lambda_1
\end{align*}
and 
\begin{align*}
z_i w \geq \big(\wh y_i^\top \alpha - \beta-\half \eta\big)^2
\Longleftrightarrow
\left\|
\begin{bmatrix}
2\wh y_i^\top \alpha - 2\beta- \eta\\
z_i - w
\end{bmatrix}\right\|_2
\leq z_i + w \quad
\forall i \in [N].
\end{align*}
Now we consider the case of $\lambda_1 = \eps^{-1}\|\alpha\|_2^2$, where the last constraint of \eqref{eq:reform-1} is equivalent to the cone constraints when $w = 0$. Finally notice that $w = 1$ recovers the original constraints when $\alpha = 0$.
Combining all of the above cases and using them to replace the last constraint of \eqref{eq:reform-1} completes the proof.
\end{proof}

\begin{proof}[Proof of Proposition~\ref{prop:mean-cvar-reform}]
Notice that $\ell$ is a pointwise maximum of two linear functions of $y_i$. In this case, the supremum problem in the last constraint of \eqref{eq:reform-1} can be written as
\begin{align*}
    &\Sup{y_i \in \mc Y} \big\{\eps^{-1} \ell(y_i, \alpha, \beta) -  \DD_{\mc Y}(y_i, \wh y_i) \lambda_1 \big\}\\
    =&\max\left\{
        \begin{array}{l}
    \Sup{y_i \in \R^m}~ \Big\{- \eps^{-1}\eta y_i^\top \alpha + \eps^{-1}\beta
    -\lambda_1\cdot \|y_i - \wh y_i\|_2^2 \Big\},\\
    \Sup{y_i \in \R^m}~ \Big\{- \eps^{-1}(\eta + \frac{1}{\tau}) y_i^\top \alpha + \eps^{-1}(1 - \frac{1}{\tau})\beta
    -\lambda_1\cdot \|y_i - \wh y_i\|_2^2 \Big\}
    \end{array}
    \right\}\\
    =& 
    \max\bigg\{
    - \eps^{-1}\eta \wh y_i^\top \alpha + \eps^{-1}\beta+
    \frac{\eps^{-2}\eta^2}{4\lambda_1}\|\alpha\|_2^2,~
    - \eps^{-1}(\eta + \frac{1}{\tau}) \wh y_i^\top \alpha + \eps^{-1}(1 - \frac{1}{\tau})\beta
    + \frac{\eps^{-2}(\eta + \tau^{-1})^2}{4\lambda_1}\|\alpha\|_2^2
    \bigg\}.
\end{align*}
Therefore, using additional variables $z_i\in\R_{+}$ and $\tilde{z}_i\in\R_{+}$,
we can reformulate the last constraint of \eqref{eq:reform-1} as the following set of constraints
\begin{equation*}
    \left.
    \begin{array}{l}
    z_i = \theta_i +\lambda_1 \|x_0 - \wh x_i\|^2+ \lambda_2 + \eps^{-1}\eta \wh y_i^\top \alpha - \eps^{-1}\beta
    \\
    4\lambda_1 z_i
    \ge\eps^{-2}\eta^2\|\alpha\|_2^2\\
    \tilde{z}_i = 
    \theta_i +\lambda_1 \|x_0 - \wh x_i\|^2+ \lambda_2 + \eps^{-1}(\eta + \frac{1}{\tau}) \wh y_i^\top \alpha - \eps^{-1}(1 - \frac{1}{\tau})\beta
    \\
    4\lambda_1 \tilde z_i
    \ge\eps^{-2}(\eta + \tau^{-1})^2\|\alpha\|_2^2 
    \end{array}
    \right\} \quad \forall i \in [N].
\end{equation*}
Using the equivalent reformulation between hyperbolic constraint and second-order cone constraint \cite{ref:lobo1998applications}, we have for each $i \in [N]$
\begin{align*}
    4\lambda_1 z_i
    \ge\eps^{-2}\eta^2\|\alpha\|_2^2 \Longleftrightarrow
    \left\|
    \begin{bmatrix}
    \eps^{-1}\eta \alpha\\
    z_i-\lambda_1
    \end{bmatrix}\right\|_2
    \leq z_i+\lambda_1,
\end{align*}
and 
\begin{align*}
    4\lambda_1 \tilde z_i
    \ge\eps^{-2}(\eta + \tau^{-1})^2\|\alpha\|_2^2\Longleftrightarrow
    \left\|
    \begin{bmatrix}
    \eps^{-1}(\eta + \tau^{-1}) \alpha\\
    \tilde z_i-\lambda_1
    \end{bmatrix}\right\|_2
    \leq \tilde z_i+\lambda_1,
\end{align*}
which finishes the proof.
\end{proof}

\subsection{Proofs of Section~\ref{sec:gamma-positive}}
\label{appendix:sec4}

    \begin{proof}[Proof of Proposition~\ref{prop:rho_upper}]
	    Without any loss of optimality, we can substitute the constraint $\QQ(X \in \mc N_{\gamma}(x_0)) = 0$ by the inequality constraint $\QQ(X \in \mc N_{\gamma}(x_0)) \le 0$ to obtain
	    \begin{align*}
	        \rho_{\max}(x_0, \gamma) 	        &= \left\{
	            \begin{array}{cl}
	                \Inf{\pi \in \mc M((\mc X \times \mc Y) \times (\mc X \times \mc Y))} & \EE_{\pi}[ \DD(\xi, \xi')] \\
	                \st 
	                & \EE_{\pi}[ \mathbbm{1}_{\mc N_\gamma(x_0)\times\mc Y}(\xi)] \le 0, \quad \EE_{\pi}[\mathbbm{1}_{(\wh x_i, \wh y_i)}(\xi')] = \frac{1}{N} \qquad \forall i \in [N],
	            \end{array}
	        \right.
	    \end{align*}
	    where $\xi$ represents the joint random vector $(X, Y)$. By a weak duality result~\cite[Section~2.2]{ref:smith1995generalized}, we have
	    \begin{align*}
	        \rho_{\max}(x_0, \gamma) &\ge \left\{
	            \begin{array}{cl}
	                \Sup{b \in \R^N,~\zeta \in \R_+} & \ds  \frac{1}{N} \sum_{i=1}^N b_i \\
                    \st & \ds \sum_{i \in [N]} b_i \mathbbm{1}_{(\wh x_i, \wh y_i)}(x', y') - \zeta \mathbbm{1}_{\mc N_{\gamma}(x_0)}(x) \le \DD\big( (x, y), (x', y') \big) \quad \forall (x, y), (x', y') \in \mc X \times \mc Y
	            \end{array}
	        \right. 
	    \end{align*}
	    If $(x', y') \neq (\wh x_i, \wh y_i)$ for all $i \in [N]$ then the constraint does not involve the variables $b_i$ and thus does not affect the optimal value. Suppose that $(x', y') = (\wh x_i, \wh y_i)$ for some $i \in [N]$ then the constraint becomes
	    \[
	        b_i -  \zeta \mathbbm{1}_{\mc N_{\gamma}(x_0)}(x) \le \DD\big( (x, y), (\wh x_i, \wh y_i) \big) \quad \forall (x, y)\in \mc X \times \mc Y.
	    \]
	    Screening the above constraint for each $i \in [N]$, we obtain the following bound
	    
	    \begin{align*}
	        \rho_{\max}(x_0, \gamma) &\ge \left\{
	            \begin{array}{cl}
	                \Sup{b \in \R^N,~\zeta \in \R_+} & \ds  \frac{1}{N} \sum_{i=1}^N
	                b_i \\
                    \st 
                    &\ds  b_i - \zeta \mathbbm{1}_{\mc N_{\gamma}(x_0)}(x_i) \le \DD\big( (x_i, y_i), (\wh x_i, \wh y_i) \big) \quad \forall (x_i, y_i)\in \mc X \times \mc Y,~\forall i \in [N]
	            \end{array}
	        \right. \\
	        &=
	            \Sup{\zeta \in \R_+}~\frac{1}{N} \sum_{i \in [N]} \Inf{(x_i, y_i) \in \mc X \times \mc Y} \left\{ \DD\big( (x_i, y_i), (\wh x_i, \wh y_i) \big) + \zeta \mathbbm{1}_{\mc N_{\gamma}(x_0)}(x_i) \right\}
	         \\
	        &=
	            \Sup{\zeta \in \R_+}~\frac{1}{N} \sum_{i \in [N]} \Inf{x_i \in \mc X} \left\{ \DD_{\mc X}( x_i, \wh x_i) + \zeta \mathbbm{1}_{\mc N_{\gamma}(x_0)}(x_i) \right\},
	    \end{align*}
	   where the last equality follows from the decomposition of $\DD$ and the fact that the minimizer in $y_i$ is $\wh y_i$. For any $\zeta \ge 0$, we have for any $i \in [N]$
	    \[
	    \Inf{x_i \in \mc X} \left\{ \DD_{\mc X}( x_i, \wh x_i) + \zeta \mathbbm{1}_{\mc N_{\gamma}(x_0)}(x_i) \right\} = \begin{cases}
	        \min\{\zeta, d_i\} & \text{if } \wh{x}_i \in \mc N_{\gamma}(x_0), \\
	        0 & \text{if } \wh{x}_i \not \in \mc N_{\gamma}(x_0),
	    \end{cases}
	    \]
	    which implies that
	    \[
	    \rho_{\max}(x_0, \gamma) \ge \Sup{\zeta \in \R_+}~\frac{1}{N} \sum_{i \in \mc I_1} \min(\zeta, d_i) = \frac{1}{N} \sum_{i \in \mc I_1} d_i.
	    \]
	    In the last step, we show that the above inequality is tight. For any value $\rho$ such that $\rho >  \frac{1}{N} \sum_{i \in \mc I_1} d_i$, Assumption~\ref{a}\ref{a:vic} and the continuity of $\DD_{\mc X}$ and $\DD_{\mc Y}$ imply that there exists  $(x_i', y_i') \in (\mc X \times \mc Y) \backslash (\mc N_{\gamma}(x_0) \times \mc Y)$ such that
	    \[
	        \DD_{\mc X}(x_i', \wh x_i) + \DD_{\mc Y}(y_i', \wh y_i) \le d_i + \frac{1}{N}(\rho -  \frac{1}{N} \sum_{i \in \mc I_1} d_i).
	    \]
	    The distribution
	    \[
	        \QQ' = \frac{1}{N} \big( \sum_{i \in \mc I_1} \delta_{(x_i', y_i')} + \sum_{i \in \mc I_2} \delta_{(\wh x_i, \wh y_i)} \big)
	    \]
	    thus satisfies $\Wass(\QQ', \Pnom) \leq \rho$ and $\QQ'(X \in \mc N_{\gamma}(x_0)) \leq 0$. This implies that $\rho_{\max}(x_0, \gamma) = N^{-1} \sum_{i \in \mc I_1} d_i$ and completes the proof.
	\end{proof}
	
	\begin{proof}[Proof of Proposition~\ref{prop:robust_eps0}]
For simplicity, we assume that $\rho_{\max}(x_0,\gamma)>0$, which implies that $\mc I_1 \neq \emptyset$. Let's consider some index $j \in \mc I_1$.
Assumption~\ref{a}\ref{a:vic} and the continuity of $\DD_{\mc X}$ and $\DD_{\mc Y}$ imply that there exists  $x_i'\in \mc X \backslash \mc N_{\gamma}(x_0) $ such that
	    \[
	        \DD_{\mc X}(x_i', \wh x_i) \le d_i + \frac{1}{2N}(\rho -  \rho_{\max}(x_0, \gamma)).
	    \] 
	Fix an arbitrary value $y_0 \in \mc Y$. Consider the distribution
    \[
        \QQ' = \frac{1}{N} \big( \sum_{i \in \mc I_1\setminus\{j\}} \delta_{(x_i', \wh y_i)} + \sum_{i \in \mc I_2 } \delta_{(\wh x_i, \wh y_i)}  + \upsilon \delta_{(\wh x_j, y_0)} + (1 - \upsilon) \delta_{(\wh x_j, \wh y_j)}\big)
    \]
in which $\upsilon \in \R_{++}$ is chosen so that
\[
    \upsilon \DD_{\mc Y}(y_0, \wh y_j)  \le \frac{1}{2} (\rho - \rho_{\max}(x_0, \gamma)).
\]
It is easy to verify that $\QQ'(X\in \mc N_{\gamma}(x_0)) = \upsilon/N > 0$, $\Wass(\QQ', \Pnom) \le \rho$ and that
\[
\EE_{\QQ'}[ \ell(Y, \alpha, \beta) | X \in \mc N_\gamma(x_0) ] = \ell(y_0, \alpha, \beta).
\]
As $y_0$ is chosen arbitrarily, a similar argument as in the proof of Proposition~\ref{prop:zero-eps} leads to the necessary result. 

Note that alternatively, the case where $\rho_{\max}(x_0,\gamma)=0$ implies that all samples $\{\wh x_i\}_{i\in[N]}$ lie in the set $(\mc X\setminus \mc N_\gamma(x_0)) \cup \partial \mc N_\gamma(x_0) $. If $\mc I_1\neq \emptyset$, then the same argument as before applies. Otherwise, fixing $j$ to be such that $\wh{x}_j\notin \mc N_\gamma(x_0)$, the same results is obtained with 
\[
        \QQ' = \frac{1}{N} \big(  \sum_{i \in \mc I_2\setminus \{j\} } \delta_{(\wh x_i, \wh y_i)}  + \upsilon \delta_{(\wh x_j^*, y_0)} + (1 - \upsilon) \delta_{(\wh x_j, \wh y_j)}\big)
    \]
where $\wh x_j\opt$ is the projection of~$\wh x_j$ onto $\partial \mc N_\gamma(x_0)$. The proof is complete.
\end{proof}

\begin{proof}[Proof of Proposition~\ref{prop:eps_lower}]
We start by proving by contradiction that, when $\rho<\rho_{\max}(x_0,\gamma)$, we necessarily have that $\inf_{\QQ \in \mbb B_\rho} \QQ(X \in \mc N_\gamma(x_0))>0$. Let us assume that
$\inf_{\QQ \in \mbb B_\rho} \QQ(X \in \mc N_\gamma(x_0))=0$.
This implies that for all $\eps>0$, there exists a $\QQ \in \mbb B_\rho$ such that $\QQ(X \in \mc N_\gamma(x_0))\leq \eps$. Based on the following representation of $\mbb B_\rho$:
	\begin{align}
	\mbb B_\rho
	= \left\{ \QQ \in \mc M(\mc X \times \mc Y):
	\begin{array}{l}
	\exists ~\pi_i \in \mc M(\mc X \times \mc Y) ~\forall i \in [N] \text{ such that }  \QQ = N^{-1} \sum_{i \in [N]} \pi_i\\
	\ds \frac{1}{N} \sum_{i \in [N]} \EE_{\pi_i}[\DD_{\mc X}(X, \wh x_i) + \DD_{\mc Y}(Y, \wh y_i)]  \leq \rho
	\end{array}
	\right\}, \label{eq:BrhoRepPi}
	\end{align}
it must be that there exists an assignment for $\{\pi_i\}_{i=1}^N$ that satisfies
\[\frac{1}{N} \sum_{i \in [N]} \EE_{\pi_i}[\DD_{\mc X}(X, \wh x_i) + \DD_{\mc Y}(Y, \wh y_i)]  \leq \rho, \quad \text{and} \quad \frac{1}{N} \sum_{i \in [N]} \pi_i(X \in \mc N_\gamma(x_0)) \leq \eps.\]
We can work out the following steps:
\begin{align*}
    \rho &\geq \frac{1}{N} \sum_{i \in [N]} \EE_{\pi_i}[\DD_{\mc X}(X, \wh x_i) + \DD_{\mc Y}(Y, \wh y_i)] \geq \frac{1}{N} \sum_{i \in [N]} \EE_{\pi_i}[\DD_{\mc X}(X, \wh x_i)|X \notin \mc N_\gamma(x_0)]\pi_i(X \notin \mc N_\gamma(x_0)) \\
    &\geq \frac{1}{N} \sum_{i \in [N]} d_i(1-\pi_i(X \in \mc N_\gamma(x_0)) \geq \frac{1}{N}\sum_{i\in[N]} d_i - \left(\max_i d_i\right) \pi_i(X\in\mc N_\gamma(x_0)) \geq  \rho_{\max} - \left(\max_i d_i\right) \eps\,.
\end{align*}
Given that this is the case for all $\eps$, we conclude that $\rho\geq \rho_{\max}$, which contradicts our assumption about $\rho$.

Next, we turn to how to evaluate $\underline \eps:= \inf_{\QQ \in \mbb B_\rho} \QQ(X \in \mc N_\gamma(x_0))$. Given any $\bar{\QQ}\in\mbb B_\rho$, associated to some $\{\bar{\pi}_i\}_{i\in[N]}$ based on the representation in \eqref{eq:BrhoRepPi}, one can construct a new measure $\bar{\QQ}\opt$ of the form:
        \[\QQ_p\opt = \frac{1}{N} \sum_{i \in \mc I_1} \big( (1-p_i) \delta_{(\wh x_i\opt, \wh y_i)} + p_i \delta_{(\wh x_i, \wh y_i)} \big) + \frac{1}{N}\sum_{i \in \mc I_2} \delta_{(\wh x_i, \wh y_i)},
    \]
where $\wh x_i\opt
\in \mc X \backslash \mc N_\gamma(x_0)$ is a point arbitrarily close to the projection of $\wh x_i$ onto $\partial \mc N_{\gamma}(x_0)$, based on $\bar\QQ\opt \Let \QQ_{\bar{p}}\opt$ with $\bar{p}_i \Let \bar \pi_i(X\in\mc N_\gamma(x_0))$. It is easy to see how $\bar{\QQ}\opt\in\mbb B_\rho$ and achieves
$\bar{\QQ}\opt(X\in\mc N_\gamma(x_0)) \leq \bar\QQ(X\in\mc N_\gamma(x_0))$.
Hence, using similar arguments as in the proof of Proposition \ref{prop:robust_eps0}, we have that
\begin{align*}
    \underline \eps &= \left\{
        \begin{array}{cl}
        \inf &\QQ_p\opt(X\in\mc N_\gamma(x_0)) \\
        \st & p \in [0, 1]^N,~ \wh x_i\opt \in \mc X \backslash\mc N_\gamma(x_0)~\forall i \in \mc I_1 \\
            & \QQ_p\opt = \frac{1}{N} \sum_{i \in \mc I_1} \big( (1-p_i) \delta_{(\wh x_i\opt, \wh y_i)} + p_i \delta_{(\wh x_i, \wh y_i)} \big) + \frac{1}{N}\sum_{i \in \mc I_2} \delta_{(\wh x_i, \wh y_i)} \\
            & \frac{1}{N}\sum_{i \in \mc I_1} d_i (1-p_i) \le \rho \end{array} \right. \\
    &= \min\left\{ \frac{1}{N} \sum_{i \in \mc I_1} p_i : p \in [0, 1]^N,~ \frac{1}{N}\sum_{i \in \mc I_1} d_i (1-p_i) \le \rho \right\}\,.
\end{align*}
This completes the proof.
\end{proof}

\begin{proof}[Proof of Theorem~\ref{thm:eps-zero}] 
By exploiting the definition of the optimal transport cost and the fact that any joint probability measure $\pi \in \Pi(\QQ, \Pnom)$ can be written as $\pi = N^{-1} \sum_{i \in [N]}\pi_i \otimes \delta_{(\wh x_i, \wh y_i)}$, where each $\pi_i$ is a probability measure on $\mc X \times \mc Y$, we have
    \begin{align*}
        \Sup{\QQ \in \mbb B_\rho, \QQ(X \in \mc N_{\gamma}(x_0)) \geq \eps}~ \EE_{\QQ}[ \ell(Y, \alpha, \beta) | X \in \mc N_\gamma(x_0) ] 
        = &\left\{
            \begin{array}{cl}
                \sup & \EE_{\QQ}[ \ell(Y, \alpha, \beta) | X \in \mc N_\gamma(x_0) ] \\
                \st & \pi_i \in \mc M(\mc X \times \mc Y) \quad \forall i \in [N] \\
                & \sum_{i \in [N]}\pi_i(\mc N_{\gamma}(x_0) \times \mc Y) \geq N\eps \\
                & \QQ = N^{-1} \sum_{i \in [N]} \pi_i, ~\sum_{i \in [N]} \Wass(\pi_i, \delta_{(\wh x_i, \wh y_i)}) \le N\rho
            \end{array}
        \right. \\
        = &\left\{
            \begin{array}{cl}
                \sup & \ds\frac{ \sum_{i \in [N]}\EE_{\pi_i}[ \ell(Y, \alpha, \beta) \mathbbm{1}_{\mc N_{\gamma}(x_0)}(X)]}{\sum_{i \in [N]}\pi_i( X \in \mc N_{\gamma}(x_0))  } \\
                \st & \pi_i \in \mc M(\mc X \times \mc Y) \quad \forall i \in [N] \\
                & \sum_{i \in [N]}\pi_i(\mc N_{\gamma}(x_0) \times \mc Y) \geq N\eps \\
                & \sum_{i \in [N]} \Wass(\pi_i, \delta_{(\wh x_i, \wh y_i)}) \le N\rho.
            \end{array}
        \right. 
    \end{align*}
For any $i \in \mc I$, let $\wh x_i\opt = \arg\Min{x \in \partial \mc N_\gamma(x_0)} \DD_{\mc X} (x, \wh x_i)$ be again the projection of $\wh x_i$ onto $\partial \mc N_{\gamma}(x_0)$. Moreover, for any $i \in \mc I_1$, let $(x_i', y_i')$ be a point on $(\mc X \times \mc Y) \backslash (\mc N_{\gamma}(x_0) \times \mc Y)$. Using a continuity and greedy argument, we can choose $(x_i', y_i')$ sufficiently close to $(\wh x_i\opt, \wh y_i)$ so that it suffices to consider the conditional distribution $\pi_i$ of the form
\be \label{eq:pi_i-relation}
    \pi_i(\mathrm{d}x \times \mathrm{d}y) = 
    \begin{cases}
        p_i \delta_{\wh x_i}(\mathrm{d}x) \muxo^i(\mathrm{d}y) + (1-p_i) \delta_{(x_i', y_i')}(\mathrm{d}x \times \mathrm{d}y) & \text{if } i \in \mc I_1, \\
        p_i \delta_{\wh x_i\opt}(\mathrm{d}x) \muxo^i(\mathrm{d}y) + (1-p_i) \delta_{(\wh x_i, \wh y_i)}(\mathrm{d}x \times \mathrm{d}y) & \text{if } i \in \mc I_2,
    \end{cases}
\ee
for a parameter $p_i \in [0, 1]$ and a conditional measure $\muxo^i \in \mc M(\mc Y)$. Intuitively, $p_i$ represents the portion of the sample point $(\wh x_i, \wh y_i)$
  that is transported to the fiber $\mc N_{\gamma}(x_0) \times \mc Y$, and $\muxo^i$ is the conditional distribution of $Y$ given $X \in \mc N_{\gamma}(x_0)$ that is obtained by transporting the sample point $(\wh x_i, \wh y_i)$. Using this representation, we can rewrite $\QQ(X \in \mc N_{\gamma}(x_0)) = N^{-1} \sum_{i \in [N]} p_i$. By exploiting the definition of $d_i$, the worst-case conditional expected loss can be re-expressed as
  \begin{align*}
	    & \Sup{\QQ \in \mbb B_\rho, \QQ(X \in \mc N_\gamma(x_0)) \geq \eps}~ \EE_{\QQ}[ \ell(Y, \alpha, \beta) | X \in \mc N_\gamma(x_0) ] \\
	    =&\left\{
	    \begin{array}{cl}
	    \sup & \ds \frac{\sum_{i \in [N]} p_i \EE_{\muxo^i}[ \ell(Y, \alpha, \beta)]}{\sum_{i\in [N]} p_i} \\
	    \st & p_i \in [0, 1],~\muxo^i \in \mc M(\mc Y)\quad \forall i \in [N] \\
	    &\ds \sum_{i \in \mc I_1} (p_i \EE_{\muxo^i}[\DD_{\mc Y}(Y, \wh y_i)] + (1- p_i) d_i ) + \sum_{i \in \mc I_2} \big( p_i \EE_{\muxo^i}[\DD_{\mc Y}(Y, \wh y_i)] + p_i d_i \big) \leq N\rho \\
	    & \sum_{i \in \mc I} p_i \geq N\eps.
	    \end{array}
	    \right.
	\end{align*}
	For any $p = (p_i)_{i=1,\ldots, N} \in [0,1]^N$ such that $\sum_{i \in [N]} p_i \geq N\eps$, define the function
	\[
	    f(p) \Let \left\{
	        \begin{array}{cl}
	            \Sup{q \in \R_+^N} & \ds \Sup{\muxo^i: \EE_{\muxo^i}[p_i \DD_{\mc Y}(Y, \wh y_i)] \leq q_i~\forall i}~ \sum_{i \in [N]} \EE_{\muxo^i}[p_i \ell(Y, \alpha, \beta)] \\
	            \st &  \sum_{i \in \mc I_1} (q_i + (1-p_i) d_i) + \sum_{i \in \mc I_2} (q_i + p_i d_i) \le N\rho.
	        \end{array}
	    \right.
	\]
Using the general Charnes-Cooper variable transformation
	\[
	    \upsilon_i = \frac{p_i}{\sum_{i\in [N]} p_i}, \quad t = \frac{1}{\sum_{i \in [N]} p_i},
	\]
	the equivalent characterization~\cite[Lemma~1]{ref:schaible1974parameter} implies that
	\begin{align}
	    \Sup{\QQ \in \mbb B_\rho, \QQ(X \in \mc N_\gamma(x_0)) \geq \eps}~ \EE_{\QQ}[ \ell(Y, \alpha, \beta) | X \in \mc N_\gamma(x_0) ] &= \Sup{p \in [0,1]^N, \sum_{i\in[N]}p_i \geq \eps} \frac{f(p)}{\sum_{i\in[N]} {p_i}} \notag \\
	    &=
	    \left\{
	    \begin{array}{cl}
	    \Sup{\upsilon \in [0, 1]^N, ~t \in \R} & t f(\upsilon/t) \\
	    \st & \upsilon_i \le t \quad \forall i \in [N],~\sum_{i \in [N]} \upsilon_i = 1\\
	    &N^{-1}\leq t\leq (N\eps)^{-1},
	    \end{array}
	    \right. \label{eq:rejoin1}
	\end{align}
	where the fact that $t\geq N^{-1}$ is implied by $v_i\leq t$ and $\sum_{i\in[N]}v_i = 1$ yet makes explicit the fact that $t>0$.

	For any feasible value of $\upsilon$ and $t$, its corresponding objective value can be evaluated as
	\begin{align*}
	    t f(\upsilon / t) &= \left\{
	        \begin{array}{cl}
	            \Sup{q \in \R_+^N} & \ds \Sup{\muxo^i: \EE_{\muxo^i}[\upsilon_i \DD_{\mc Y}(Y, \wh y_i)] \leq q_i t~\forall i}~ \sum_{i \in [N]} \EE_{\muxo^i}[\upsilon_i \ell(Y, \alpha, \beta)] \\
	            \st &  \sum_{i \in \mc I_1} (q_i + (1- \frac{\upsilon_i}{t}) d_i) + \sum_{i \in \mc I_2} (q_i + \frac{\upsilon_i}{t} d_i) \le N\rho
	        \end{array}
	    \right. \\
	    &= \left\{
	        \begin{array}{cl}
	            \Sup{q \in \R_+^N} & \ds \Sup{\muxo^i: \EE_{\muxo^i}[\upsilon_i \DD_{\mc Y}(Y, \wh y_i)] \leq q_i t~\forall i}~ \sum_{i \in [N]} \EE_{\muxo^i}[\upsilon_i \ell(Y, \alpha, \beta)] \\
	            \st & \sum_{i \in \mc I_1} (q_i t + (t- \upsilon_i) d_i) + \sum_{i \in \mc I_2} (q_i t + \upsilon_i d_i) \le N\rho t
	        \end{array}
	    \right. \\
	    &= \left\{
	        \begin{array}{cl}
	            \Sup{\theta \in \R_+^N} & \ds \Sup{\muxo^i: \EE_{\muxo^i}[\upsilon_i \DD_{\mc Y}(Y, \wh y_i)] \leq \theta_i~\forall i}~ \sum_{i \in [N]} \EE_{\muxo^i}[\upsilon_i \ell(Y, \alpha, \beta)] \\
	            \st &  \sum_{i \in \mc I_1} (\theta_i + (t- \upsilon_i) d_i) + \sum_{i \in \mc I_2} (\theta_i + \upsilon_i d_i) \le N\rho t,
	        \end{array}
	    \right. 
	\end{align*}
	where the last equality follows from the change of variable $\theta_i \leftarrow q_i t$. The inner supremum problems over $\muxo^i$ are separable, and can be written using standard conic duality results~\cite[Theorem~1]{ref:blanchet2019quantifying} as
	\begin{align*}
	    \Sup{\muxo^i: \EE_{\muxo^i}[\upsilon_i \DD_{\mc Y}(Y, \wh y_i)] \leq \theta_i~\forall i}~ \sum_{i \in [N]} \EE_{\muxo^i}[\upsilon_i \ell(Y, \alpha, \beta)] = \Inf{\lambda \in \R_+^N} \sum_{i\in[N]} \theta_i \lambda_i + \upsilon_i\Sup{y_i \in \mc Y} \{ \ell(y_i, \alpha, \beta)  - \lambda_i \DD_{\mc Y}(y_i, \wh y_i) \}.
	\end{align*}
	Thus we find for any feasible $(\upsilon, t)$
	\begin{align}
	    t f(\upsilon/t) &= 
	    \left\{
	    \begin{array}{cl}
            \Sup{\theta \in \R_+^N} & \ds \Inf{\lambda \in \R_+^N} \sum_{i\in[N]} \theta_i \lambda_i + \upsilon_i\Sup{y_i \in \mc Y} \{ \ell(y_i, \alpha, \beta)  - \lambda_i \DD_{\mc Y}(y_i, \wh y_i) \} \\
            \st & \sum_{i \in \mc I_1} (\theta_i + (t- \upsilon_i) d_i) + \sum_{i \in \mc I_2} (\theta_i + \upsilon_i d_i) \le N\rho t,
        \end{array}
        \right. \notag \\
        &= 
	    \left\{
	    \begin{array}{ccl}
            \ds \Inf{\lambda \in \R_+^N} &\Sup{\theta \in \R_+^N} &  \sum_{i\in[N]} \theta_i \lambda_i + \upsilon_i \Sup{y_i \in \mc Y} \{ \ell(y_i, \alpha, \beta)  - \lambda_i \DD_{\mc Y}(y_i, \wh y_i) \} \\
            &\st &  \sum_{i \in \mc I_1} (\theta_i + (t- \upsilon_i) d_i) + \sum_{i \in \mc I_2} (\theta_i + \upsilon_i d_i) \le N\rho t
        \end{array}
        \right.\label{eq:rejoin2}
	\end{align}
	where the second equality holds thanks to Sion's minimax theorem because the feasible set for $\theta$ is compact, and the function
	\[
	    (\theta, \lambda) \mapsto \sum_{i\in[N]} \theta_i \lambda_i + \upsilon_i \Sup{y_i \in \mc Y} \{ \ell(y_i, \alpha, \beta)  - \lambda_i \DD_{\mc Y}(y_i, \wh y_i) \}
	\]
	is convex in $\lambda$, affine in $\theta$ and jointly continuous in both $\theta$ and $\lambda$. Let $\mc Q(\upsilon, t)$ be the feasible set for $\theta$, that is,
	\[
	\mc Q(\upsilon, t) = \left\{ \theta \in \R_+^N: \sum_{i \in \mc I_1} (\theta_i + (t- \upsilon_i) d_i) + \sum_{i \in \mc I_2} (\theta_i + \upsilon_i d_i) \le N\rho t \right\}.
	\]
	By rejoining~\eqref{eq:rejoin1}, \eqref{eq:rejoin2} and the definition of $\mc Q(\upsilon, t)$ above, we have
	\begin{align*}
	    &\Sup{\QQ \in \mbb B_\rho, \QQ(X \in \mc N_\gamma(x_0)) \geq \eps}~ \EE_{\QQ}[ \ell(Y, \alpha, \beta) | X \in \mc N_\gamma(x_0) ] \\
	    =& \Sup{(\upsilon, t) \in \Upsilon} \Inf{\lambda \in \R_+^N} \Sup{\theta \in \mc Q(\upsilon, t)}~\sum_{i\in[N]} \theta_i \lambda_i + \upsilon_i \Sup{y_i \in \mc Y} \{ \ell(y_i, \alpha)  - \lambda_i \DD_{\mc Y}(y_i, \wh y_i) \} ,
	\end{align*}
	where the feasible set $\Upsilon$ of $(\upsilon, t)$ in the above problem is understood to be the feasible set of~\eqref{eq:rejoin1}, that is,
	\[
	    \Upsilon \Let \left\{(\upsilon,t)\in [0, 1]^N\times\R : 
	    \upsilon_i \le t \quad \forall i \in [N],~\sum_{i \in [N]} \upsilon_i = 1,~N^{-1}\leq t\leq (N\eps)^{-1} \right\}.
	\]
Since $\Upsilon$ is compact, Sion's minimax theorem~\cite{ref:sion1958minimax} applies and we obtain
\begin{align}
&\Sup{\QQ \in \mbb B_\rho, \QQ(X \in \mc N_\gamma(x_0)) \geq \eps}~ \EE_{\QQ}[ \ell(Y, \alpha, \beta) | X \in \mc N_\gamma(x_0) ] \notag\\
&=  \Inf{\lambda \in \R_+^N} \Sup{(\upsilon, t)\in\Upsilon} \Sup{\theta \in \mc Q(\upsilon, t)}~\sum_{i\in[N]} \theta_i \lambda_i + \upsilon_i \Sup{y_i \in \mc Y} \{ \ell(y_i, \alpha, \beta)  - \lambda_i \DD_{\mc Y}(y_i, \wh y_i) \}\notag\\
&=  \Inf{\lambda \in \R_+^N} \Sup{(\upsilon, t)\in\Upsilon} \Sup{\theta \in \mc Q(\upsilon, t)}~\sum_{i\in[N]} \theta_i \lambda_i + \upsilon_i \inf_{s_i \in \mc S_i(\alpha,\beta,\lambda_i)} s_i \notag\\
&=  \Inf{\lambda \in \R_+^N,s_i\in\mc S_i(\alpha,\beta,\lambda_i)\forall i} \Sup{(\upsilon, t)\in\Upsilon} \Sup{\theta \in \mc Q(\upsilon, t)}~\sum_{i\in[N]} \theta_i \lambda_i + \upsilon_i s_i\label{eq:infsupsupThm44}
\end{align}
where $\mc S_i(\alpha,\beta,\lambda_i) \Let \{s\in\R : s\geq  \ell(y_i, \alpha, \beta)  - \lambda_i \DD_{\mc Y}(y_i, \wh y_i) \,\forall y_i\in\mc Y\}$, and where we exploited the fact that the supremum over $y_i\in\mc Y$ is not affected by the choice of $v$, $t$, and $\theta$, and that each $v_i\geq 0$ to get rid of the apparent bilinearity in $v_i$ and $y_i$. Given that the inner two layers of supremum represent a linear program, linear programming duality can be applied to obtain
\begin{equation}\label{eq:dualvttheta}
        \begin{array}{cll}
            \inf & \phi + (N\eps)^{-1} \nu^+ - N^{-1}\nu^- \\
            \st & \nu^+ \in \R_+,~\nu^-\in\R_+,~\phi \in \R,~\varphi \in \R_+,~\psi\in \R_+^N \\
            & \phi - d_i \varphi + \psi_i   - s_i \ge 0 &\forall i \in \mc I_1 \\
            & \phi + d_i \varphi + \psi_i  - s_i \ge 0 &\forall i \in \mc I_2 \\
            &\nu^+ - \nu^- + (\sum_{i \in \mc I_1} d_i - N\rho) \varphi  -\sum_{i \in [N]} \psi_i  \ge 0 \\
            &\varphi - \lambda_i \ge 0 &\forall i \in [N].
        \end{array}    
\end{equation}

One can confirm that strong linear programming duality necessarily applied given that $\rho>\rho_{\min} (x_0,\gamma,\eps)$ implies the existence of a feasible $\QQ \in \mbb B_\rho$ such that $\QQ(X \in \mc N_\gamma(x_0)) \ge \eps$. This probability measure can be used to create a feasible triplet $(\upsilon,t,\theta)$ for the supremum problem
\begin{align*}
t \Let 1/(N\QQ(X\in\mc N_\gamma(x_0))), \quad v_i \Let \pi_i(X\in \mc N_\gamma(x_0))/(N\QQ(X\in\mc N_\gamma(x_0))), \quad 
\theta_i &\Let 0,
\end{align*}
 with $\{\pi_i\}_{i\in[N]}$ as the set of probability measures that certify that $\QQ\in \mbb B_\rho$. 	Rejoining the infimum operation in \eqref{eq:dualvttheta} into \eqref{eq:infsupsupThm44} and exploiting the definition of the set $\mc V$ complete the proof.
\end{proof} 
	
    \begin{proof}[Proof of Proposition~\ref{prop:markowitz-reform-2}]
        For each $i \in [N]$, by the definition of the loss function~$\ell$ in~\eqref{eq:mean-var-loss} and by the choice of the transport cost $\DD_{\mc Y}$, the constraint 
        \[
            \Sup{y_i \in \mc Y}~ \ell(y_i, \alpha, \beta) - \lambda_i \DD_{\mc Y}(y_i, \wh y_i) \le s_i
        \]
        is equivalent to
        \[
            y_i^\top (\alpha \alpha^\top - \lambda_i I) y_i + \big( 2 \lambda_i \wh y_i - (2\beta + \eta) \alpha \big)^\top y_i + \beta^2 - \lambda_i \| \wh y_i \|_2^2 - s_i \le 0 \quad \forall y_i \in \mc Y.
        \]
        This constraint can be further linearized in $\alpha$ and $\beta$ using the fact that it is equivalent to: 
        \[
            \exists A_i \in \PSD^m, t \in \R_+,\, -A_i\succeq \alpha \alpha^\top - \lambda_i I,~t\geq \beta^2,~-y_i^\top A_i y_i + \big( 2 \lambda_i \wh y_i - (2\beta + \eta) \alpha \big)^\top y_i + t - \lambda_i \| \wh y_i \|_2^2 - s_i \le 0 \quad \forall y_i \in \mc Y.
        \]        
        Because $\mc Y = \R^m$, the set of constraints indexed by $y_i$ can be further written as the semi-definite constraint
        \[
            \begin{bmatrix}
             A_i & (\beta + \eta/2) \alpha - \lambda \wh y_i \\ (\beta + \eta/2) \alpha^\top  - \lambda_i \wh y_i^\top & s_i + \lambda \| \wh y_i\|_2^2 - t 
            \end{bmatrix} \succeq 0
        \]
        while using Schur complement, we also obtain two additional linear matrix inequalities to capture the two other constraints on $A_i$ and $t$ as
        \[\begin{bmatrix}
                \lambda_i I - A_i & \alpha \\ \alpha^\top & 1
            \end{bmatrix} \succeq 0,~\begin{bmatrix}
                t & \beta \\ \beta & 1
            \end{bmatrix} \succeq 0\,.\]
        This completes the proof.
    \end{proof}
    
     \begin{proof}[Proof of Proposition~\ref{prop:mean-cvar-reform-2}]
    
    Notice that $\ell$ is a pointwise maximum of two linear functions of $y_i$. In this case, the epigraphical formulation of the constraint
\begin{align*}
    &\Sup{y_i \in \mc Y} \big\{ \ell(y_i, \alpha, \beta) -  \lambda_i \DD_{\mc Y}(y_i, \wh y_i)  \big\} \le s_i
\end{align*}
    is equivalent to the following set of two semi-infinite constraint
    \begin{align*}
    \left\{
        \begin{array}{ll}
   - \eta y_i^\top \alpha + \beta
    -\lambda_i \|y_i - \wh y_i\|_2^2  \le s_i &\forall y_i \in \R^m\\
    - (\eta + \frac{1}{\tau}) y_i^\top \alpha + (1 - \frac{1}{\tau})\beta
    -\lambda_i \|y_i - \wh y_i\|_2^2  \le s_i &\forall y_i \in \R^m.
    \end{array}
    \right.
\end{align*}
    The proof now follows from the same line of argument as the proof of Proposition~\ref{prop:markowitz-reform-2}.
    \end{proof}
%%%%%%%%%%%%%%%%%%%%%%%%%%%%%

\section{Semi-infinite Program Reformulations under Compactness} \label{appendix:refor-compact}

    In this section, we provide the complementary reformulations for the case where $\mc Y$ is a compact ellipsoid with a non-empty interior. Due to space constraints, we focus on the most general setup of Section~\ref{sec:gammapositive_epspositive}. Here, the set $\mc V$ is defined as in~\eqref{eq:dual-set}.
    
	\begin{proposition}[Mean-variance loss function]\label{prop:markowitz-reform-compact} 
    Suppose that $\ell$ is the mean-variance loss function of the form~\eqref{eq:mean-var-loss}, $\gamma \in \R_{++}$, $\eps \in (0, 1]$ and $\rho > \rho_{\min}(x_0,\gamma,\eps)$. Suppose in addition that $\DD_{\mc Y}(y,\wh y) = \|y-\wh y\|_2^2$, $\mc Y  = \{ y \in \R^m: y^\top Q y + 2 q^\top y + q_0 \le 0\}$ for some symmetric matrix $Q$ and $\mc Y$ has non-empty interior. Let the parameters $d_i$ be defined as in~\eqref{eq:d-def}. The distributionally robust portfolio allocation model with side information~\eqref{eq:optimization-eps-0} is equivalent to the semi-definite optimization problem
    \begin{equation*}
        \begin{array}{cll} 
            \min & \phi + (N\eps)^{-1} \nu^+ - N^{-1}\nu^- \\
            \st & \alpha \in \mc A,~\beta \in \mc B,~t\in\R_+,~A_i \in \PSD^m\quad \forall i \in [N],~(\lambda, s, \nu^+, \nu^-, \phi, \varphi, \psi) \in \mc V,~\omega \in \R_+^N \\
            &\begin{bmatrix}
                \lambda_i I - A_i & \alpha \\ \alpha^\top & 1
            \end{bmatrix} \succeq 0,~\begin{bmatrix}
                A_i + \omega_i Q & (\beta + \eta/2) \alpha - \lambda_i \wh y_i + \omega_i q \\ (\beta + \eta/2) \alpha^\top  - \lambda_i \wh y_i^\top + \omega_i q & s_i + \lambda_i \| \wh y_i\|_2^2 - t + \omega_i q_0
            \end{bmatrix} \succeq 0 & \forall i \in [N] \\
            & \begin{bmatrix}
                t & \beta \\ \beta & 1
            \end{bmatrix} \succeq 0.
        \end{array}
    \end{equation*}
    \end{proposition}
    \begin{proof}[Proof of Proposition~\ref{prop:markowitz-reform-compact}]
        The proof follows almost verbatim from that of Proposition~\ref{prop:markowitz-reform-2}. In the penultimate step, the constraint
        \[
        -y_i^\top A_i y_i + \big( 2 \lambda_i \wh y_i - (2\beta + \eta) \alpha \big)^\top y_i + t - \lambda_i \| \wh y_i \|_2^2 - s_i \le 0 \quad \forall y_i \in \mc Y
        \]        
         can be further written as 
        \[
            \exists \omega_i \in \R_+: ~\begin{bmatrix}
             A_i  & (\beta + \eta/2) \alpha - \lambda \wh y_i \\ (\beta + \eta/2) \alpha^\top  - \lambda_i \wh y_i^\top & s_i + \lambda \| \wh y_i\|_2^2 - t 
            \end{bmatrix} + \omega_i \begin{bmatrix} Q & q \\ q^\top & q_0 \end{bmatrix} \succeq 0
        \]
        thanks to the S-lemma. This completes the proof.
    \end{proof}
    
    \begin{proposition}[Mean-CVaR loss function]\label{prop:mean-cvar-reform-compact} 
    Suppose that $\ell$ is the mean-CVaR loss function of the form~\eqref{eq:mean-cvar-loss}, $\gamma \in \R_{++}$, $\eps \in (0, 1]$ and $\rho > \rho_{\min}(x_0, \gamma, \eps)$. Suppose in addition that $\DD_{\mc Y}(y,\wh y) = \|y-\wh y\|_2^2$, $\mc Y  = \{ y \in \R^m: y^\top Q y + 2 q^\top y + q_0 \le 0\}$ for some symmetric matrix $Q$ and $\mc Y$ has non-empty interior. Let the parameters $d_i$ be defined as in~\eqref{eq:d-def}. The distributionally robust portfolio allocation model with side information~\eqref{eq:optimization-eps-0} is equivalent to the semi-definite optimization problem
    \begin{equation*}
        \begin{array}{cll} 
            \min & \phi + (N\eps)^{-1} \nu^+ - N^{-1}\nu^- \\
            \st & \alpha \in \mc A,~\beta \in \mc B,~\omega_1 \in \R_+^N,~\omega_2 \in \R_+^N,~(\lambda, s, \nu^+, \nu^-, \phi, \varphi, \psi) \in \mc V \\
            &\begin{bmatrix}
                \lambda_i I + \omega_{1i} Q & \frac{\eta}{2} \alpha - \lambda_i \wh y_i + \omega_{1i} q\\ \frac{\eta}{2} \alpha^\top  - \lambda_i \wh y_i^\top + \omega_{1i} q^\top & s_i + \lambda_i \| \wh y_i\|_2^2 - \beta + \omega_{1i} q_0
            \end{bmatrix} \succeq 0 & \forall i \in [N] \\ 
            &
            \begin{bmatrix}
                \lambda_i I + \omega_{2i} Q& \frac{\eta \tau + 1}{2\tau}\alpha - \lambda_i \wh y_i + \omega_{2i} q\\ \frac{\tau + 1}{2\tau} \alpha^\top  - \lambda_i \wh y_i^\top + \omega_{2i} q^\top & s_i + \lambda_i \| \wh y_i\|_2^2 - (1-1/\tau)\beta + \omega_{2i} q_0
            \end{bmatrix} \succeq 0 & \forall i \in [N].
        \end{array}
    \end{equation*}
    \end{proposition}
    \begin{proof}[Proof of Proposition~\ref{prop:mean-cvar-reform-compact}]
        The proof follows almost verbatim from that of Proposition~\ref{prop:mean-cvar-reform-2}, with the last step involves rewriting the semi-infinite constraints
        \begin{align*}
    \left\{
        \begin{array}{ll}
   - \eta y_i^\top \alpha + \beta
    -\lambda_i \|y_i - \wh y_i\|_2^2  \le s_i &\forall y_i \in \mc Y \\
    - (\eta + \frac{1}{\tau}) y_i^\top \alpha + (1 - \frac{1}{\tau})\beta
    -\lambda_i \|y_i - \wh y_i\|_2^2  \le s_i &\forall y_i \in \mc Y
    \end{array}
    \right.
\end{align*}
using the S-lemma.
    \end{proof}
    
    Notice that the conditions in Propositions~\ref{prop:markowitz-reform-compact} and \ref{prop:mean-cvar-reform-compact} imply that Lemmas~\ref{lemma:MV} and \ref{lemma:MCVaR} also hold. The reformulations in Propositions~\ref{prop:markowitz-reform-compact} and \ref{prop:mean-cvar-reform-compact} are thus exact for the robustified conditional mean-variance and mean-CVaR portfolio allocation problems, respectively.

\section{Portfolio Allocation with type-$\infty$ Optimal Transport Cost Ambiguity Set}
\label{appendix:type-infty}

In this appendix, we elaborate the reformulations for the conditional portfolio allocation problems using the type-$\infty$ Wasserstein ambiguity set.
We first revisit the definition of the type-$\infty$ optimal transport distance.

\begin{definition}[Type-$\infty$ optimal transport cost]
		Let $\DD$ be a nonnegative and continuous function on $\Xi \times \Xi$. The type-$\infty$ optimal transport cost between two distributions $\QQ_1$ and $\QQ_2$ supported on $\Xi$ is defined as
		\[	\Wass_{\infty}(\QQ_1, \QQ_2) \Let \inf \left\{ \mathrm{ess} \Sup{\pi} \big\{ \DD(\xi_1, \xi_2) : (\xi_1, \xi_2)  \in \Xi \times \Xi \big\} :
		\pi \in \Pi(\QQ_1, \QQ_2)
		\right\}.
		\]
	\end{definition}

    The type-$\infty$ ambiguity set can be formally defined as
    \[
    \mbb B^\infty_\rho= \left\{ \QQ \in \mc M(\mc X \times \mc Y): \Wass_\infty(\QQ, \Pnom) \leq \rho  \right\}.
    \]
    
    We now provide the reformulation for the mean-variance portfolio allocation problem. To this end, recall that the parameters $\kappa$ are defined as in~\eqref{eq:kappa-def}. In addition, define the following set
    \[
	\mc J \Let \left\{ i \in [N] : \DD_{\mc X}(x_0, \wh x_i) \le \rho + \gamma \right\},
	\]
	and $\mc J$ is decomposed further into two disjoint subsets
	\[
	    \mc J_1 = \left\{ i \in \mc J: \DD_{\mc X}(x_0, \wh x_i) + \rho \le \gamma \right\} ~~ \text{and} ~~ \mc J_2 = \mc J \backslash \mc J_1.
	\]
    
    \begin{proposition}[Mean-variance loss function] \label{prop:infinity}Suppose that $\ell$ is the mean-variance loss function of the form~\eqref{eq:mean-var-loss}, $\gamma \in \R_+$. Suppose in addition that $\mc X = \R^n, \mc Y  = \R^m$, $\DD_{\mc X}(x,\wh x) = \|x-\wh x\|^2$, $\DD_{\mc Y}(y,\wh y) = \|y-\wh y\|_2^2$ and $\rho > \min_{i \in [N]}~\kappa_i$. The distributionally robust portfolio allocation model with side information
    \[
    \Min{\alpha \in \mc A,~\beta\in \R} \Sup{\QQ \in \mbb B^\infty_\rho, \QQ(X \in \mc N_\gamma(x_0)) > 0} ~\EE_{\QQ} [ \ell(Y, \alpha, \beta) | X \in \mc N_{\gamma}(x_0)]
    \]
    is equivalent to the second-order cone program 
    \[
    \begin{array}{cll}
        \min & \lambda \\
        \st & \alpha\in \mc A,~\beta \in \R,~\lambda \in \R,~u_i \in \R~\forall i \in \mathcal J_1,~u_i \in \R_+~\forall i \in \mathcal J_2,~t \in \R^N,~z \in \R_+^N \\
         & \sum_{i \in \mathcal J} u_i \le 0 \\
        & 
        \!\!\!\!\left. \begin{array}{l}
        \left\Vert\begin{bmatrix}
            2 z_i\\
            1-\lambda - u_i - \eta \beta - \frac{1}{4} \eta^2
        \end{bmatrix}\right\Vert_2
        \leq 1+\lambda + u_i + \eta \beta + \frac{1}{4} \eta^2 
        \\
         \wh y_i^\top \alpha - \beta - 0.5\eta \le t_i,~-\wh y_i^\top \alpha + \beta + \frac{1}{2}\eta \le t_i  \\
        t_i + (\rho -  \DD_{\mc X}(\wh x_i^p, \wh x_i))^{1/2} \|\alpha\|_2 \le z_i
        \end{array} \right\} \quad \forall i \in \mc J.
    \end{array}
    \]
    \end{proposition}
    \begin{proof}[Proof of Proposition \ref{prop:infinity}]
    Using \cite[Theorem~2.3]{ref:nguyen2020distributionally}, we have
    \[
    \begin{array}{cll}
        \min & \lambda \\
        \st & \alpha\in \mc A,~\beta \in \R,~\lambda \in \R,~u_i \in \R~\forall i \in \mathcal J_1,~u_i \in \R_+~\forall i \in \mathcal J_2 \\
        & \lambda + u_i \ge v_i\opt(\alpha,\beta) \quad \forall i \in \mathcal J \\
        & \sum_{i \in \mathcal J} u_i \le 0,
    \end{array}
    \]
    where for each $i \in \mc J$, the value $v_i\opt(\alpha, \beta)$ is
    \begin{align*}
        v_i\opt(\alpha, \beta) &= \sup \left\{ \ell(y_i, \alpha, \beta) : y_i \in \mc Y,~ \| y_i-\wh y_i\|_2^2 \leq \rho -  \DD_{\mc X}(\wh x_i^p, \wh x_i) \right\}\\
        &= \Big(|\wh y_i^\top \alpha - \beta - \frac{1}{2}\eta| + \|\alpha\|_2(\rho -  \DD_{\mc X}(\wh x_i^p, \wh x_i))^{1/2}\Big)^2 -\eta\beta -\frac{1}{4} \eta^2.
    \end{align*}
    Formulating each constraint $\lambda + u_i \ge v_i\opt(\alpha,\beta)$ using an epigraphical formulation of the form
    \[
        |\wh y_i^\top \alpha - \beta - \frac{1}{2}\eta| \le t_i,~t_i+ \|\alpha\|_2(\rho -  \DD_{\mc X}(\wh x_i^p, \wh x_i))^{1/2} \le z_i,~v_i\opt(\alpha, \beta) \ge z_i^2 - \eta \beta - \frac{1}{4} \eta^2
    \]
    and formulating them as linear or second-order cone constraints completes the proof. 
   \end{proof}
%%%%%%%%%%%%%%%%%%%%%%%

\section{\modified{Study of Computational  Time of Solutions}} \label{sec:mean-cvar}
\modified{
We provide the average CPU solution time (in seconds) for solving the portfolio allocation problems with 50 and 399 assets, based on the $252\times 2$ instances used in the out-of-sample tests. All experiments are conducted on the Cedar cluster of Compute Canada with a single CPU core and 4G memory.

From Table~\ref{tab:time_50assets}, we observe that 
while OTCMV and OTCMC have a larger solution time than the competing baselines, their average solution times for one instance are less than 4 seconds for 50 assets and 18 seconds for 399 assets. The relative difference in the computational time also seems stable as we increase the number of assets. This confirms that our models are tractable and can handle the size of problems encountered in practice.
\begin{table}[!htp]
\centering
\caption{The average CPU solution time (in seconds) for portfolio problems with 50 and 399 assets.}\label{tab:time_50assets}
\begin{threeparttable}
\begin{tabular}{clllllllllll}
\hline
No. assets           & $\eta$ & MV   & DRMV & CMV  & DRCMV & OTCMV & MC   & DRMC & CMC  & DRCMC & OTCMC \\\hline
\multirow{5}{*}{50}  & 1      & 0.46 & 0.47 & 0.47 & 0.60  & 3.42  & 0.47 & 0.49 & 0.42 & 0.47  & 1.00  \\
                     & 3      & 0.49 & 0.50 & 0.49 & 0.63  & 3.20  & 0.47 & 0.48 & 0.41 & 0.46  & 1.06  \\
                     & 5      & 0.49 & 0.50 & 0.50 & 0.65  & 3.21  & 0.49 & 0.49 & 0.41 & 0.47  & 1.12  \\
                     & 7      & 0.46 & 0.47 & 0.46 & 0.63  & 3.45  & 0.48 & 0.50 & 0.42 & 0.48  & 1.11  \\
                     & 9      & 0.39 & 0.40 & 0.39 & 0.55  & 3.19  & 0.48 & 0.49 & 0.42 & 0.47  & 1.12  \\\hline
\multirow{5}{*}{399} & 1      & 3.70 & 3.29 & 2.80 & 3.61  & 17.09 & 3.98 & 3.94 & 2.92 & 3.80  & 8.47  \\
                     & 3      & 3.74 & 3.51 & 2.97 & 3.79  & 17.60 & 4.00 & 4.01 & 2.94 & 3.79  & 8.90  \\
                     & 5      & 3.55 & 3.28 & 2.74 & 4.20  & 16.98 & 3.96 & 3.96 & 2.93 & 3.78  & 9.12  \\
                     & 7      & 3.57 & 3.35 & 2.78 & 4.30  & 17.75 & 3.96 & 3.96 & 2.93 & 3.79  & 9.11  \\
                     & 9      & 3.54 & 3.35 & 2.76 & 4.36  & 17.55 & 4.47 & 4.25 & 3.19 & 4.05  & 9.44 \\\hline
\end{tabular}
\end{threeparttable}
\end{table}
}

\section{Implications for Contextual Two-stage Stochastic Linear programs}\label{sec:contextTSLP}

In this appendix, we explore how our distributionally robust conditional decision-making framework can be exploited in a two-stage stochastic linear optimization problem. 
Let us redefine the loss function as $\ell(Y, \alpha, \beta)\Let 
    r(h(Y,\alpha),\beta) - \eta\cdot h(Y,\alpha)$,
where $h(Y,\alpha)$ captures the optimal profit generated by a linear recourse problem with right-hand side uncertainty:
\begin{align*}
h(y,\alpha) \Let \left\{
\begin{array}{cl}
\max\;& c^\top \upsilon\\
\st &\upsilon \in \R^m,~A\alpha+B\upsilon\leq Cy,
\end{array} \right.
\end{align*}
constrained by $K$ linear constraints prescribed by the matrices $A$, $B$ and $C$ of appropriate dimensions. 
This formulation can capture many interesting operational management problems, including multi-item newsvendor problems~\cite{doi:10.1287}, facility location problems~\cite{SAIF2021995}, and scheduling problems~\cite{denton:ORUU2010}.

In the case that $r(\cdot,\beta)$ is non-increasing for all $\beta$, such as for the CVaR formulation, one can obtain a conservative approximation for problem \eqref{eq:stochastic-opt} by employing linear decision rules. Namely,
\[\EE_{\PP}[ \ell(Y, \alpha, \beta)]=\EE_{\PP}[r(h(Y,\alpha),\beta) - \eta\cdot h(Y,\alpha)]\leq \EE_{\PP}[r(c^\top(\upsilon+{\Upsilon}Y),\beta) - \eta\cdot c^\top(\upsilon+{\Upsilon}Y)]\]
as long as
\begin{equation}
    A\alpha+B(\upsilon+{\Upsilon}y)\leq Cy \quad \forall y\in \mc Y  \label{appF:LDR:constraint}
\end{equation}
since in this case $h(y,\alpha)\geq c^\top(\upsilon+{\Upsilon}Y)$.
% \[
%     \begin{array}{cl}
%     \Min{\alpha \in \mc A,~\beta \in \mc B,~\upsilon,~{\Upsilon}}~&\EE_{\PP}[ \hat{\ell}(Y, \beta,\upsilon,{\Upsilon})]\\
%     \mbox{s.t.}\;&A\alpha+B(\upsilon+{\Upsilon}y)\leq Cy ,\;\forall y\in \mc Y,
%     \end{array}
% \]
% with $\hat{\ell}(Y,  \beta,\upsilon,{\Upsilon})\Let \
%     r(c^\top(\upsilon+{\Upsilon}Y),\beta) - \eta\cdot c^\top(\upsilon+{\Upsilon}Y)$.

%    
It is therefore not surprising that Propositions~\ref{prop:mean-cvar-reform}
and~\ref{prop:mean-cvar-reform-2} have natural extensions to provide conservative approximation models in the form of second-order cone and semi-definite programs. For example, Proposition~\ref{prop:mean-cvar-reform} can be extended to the following.

\begin{corollary}[Mean-CVaR TSLP for single fiber set]\label{corol:tslp} 
Suppose that $\ell$ is the mean-CVaR loss function of the form $\ell(y, \alpha, \beta)=-\eta h(y,\alpha) + \beta+ \frac{1}{\tau} (-h(y,\alpha) - \beta)^+$, $\gamma = 0$, $\eps \in (0, 1]$ and $\rho > \rho_{\min}(x_0, 0, \eps)$. Suppose in addition that $\mc X = \R^n$, $\mc Y $ is polyhedral, $\DD_{\mc X}(x,\wh x) = \|x-\wh x\|^2$ and $\DD_{\mc Y}(y,\wh y) = \|y-\wh y\|_2^2$. The distributionally robust two-stage stochastic linear program with side information is conservatively approximated by the second-order cone program 
\be \label{eq:appF:CVaRTSLP:reform}
\begin{array}{cll}
\inf & \ds \rho \lambda_1 + \eps \lambda_2 -(1+\eta) c^\top\upsilon+ \frac{1}{N}\sum_{i \in [N]} \theta_i \\
\st & \alpha\in \mc A,\;\beta\in\mc B,\; \lambda_1 \in \R_+,\; \lambda_2 \in \R,\; \theta \in \R_+^N,\; z\in \R_+^N,\; \tilde z\in \R_+^N, \; \upsilon\in\R^m, \; {\Upsilon}\in\R^{m\times m}\\
&\left. 
\!\!\!
\begin{array}{l}
z_i = \theta_i +\lambda_1 \|x_0 - \wh x_i\|^2+ \lambda_2 + \eps^{-1}\eta \wh y_i^\top {\Upsilon}^\top c - \eps^{-1}\beta \\
\tilde{z}_i = \theta_i +\lambda_1 \|x_0 - \wh x_i\|^2+ \lambda_2 + \eps^{-1}(\eta + \frac{1}{\tau}) \wh y_i^\top {\Upsilon}^\top c - \eps^{-1}(1 - \frac{1}{\tau})\beta\\
\left\|
\begin{bmatrix}
\eps^{-1}\eta {\Upsilon}^\top c\\
z_i-\lambda_1
\end{bmatrix}\right\|_2
\leq z_i+\lambda_1, \quad \left\|
\begin{bmatrix}
\eps^{-1}(\eta + \tau^{-1}) {\Upsilon}^\top c\\
\tilde z_i-\lambda_1
\end{bmatrix}\right\|_2
\leq \tilde z_i+\lambda_1
\end{array} \right\} \forall i \in [N]\\
& e_k^\top(A\alpha+B\upsilon) + \delta^*((\Upsilon^\top B^\top - C^\top) e_k|\mathcal{Y})\leq 0 \qquad\forall k\in[K],
\end{array}
\ee
where $e_k$ is the $k$-th column of the identity matrix, and $\delta^*(z|\mathcal{Y})\Let\sup_{y\in\mathcal{Y}} z^\top y$ is a linear programming representable support function of $\mathcal{Y}$.
\end{corollary}

\begin{proof}[Proof of Corollary~\ref{corol:tslp}]
We start with the following derivations
\begin{align*}
\min_\beta &\Sup{\QQ \in \mbb B_\rho, \QQ(X = x_0) \ge \eps}~ \EE_{\QQ}[ \ell(Y, \alpha, \beta) | X = x_0 ] \\
&\leq \min_\beta  \Sup{\QQ \in \mbb B_\rho, \QQ(X = x_0) \ge \eps}~ \EE_{\QQ}[r(c^\top(\upsilon+{\Upsilon}Y),\beta) - \eta\cdot c^\top(\upsilon+{\Upsilon}Y) | X = x_0 ]\\
&= \min_\beta  \Sup{\QQ \in \mbb B_\rho, \QQ(X = x_0) \ge \eps}~ \EE_{\QQ}[ -\eta c^\top(\upsilon+{\Upsilon}Y) + \beta+ \frac{1}{\tau} (-c^\top(\upsilon+{\Upsilon}Y) - \beta)^+ | X = x_0 ]\\
&= -(1+\eta) c^\top\upsilon + \min_{\beta'}  \Sup{\QQ \in \mbb B_\rho, \QQ(X = x_0) \ge \eps}~ \EE_{\QQ}[ -\eta c^\top{\Upsilon}Y + \beta'+ \frac{1}{\tau} (-c^\top{\Upsilon}Y - \beta')^+ | X = x_0 ]
\end{align*}
where we assume that the pair  $(\upsilon, \Upsilon)$ satisfies the constraint \eqref{appF:LDR:constraint}, 
%\[A\alpha+B(\upsilon+{\Upsilon}y)\leq Cy ,\;\forall y\in \mc Y,\]
and where we perform the replacement $\beta'\leftarrow\beta+c^\top\upsilon$. One obtains problem \eqref{eq:appF:CVaRTSLP:reform} after relaxing the support for $\mathbb{Q}$ to $\mc X \times \R^m$ in order to employ the reformulation in Proposition \ref{prop:mean-cvar-reform}, and finally reformulating each robust constraint indexed by $k$ as:
\[
e_k^\top(A\alpha+B\upsilon+(B{\Upsilon}-C)y)\leq 0 \qquad \forall y\in \mc Y
\]
using the definition of the support function. We refer the interested readers to~\cite{ref:ben2015deriving} for examples of linear programming representations of support functions of polyhedral sets.
\end{proof}

\bibliographystyle{siam}
\bibliography{bibliography.bbl}

\end{document}